\documentclass[10pt]{article}

\usepackage[utf8]{inputenc}
\usepackage[T1]{fontenc}
\usepackage{url}
\usepackage[colorlinks=true,allcolors=blue]{hyperref}
\usepackage{empheq}
\usepackage{latexsym,url}
\usepackage{marvosym}
\usepackage{enumitem}
\usepackage{tabularx}

\usepackage{amsmath}
\usepackage{amssymb}
\usepackage{amsthm}
\usepackage{mathtools}
\usepackage{color}
\usepackage{dsfont}

\usepackage{graphicx}
\usepackage{enumerate}

\usepackage{geometry}
\usepackage{lmodern,textcomp} 
\usepackage{microtype} 
\usepackage{eurosym}

\usepackage{caption}
\usepackage{subcaption} 
\usepackage{array}
\usepackage{multirow}
\usepackage{etoolbox}

\usepackage{float}
\usepackage[section]{placeins}

\usepackage{breqn} 
\usepackage{booktabs}

\usepackage{cite}

\usepackage{natbib}
\DeclareSymbolFont{calletters}{OMS}{cmsy}{m}{n}
\DeclareSymbolFontAlphabet{\mathcal}{calletters}

\def\be{\begin{eqnarray}}
\def\ee{\end{eqnarray}}

\def\b*{\begin{eqnarray*}}
	\def\e*{\end{eqnarray*}}

\newcommand{\ba}{\[\begin{aligned}}
\newcommand{\ea}{\end{aligned}\]}
\newtheorem{Theorem}{Theorem}[section]
\newtheorem{Definition}[Theorem]{Definition}
\newtheorem{Proposition}[Theorem]{Proposition}

\newtheorem{Lemma}[Theorem]{Lemma}
\newtheorem{Corollary}[Theorem]{Corollary}
\newtheorem{Remark}[Theorem]{Remark}

\makeatletter \@addtoreset{equation}{section}

\newcommand{\brak}[1]{\left(#1\right)}    
\newcommand{\crl}[1]{\left\{#1\right\}}   
\newcommand{\edg}[1]{\left[#1\right]}     
\newcommand{\abs}[1]{\left|#1\right|}

\definecolor{orange}{rgb}{1,0.5,0}
\definecolor{violet}{rgb}{0.6,0.4,0.8}

\def \E{\mathbb{E}}
\def \F{\mathbb{F}}

\def \P{\mathbb{P}}

\def \R{\mathbb{R}}

\def \M{\mathbb{M}}
\def \N{\mathbb{N}}

\def \Gb{\mathbb{G}}

\def\Cc{{\cal C}}

\def\Ec{{\cal E}}
\def\Fc{{\cal F}}

\def\Jc{{\cal J}}

\def\Mc{{\cal M}}
\def\Nc{{\cal N}}

\def\Pc{{\cal P}}

\def\Rc{{\cal R}}

\def\Tc{{\cal T}}
\def\Uc{{\cal U}}
\def\Vc{{\cal V}}

\def\Zc{{\cal Z}}

\def \a{\alpha}

\def \0{\mathbf{0}}

\newcommand{\rmi}{{\rm (i)$\>\>$}}
\newcommand{\rmii}{{\rm (ii)$\>\>$}}
\newcommand{\rmiii}{{\rm (iii)$\>\,    \,$}}

\newcommand{\rmv}{{\rm (v)$\>\>$}}
\newcommand{\rmvi}{{\rm (vi)$\>\>$}}

\def\x{\times}

\def\1{{\bf 1}}

\def \no{\noindent}

\def \Frac{\displaystyle\frac}

\def \Sup{\displaystyle\sup}

\newcommand{\ie}{i.e., }

\def \etaa{{\eta}_{A}}
\def \etap{{\eta}_{P}}

\def \diff{\mbox{d}}
\DeclareMathOperator*{\esssup}{ess\,sup}
 
\title{ A Principal--Agent approach to Capacity Remuneration Mechanisms }

\author{Cl\'emence Alasseur\thanks{clemence.alasseur@edf.fr}, \ \ Heythem Farhat\thanks{heythem.farhat@polytechnique.edu}\ \  and\ \ Marcelo Saguan\thanks{marcelo.saguan@edf.fr}}

\date{\today}

\begin{document}
	
\sloppy 
	
\maketitle
	
\abstract{We propose to study electricity capacity remuneration mechanism design through a principal-agent approach. The
principal represents the aggregation of electricity consumers (or a representative entity), subject to the physical risk of shortage, and the agent represents the electricity capacity owners, who invest in capacity and produce electricity to satisfy consumers' demand, and are subject to financial risks. Following the methodology of Cvitani{\'c} et al. (2017), we propose an optimal contract, from consumers' perspective,  which complements the revenue capacity owners achieved from the spot energy market, and incentivizes both parties to perform an optimal level of investment while sharing the physical and financial risks. Numerical results provide insights on the necessity of a capacity
	remuneration mechanism and also show how this is especially true when the level of uncertainties on demand or production side increases.
		
\vspace{0.5cm}
		
{\bf Key words.}  Capacities market, capacity remuneration mechanism, principal-agent problem, contract theory.
}

\section{Introduction}
	Electricity market is characterized by the constraint that production must be equal to the consumption at any time. In case of non respect of this constraint, the system can incur a power outage whose consequences might be highly problematic. For example, the total economic cost of the August 2003 blackout in the USA was estimated to be between seven and ten billion dollars (\cite{council2004economic}). This blackout resulted in the loss of around 62 GW of electric load that served more than 50 million people at the USA-Canada border. Besides, it took 2 days for major affected areas to have the power restored, while some regions had to wait up to a full week.\\
	
As electricity can hardly be stored; hydro storage is limited in size, and developing a large fleet of batteries is still highly costly, the power production capacity must be high enough to cope with major peak load events which can reach extreme levels compared to the average load. In France for example the average load was around 55GW in 2017\footnote{See the French TSO website \url{http://bilan-electrique-2017.rte-france.com}.}, whereas the peak of electricity consumption record was above 100 GW in February 2012. Indeed, electrical load is characterized by a high variability implied by meteorological variations and economic conditions on different time scales. Again, in France for example, the difference between peak load in 2012 and 2014 is around 20 GW, which corresponds to an equivalent capacity of around 40 combined cycle gas turbines (CCGT) of 500 MW\footnote{See \url{https://www.ecologique-solidaire.gouv.fr/securite-dapprovisionnement-en-electricite}.}.
To ensure security of supply, most electricity systems specify the \emph{Loss of Load Expectation} (LoLE) which is a reliability target for the electricity system, and has been fixed in some countries at three hours at most per year \cite{Newberry16}.\\
	
The consequence of such constraints on the production system is that some power plants are used rarely (only during extreme peak load) but remain necessary for the system security, and insuring their economical viability with energy markets only is not guaranteed. This question has already motivated a great amount of economic literature under the name of ``missing money'' \cite{Joskow07}. This lack of revenue can occur because of energy market imperfections such as price caps or out-of-market actions made by the transmission system operator as well as reliability targets going beyond reliability outcome provided by the market.\\
	
The ``missing money'' issue might be further increased when the share of renewable energies increases in the system \cite{Newberry16}. Indeed, renewable energies have low variable costs so their introduction has made electricity prices lower. See  \cite{brown2018capacity} or \cite{Levin2015Capacity} for a  proof that subsidized renewable capacity pushes downward energy prices. This could result in the withdrawal of most expensive power plants, jeopardizing the security of electricity system and the lack of incentives to invest in new capacities. \\
	
In addition, even without the ``missing money'' problem, electricity markets are highly volatile for the reasons already stated (i.e. low level of storage, high uncertainties on load and production levels induced by outages and meteorological conditions impacting the production capacities) and suffer from ``missing markets'' issues \cite{Newberry16} such as the horizon shortness of contracts proposed by electricity markets compared to the lifetime of power-plants. For all these reasons, the financial risk is particularly high for investors in electricity capacities and may lead to high hurdle rates, see \cite{hobbs2007Capacity} for a model which studies how capacity markets variations can lower capital costs for generators by reducing risks).\\
	
For all the reasons cited above, several regions of the world have decided to put in place a capacity remuneration mechanism (CRM), in addition to energy markets. This kind of markets aims at insuring a payment for electricity generating assets for the capacities they provide, regardless of their actual production. This market can be thought of as a payment for an insurance provided by the power plant against shortage and blackout risk. However, no consensus on the design of such CRM arose so far, see for example \cite{bublitz2019survey} for a review of theoretical studies and implementations of CRM or \cite{Bhagwat17} for a survey of different capacity markets implemented in the USA.\\ 
	
In \cite{scouflaire18}, the author argues that CRM do improve security of supply, in exchange of a significant impact on consumer's bill in the USA, as opposed (surprisingly) to the EU where the impact on the end users price is not significant. An analysis of the impact of capacity on welfare under a price-capped electricity market is made on the Texas market (ERCOT) in \cite{bajo2017welfare} showing that capacity markets have several effects: an increase of the wholesale electricity price and reliability and  a reduction of price volatility. Several mechanisms and their corresponding conditions for achieving efficiency are studied in \cite{Leautier16}. In \cite{briggs2013resource}, distortion of capacity markets implied by subsidies of base load capacities are pointed out and correction mechanisms such as the minimal offer price rule (MOPR) tested in the PJM markets are studied. Currently in Europe, several designs of CRM have been adopted such as a market capacity for example in the United Kingdom, France, Italy, Belgium and Ireland\footnote{See  \url{https://www.sem-o.com/markets/capacity-market-overview} and \url{https://ec.europa.eu/commission/presscorner/detail/en/IP_17_4944}.}(under construction); a capacity payment as in Spain or Portugal; strategic reserves in Sweden or in Germany.\\
	
In the literature, several papers study different CRM designs with distinct modeling approaches. For example in \cite{HaryRiousSagan16}, the authors compare the benefits of capacity markets or strategic reserves versus energy-only design in terms of security of supply, investment and generation costs in a dynamic model of investing. This is the same approach developed by \cite{hach2016capacity} and applied to the UK market. \cite{hoschle2017electricity} analyse the impact of capacity mechanism on energy markets and on the remuneration of flexibility and emission-neutral renewable capacities. In \cite{Bhagwat17}, the authors implement the UK capacity markets in an Agent-based model where Agents have a limited vision of the future. In \cite{hermon2007asymmetric}, the authors model two CRMs--in particular under information asymmetry--using agency theory. They model capacity payments as a menu of contracts and strategic reserves as a retention rule of a bilateral contract between the TSO and a producer and then compare these CRMs. The information asymmetry is mainly on the ``type'' of the generator, namely its access to the capital market which impacts its efficiency.\\	
	
In this work, we propose a principal-agent framework to shed light on the design of CRM in a context of information asymmetry and external uncertainties (in production and demand). Using the recent developments of contract theory \cite{cvitanicpt}, we model and solve the problem in a continuous time setting, which allows us to dive deeper into the incentive mechanism, and provide a recommended policy for investment in electric power plants, with an optimal dynamic capacity payment allowing for an efficient (financial and physical) risk sharing between consumers and producers.\\
		
In the scope of the paper, producers and retailers are fully separated and exchange electricity through spot markets, and the electricity demand is considered to be inelastic. The relationship between consumers and producers is modeled by a principal-agent problem, with the principal being the aggregation of power consumers (or an equivalent entity representing them), and the agent the collection of producers. 
Note that the transmission system operator (TSO) which operates in real-time in many electricity markets could be considered as a representative entity of the aggregation of consumers.\\
	
Our model accounts for the information asymmetry on the actions of the agent --moral hazard--, \ie consumers do not observe producers' actions but only the results of their actions. In fact, consumers want to incentivize producers in an optimal way to provide electricity when needed, but at the same time they have no information on the commitment of the latter (producers) to build or maintain power plants, as they only observe the volume of electricity produced, not the effort of capacity owners to install new power plants or to keep the existing ones in good operation conditions. 
The model developed enables us to specify an optimal contract which incentivizes the producers to make the right level of investment to achieve a certain level of security for the system.\\
		
This proposed contract remunerates the capacity owners depending on realized uncertainties on the demand and available capacities while sharing the financial risks between consumers and producers. It is also shown that the more uncertainties on the system, such as the increase share of variable renewable capacities, the more a capacity remuneration is needed to ensure correct levels of investment and maintenance.\\
		
Finally, we provide a numerical illustration of the optimal capacity payment obtained with our proposed optimal contract, compared with the payment supplied by the spot market. This numerical illustration is inspired by the French electricity system.\\
 
The second section of this paper is devoted to the presentation of the model and the objective functions of both the agent and the principal, with a brief summary of the resolution methodology (the details of which are left to the Appendix). In the third section, we present our case study; the French electricity system and provide some numerical interpretations of the optimal capacity payment. Mathematical proofs and details are included in the Appendix for the sake of clarity.  
\section{The model}
 In order to study CRM in a context of information asymmetry, we propose a non-zero sum Stackelberg game with a principal-agent formulation, \ie the gain of one party does not come necessarily from the loss of the other. In this setting, the aggregation of consumers or an entity representing them such as the TSO, proposes to producers a capacity payment which optimally complements the revenue they (producers) obtain on the spot energy market. This payment incentivizes them to invest optimally in power plants management (construction, maintenance, etc..) to ensure an acceptable level (for consumers) of shortage occurrences. The proposed payment is a way to correct the information asymmetry faced by consumers (as they cannot observe directly producers decisions concerning the capacities of the production mix, thus the need for incentive), and to share the risks coming from demand and available capacities uncertainties between the two parties. Moreover, the proposed payment limits producers' potential abuse of market power. Indeed, without capacity payment, producers may decide to under-invest in order to obtain high remunerations from a spot market with more shortages and price spikes.
	
\subsection{Principal-Agent Problem: a brief review} 
	
Contract theory, or principal--agent problem, is a classical moral hazard problem in microeconomics. The simplest formulation involves a controlled process $X$ and two parties; the principal and the agent. The controlled process is called output process and represents the value of the firm for example. Principal owns the firm, and delegates its management to agent, i.e., the control of the output process $X$. So principal hires agent at time $t=0$ for the period $\edg{0,T}$, in exchange for a terminal payment (a contract) $\xi$ paid at time $T$, based upon the evolution of the output process during the contracting time period. In other words, $\xi$ is an $\Fc_T$--measurable random variable, (a function of the realized uncertainties on $X$ up to time $T$), and thus can be a function of the firm value (a percentage of the final gain for example, or a function of the whole trajectory, etc..). However, agent's effort is not observable and/or not contractible for principal, which means that $\xi$ cannot depend on the effort (work) of agent, hence the moral hazard.\\

	Each of the parties aims at maximizing a utility function. The agent acts on the output process $X$ via some control $\a$ (his management decision) and has to pay a cost $c^{A}\brak{\a}$ as a function of the efforts (the management decision $\a$), and expects a payment $\xi$ from principal at time $T$. Agent also has a reservation utility $U_A\brak{\Rc}$, to be thought of as a participation constraint, with $\Rc$ the cash equivalent of this constraint: agent accepts the contract $\xi$ only if $\xi$ satisfies $V^{A}\brak{\xi}\geq U_A\brak{\Rc}$, otherwise he will refuse it. In the case where agent accepts the contract, we can formulate his problem as follow:
	
	\begin{align} \label{intro : Agent pb}
	V^{A}\brak{\xi}  =  \sup_{\a}\E\edg{U_{A}\brak{\xi-\int_{0}^{T}c^{A}\brak{\a_{t}}dt}}.
	\end{align}

	The principal benefits from the output process $X$, and so 
	incites agent to put the effort $\a$ (to work hard) via	the \mbox{contract $\xi$.} Principal tries to find the optimal incentive ($\xi$), while
	respecting agent's participation constraint. Principal's problem is written therefore as:
	
	\begin{align}
	V^{P}=\sup_{\substack{\xi\\
			V^{A}\brak{\xi}\geq U_A\brak{\Rc}
		}
	}\sup_{{\a}^{\star}(\xi)}\E\edg{U_{P}\brak{-\xi+X_{T}^{{\a}^{\star}(\xi)}}},
	\end{align}
	where the contract $\xi$ satisfies the participation constraint $V^{A}\brak{\xi}\geq U_A\brak{\Rc}$ and ${\a}^{\star}(\xi)$ denotes agent's optimal effort (response) given the contract $\xi$, \ie the solution to \eqref{intro : Agent pb}. The first supremum (over ${\a}^{\star}(\xi)$) expresses the fact that given a contract $\xi$, agent solves his problem (we will see later that the existence of at least one solution is guaranteed). Then in case of existence of multiple solutions to \eqref{intro : Agent pb}, we assume that agent would choose the one which maximizes principal's value function (once his utility maximized, agent is cooperative with principal). We refer the interested reader to \cite{cvitanicz}, and \cite{cvitanicpt} for a more detailed exposition of contract theory and principal-agent problem.\\
	
	As mentioned before, in our model, principal is the aggregation of consumers or an equivalent entity representing them, and agent is the collection of producers. In the sequel we will ease the presentation by referring to these two parties by simply saying ``the consumer'' and ``the producer''.\\
	
	So agent is the producer who exerts an effort (a process which we will denote $\alpha$), to build or invest in the maintenance of peak power plants, to increase the total capacity of the fleet. Agent is compensated an amount $\xi$ by principal (the consumer) for the utility received; the satisfaction of consumption and the insurance against shortage risk. We also account for the moral-hazard (second best in principal-agent terminology), in the sense that effort performed by the agent is not observable by principal. Therefore, principal does not observe $\alpha$, and is not able to know if the available capacity is the result of decisions of maintenance and investments made by agent, or if it is due to market conditions not controlled by the producer, such as unanticipated failures, good or bad weather conditions for renewable energy sources of production. In mathematical words, the capacity compensation $\xi$ given to the producer, cannot be a function of the effort $\a$.
	
	\subsection{Model, state variables and control}\label{subsection Model, state variables and control}
	
	We fix a maturity $T\in\brak{0,\ +\infty}$, and describe the system with two continuous processes $X^{C}$ and $X^{D}$, denoting respectively the electricity generation capacity available at each time $t\in[0,T]$ and the instantaneous electricity demand, both in GigaWatt (GW). $X^C$ represents the aggregation of all production capacities, regardless of the corresponding production technology. The uncertainty of  $X^C$ represents the power outages of conventional power plants and the variability of the availability factor of renewable productions. We denote the state variable $X:=\brak{X^{C},X^{D}}^{\intercal}$, a stochastic process valued in $\R^{2}$.\\

	Agent (electricity producers) controls the generation capacity $X^C$ via an  $\F$--predictable process $(\alpha_t)_{t\in[0,T]}$; at each time $t\in\edg{0,T}$ the control is only based on information prior to $t$ without knowledge of the future. The control $\a$ is expressed as a yield in $\edg{\mbox{Year}}^{-1}$ and  represents the decision at each time $t$ to change the generation capacity by building or dismantling peak power plants; gas turbines for instance. This restriction in the choice of only one technology for the control (maintaining/building or destroying) simplifies parameters calibration and allows for the use of a continuous time setup. Numerically, this is approximated by a small-step discretization as peak power plants are quite rapidly adjustable. The control is only on the average value of the available capacity and not the volatility, which spares us a lot of technicalities. Indeed, one could expect that investing in wind power or solar panels would increase the uncertainty of available generation capacity (volatility), as opposed to thermal plants which have a more controllable production.\\
		
	The instantaneous available capacity process $X^C$ is driven by a controlled geometric Brownian motion, which has the property of staying positive consistent with available capacity. $X^C$ starts from $x_0^C$ and has the infinitesimal increments over $dt$:
	\begin{align}
	dX_t^C = \a_tX_t^Cdt + \sigma^CX_t^CdW^{C,\a}_t,\mbox{ for $t$ in $\edg{0,T}$},
	\end{align}	 
	where $\a_tX_t^Cdt$ is the variation on average capacity implied by the effort $\alpha$  and $\sigma^CX_t^CdW^{C,\a}_t$ is the stochastic part in the available capacities due to uncertainties, with $\sigma^C>0$ the volatility parameter.  Remark that we overlook ageing and deterioration in our model, which is justified by taking a short maturity $T$ compared to the average lifetime of power plants.\\

	The demand $X^D$ is modeled as the exponential of a mean reverting process to ensure that $X_t^D>0$, for $t$ in $\edg{0,T}$, and that the demand oscillates around some average level. The initial condition is fixed as $X_0^D=x_0^D$, and the infinitesimal variation over $dt$ is modeled by:
	\begin{align}
	d\log\brak{X_t^D} = \mu^{D}\brak{m^{D}-\log\brak{X_{t}^{D}}}dt + \sigma^DdW^{D}_t,\mbox{ for $t$ in $\edg{0,T}$}.
	\end{align}
	The term $\sigma^DdW^{D}_t$ is the random part in the variation with $\sigma^D>0$, and $m^{D}\in\R$ the long term average (of $\log\brak{X_t^D}$) and $\mu^{D}>0$ the speed of mean reversion.\\
	
We denote by $\Uc$ the set of admissible control processes $\a$, defined as the $\F$--predictable processes valued in the compact interval $[\alpha_{min},\alpha_{max}]$, with $\alpha_{min}<0$ and $\alpha_{max}>0$, which implies that construction and dismantling are both allowed. The restriction $\alpha\in[\alpha_{min},\alpha_{max}]$ although reasonable from an economical perspective, will only be used to simplify the rigorous mathematical resolution of the problem, and so $\abs{\alpha_{min}}$ and $\abs{\alpha_{max}}$ are set arbitrarily large.\\

	We write in a more compact form the dynamic of the state variables
	
	\begin{align} \label{eq:dyn-SDE}
	X_{t}=x_{0}+\int_{0}^{t}\mu\brak{X_{s},\a_{s}}ds+\int_{0}^{t}\sigma\brak{X_{s}}dW_{s}^{\a},\ \forall t\in[0,T],
	\end{align}
	with $\alpha$ the control process, $x_{0}\in\R_{+}^{2}$ a fixed initial condition and
	\begin{align}
	{\mu\brak{x,\a}:=
		\tilde{\mu}\brak{x}+
		\brak{
			\begin{matrix}
			\a x^{C}\\
			0
			\end{matrix}}},\mbox{ and } 
	\sigma\brak{x}:=\brak{
		\begin{matrix}
		\sigma^{C}x^{C} & 0\\
		0 & \sigma^{D}x^{D}
		\end{matrix}},
	\end{align}	
	with
	\begin{align}
	{\tilde{\mu}\brak{x}:=
		\brak{
			\begin{matrix}
			0\\
			\brak{\mu^{D}\brak{m^{D}-\log\brak{x^{D}}}+\frac{\sigma_{D}^{2}}{2}}x^{D}
			\end{matrix}}},
	\end{align}
	or equivalently
	\begin{align}
	\brak{\begin{array}{c}
		X_{t}^{C}\\
		\log\brak{X_{t}^{D}}
		\end{array}}
	=
	\brak{\begin{array}{c}
		x_{0}^{C}\\
		\log\brak{x_{0}^{D}}
		\end{array}}
	+
	\int_{0}^{t}\brak{
		\begin{array}{c}
		\a_{s}X_{s}^{C}\\
		\mu^{D}\brak{m^{D}-\log\brak{X_{s}^{D}}}
		\end{array}}
	ds
	+
	\int_{0}^{t}\brak{
		\begin{matrix}
		\sigma^{C}X_{s}^{C} & 0\\
		0 & \sigma^{D}
		\end{matrix}}
	dW_{t}^{\a}.
	\end{align}

	\begin{Remark}\label{seasonality}
{\upshape 
	In reality, the processes $X^C$ and $X^D$ exhibit a strong seasonal behavior (annual, weekly and daily patterns). These seasonalities are explained by patterns of electricity consumption in day to day life and weather conditions (heating in winter, solar production in the day ...).  
	For expository purposes, we consider a deseasonalized version of state variables. The aim of this simplification is to focus
	the analysis on random consumption peaks, and how they should be dealt with, as opposed to seasonal variations which can be anticipated. }
	\end{Remark}
	
	\begin{Remark}{\upshape
		Power demand  is considered to be inelastic with respect to electricity prices.} 
	\end{Remark}
	
	 \begin{Remark}
	{\upshape Even though we work on the space of processes valued in $\R^{2}$, our model ensures that $X$ takes only positive values, in $\R_+^{2}$, since the capacity process $X^C$ follows a log-normal distribution, and $X^D$ is defined as the exponential of an Ornstein Uhlenbeck process. One could simplify the model and define these two variables as the canonical processes on the space of exponentials of continuous functions. This would indeed simplify the description of the dynamics. Adjustements of the cost functions would also be required.  
		However, the numerical calibration on real data of these adjusted costs is not possible, so we chose  to keep on with the current model.}
	\end{Remark}
	\begin{Remark}
		{\upshape 
	As mentioned above, the capacity and demand processes defined by the controlled SDE \eqref{eq:dyn-SDE} are both non negative. However, they are unbounded from above and have poor integrability properties which could cause technical issues especially since we will be using exponential utility functions.\\
	To avoid this problem we define an arbitrarily large constant $x_{\infty}\in\R_+$, and the function $x\mapsto \underline{x}:=x\land x_{\infty}$. So whenever necessary we will use $\underline{X}^C=X^C\land x_{\infty}$ and $\underline{X}^D=X^D\land x_{\infty}$ or the vector version defined componentwise, \ie $\underline{X} = \brak{\underline{X}^{C},\underline{X}^{D}}^{\intercal}$. Remark that since $x_{\infty}$ can be set arbitrarily large, it neither impacts the results nor represents a restriction on the capacity or the demand.}
	\end{Remark}
	\subsection{Spot energy payment}
	Without capacity payment, the only transaction between consumers and producers is the reward for energy production. This reward corresponds to the spot price of electricity. Therefore, for a time interval $\edg{0,T}$, the consumer pays for his consumption the amount 
	\begin{align}
	S_T:=\int_{0}^{T}s(X_t)dt,
	\end{align} 
	and the producer receives $S_T$, where $s : \R^2\rightarrow \R_+$ is the reward per unit of time, defined as
	\begin{align}\label{eq : definition small s}
		s(x) :=P\brak{\underline{x}}\underline{x}^C\land\underline{x}^D,
	\end{align}
	and $P : \R^2\rightarrow \R_+$ is the spot price function which we define as
	\begin{align}\label{eq : prix spot}
	P\brak{x}:=\beta_{0}e^{-\beta_{1}\brak{x^{C}-x^{D}}},\mbox{ with $\beta_0,\beta_1>0$.}
	\end{align}
	\no Note that the spot energy payment only accounts for delivered energy, i.e., the minimum between the demand and available capacity; the requested power $X^D$   in standard situations, and just the available capacity $X^C$ in the case of a shortage. Different choices of electricity spot price functions can be found in \cite{aid} or \cite{aidcht} and our model is directly inspired by them. In particular, the function $P$ in  \eqref{eq : prix spot} captures a key feature in our problem: the relationship between the spot price and the residual \mbox{capacity ($X^{C}-X^{D}$).}  
	\no In the sequel, we will call $S_T$ the spot payment. Remark that $S_T$ represents a cost for the Principal and a reward for the Agent. 
	
	\subsection{Producer's problem}
	
	The agent is the electricity producer, and provides consumer with electricity, for a terminal payment $\xi+S_T$. The producer is in charge of choosing the investment policy in power plants via the process $\a$, and is subject to its costs.\\

	\no We model producer's instantaneous costs as a quadratic function of state variables 
	\begin{align}\label{eq : definition cost}
	c^{A}(x,\a):=\tilde{c}^{A}(x)+\kappa_1(\a \underline{x}^{C})+\kappa_2\frac{(\a \underline{x}^{C})^{2}}{2},\mbox{ for $(x,\a)\in\R_+^2\x[\alpha_{min},\alpha_{max}]$, and $\kappa_1,\kappa_2>0$.}    
	\end{align}
	The term $\kappa_1(\a \underline{x}^{C})+\kappa_2\frac{(\a \underline{x}^{C})^{2}}{2}$ is the cost of building or dismantling peak power plants,
	where $\kappa_1$ is the cost per unit, and $\kappa_2$ is a penalization adjustment term as the quadratic cost of construction, since the marginal cost of building at a given time step is increasing.  
	We define then
\begin{align}
\tilde{c}^{A}(x):=a \underline{x}^{C}+b (\underline{x}^{C}\land \underline{x}^{D}),\mbox{ with $a,b>0$},
\end{align}
where the first term is the cost of maintenance and the second models the variable cost of production. These variable costs of production are proportional to the available generation capacity $x^C$ and the minimum between this capacity and demand $x^{C}\land x^{D}$.\\
	
During the time period $\edg{0, T}$, producer provides electricity to consumer, and receives the payment $\xi+S_T$ at time $T$. The amount $\xi$ represents the payment producer receives for the availability of capacity, in addition to the spot payment $S_T$.\\
	
To include the moral hazard in our problem, we use the weak formulation and so we introduce $\P^{\a}$, the law of the process $X$, weak solution of the stochastic differential equation \eqref{eq:dyn-SDE} with a control process $\a$, and $\Pc$ the set of probability measures $\mathbb{P}^{\a}$.

\no The producer's objective function or his average perceived utility is defined as 
\begin{align}
J_{0}^{A}\brak{\xi,\P^{\a}}:=\E^{\P^{\a}}\edg{U_{A}\brak{\xi+S_T-\int_{0}^{T}c^{A}\brak{X_{t},\a_{t}}dt}},\mbox{ for $\brak{\xi,\P^{\a}}\in\Xi\x\Pc$},
\end{align}
for a given contract $\xi\in\Xi$ and a choice of $\a\in\Uc$ to which we associate the probability measure $\P^{\a}\in\Pc$, and  where $U_{A}$ is a utility function expressing the risk aversion; increasing
and concave. For tractability, we choose an exponential utility function, $U_{A}(x):=-\exp(-\eta_{A}x)$ with $\eta_{A}>0$, the agent's risk aversion. A rigorous definition of $\Pc$ and $\Xi$ is provided in Appendix \ref{appendix : Rigorous mathematical framework and weak formulation}, along with the weak formulation of the problem.\\ 
	
	The producer is encouraged to provide enough capacity, otherwise the consumer would reduce his payment $\xi$. The moral hazard is modeled by adding the restriction that the payment $\xi$ is a function only of $X$, not $\a$ or $\P^{\a}$. So $\xi$ is $\Fc_{T}$--measurable, where $\Fc_{T}$ by definition models the information gathered from the observation of the process $X$ up to time $T$, and $\xi$ is a function of this information. 
	That is to say the consumer only observes the state variables $X$ as stochastic processes, and has no access to the control $\a$ or $\P^{\a}$ and cannot see if the randomness of $X$ is coming from external uncertainties or from producer's actions. In other words, the producer controls the \emph{law} $\P^{\a}$ of the process $X$, \ie the probability of having some trajectories rather than others, and the consumer observes the realized trajectory and fixes the payment $\xi$ as a function of $X$. We stress here that an important feature of our model is that the structure of the contract is defined \emph{ex-ante} while its exact value is provided only \emph{ex-post} depending on the realized uncertainties.\\
	
In addition, the producer has a participation constraint $U_A\brak{\Rc}\in\R$, with $\Rc\geq 0$ its cash equivalent. Thus, producer will accept the contract $\xi$ only if he can expect to retrieve from $\xi$ a utility above the level $U_A\brak{\Rc}$. Indeed, the producer has no obligation to accept the contract and is free to refuse it before the start of the time period $\edg{0,T}$.\\
	
Whenever the producer accepts a given contract $\xi$, he wants to make the optimal investment by choosing an appropriate control $\P^{\a}$. Producer (agent) solves the problem	
	\begin{align}
	V_{0}^{A}\brak{\xi} & :=  \sup_{\P^{\a}\in\Pc}J_{0}^{A}\brak{\xi, \P^{\a}}.
	\end{align}
	\\
An agent's control $\P^{\a^\star}\in\Pc\brak{\xi}$ (or equivalently $\a^{\star}\brak{\xi}$) is said to be optimal if it satisfies 
\begin{align}
V_{0}^{A}\brak{\xi}=J_{0}^{A}\brak{\xi,\ \P^{\a^\star}}\!.
\end{align}
We denote by $\Pc^{\star}\brak{\xi}$ the set of agent's optimal controls for some admissible contract $\xi$.
	
\subsection{Consumer's problem}
	
The consumer buys and consumes electricity from the producer during the time period $\edg{0, T}$ and pays for the energy consumed at the spot price $S_T = \int_0^Ts(X_t)dt$, and a capacity remuneration given by the contract $\xi$. Consumer gets an instantaneous utility from  electricity consumption, and a disutility in case of shortage. We therefore model consumer (principal) overall instantaneous utility as
\begin{align}\label{eq : def cp}
c^{P}\brak{x}:=\theta\brak{\underline{x}^{C}\land \underline{x}^{D}}-k\brak{\underline{x}^{D} - \underline{x}^{C}}^{+}\mbox{ with $\theta,k>0$.}   
\end{align}

The first term represents consumers’ reservation value or their willingness to pay for effective consumption (which is $\min\brak{x^C, x^D}$). The larger $\theta$, the more valuable consumption to the consumers. In the literature, $\theta$ is often set using the Value of Lost Load (VOLL) \cite{fabra2018primer}.\\

The second term can be thought of as consumer’s disutility induced by the risk of total or partial blackout. Indeed, in critical situations where some shortage (whenever available capacity is less than the total demand), the system operator's ability to keep the system running decreases and total (all the system) or partial (large geographical zones) blackout may occur. The coefficient $k$ is defined to represent this disutility, which is represented in the model with a function proportional to the level of shortage (the higher the shortage level, the higher the risk of blackout). This term plays a role of ``punishment''--via the contract-- for producer whenever there is a failure to provide sufficient generation capacity to cover the instantaneous demand.\\

\no Altogether, consumer's objective function or expected utility is defined as
\begin{align}
J_{0}^{P}\brak{\xi,\P^{\a}}:=
\E^{\P^{\a}}\edg{U_{P}\brak{-\xi-S_T+\int_{0}^{T}c^{P}\brak{X_{t}}dt}},\mbox{ for $\brak{\xi,\P^{\a}}\in\Xi\x\Pc$},
\end{align}
with $U_{P}$ denoting principal's utility function, similar to agent's
utility function  with a risk aversion $\etap$;
\begin{align}
U_{P}\brak{x}:=-\exp\brak{-\eta_{P}x}\mbox{ with $\eta_{P}>0$}.
\end{align}	
\no Principal's goal is to choose the optimal incentive (payment) for the agent to make an optimal effort. Principal's problem is written 
\begin{align}
V_{0}^{P} 
&:=	\sup_{\xi\in\Xi}\sup_{\P^{\a}\in\mathcal{P}^{\star}\brak{\xi}}J_{0}^{P}\brak{\xi,\P^{\a}},
\end{align}
\ie given the optimal response of the agent to the compensation scheme, the Principal chooses the best contract which maximizes $\sup_{\P^{\a}\in\mathcal{P}^{\star}\brak{\xi}}J_{0}^{P}\brak{\xi,\P^{\a}}$. Furthermore, (as stated earlier) we assume that when given different optimal controls, agent will choose the one that maximizes principal's objective function, which is a standard assumption in contract theory \cite{holmstrom1987aggregation,Sannikov08,cvitanicz}.	
\subsection{Optimal contract and capacity payment}\label{ section : contract decomposition}
	
To find producer's and consumer's optimal policy, we follow the approach presented in \cite{cvitanicpt} for principal-agent problems. We start by considering a special class of contracts; the ``revealing contracts'' as capacity payments. These contracts satisfy the incentive compatibility property, which means that consumer provides them with a recommended policy (or effort) for producer, and producer's optimal response to these contracts corresponds to the recommended effort. Therefore, the consumer can maximize his utility over the set of revealing contracts by identifying first producer's response and choosing the best trade-off between the payment of the contract and utility induced by the corresponding response.\\

We next use a representation result to prove that any contract can be represented as ``revealing'', and thus there is no loss of generality or utility for consumers in optimizing only over such contracts.\\
	
In mathematical words we solve the problem for contracts which can be written as a terminal value of a (special) controlled forward stochastic differential equation (SDE) designed to make agent's response ``predictable''.
Then we prove that we can associate to any admissible contract such a controlled SDE, obtained by solving an appropriate backward SDE.\\
	
The revealing contracts are introduced via an appropriate parametrization of contracts; principal considers the contract as a terminal value of a controlled diffusion process, and controls its initial level and the increments linear in the state variable. The class of \emph{revealing contracts} $\Zc$ is then defined as 
\begin{align}\label{eq : ensemble des contrats revelateur}
\Zc := \crl{Y_T^{Y_0,Z}\mbox{ for some $\brak{Y_0,Z}\in\R\x\Vc$ with }Y_T^{Y_0,Z}\in\Xi},
\end{align}
\no where $\Vc$ is the set of $\F$-predictable processes $Z$ valued in $\R^2$ satisfying some integrability conditions (rigorously defined in Appendix \ref{appendix : definition of the revealing contracts}), and 
\begin{align} \label{eq : contrat revelateur}
Y_{t}^{Y_0,Z}:=Y_{0}+\int_{0}^{t}Z_{s}\cdot dX_{s}-\int_{0}^{t}H\brak{X_{s},Z_{s}}ds,\mbox{ for all }t\in \edg{0,T},
\end{align}
with $H$ corresponding to producer's Hamiltonian,  defined by \begin{align}\label{eq big hamiltonian}
H\brak{x,z}:=\sup_{\a \in [\alpha_{min},\alpha_{max}]}h\brak{x,z,\a},\ \mbox{ for } \brak{x,z}\in\R^{2}\x\R^2,
\end{align}
and $h : \R^{2}\x\R^2\x[\alpha_{min},\alpha_{max}] \rightarrow \R$ defined as
\begin{align}\label{eq : small hamiltonian}
h\brak{x,z,\a}&:=z\cdot\mu\brak{x,\a}+s(x)-c^{A}\brak{x,\a}-\frac{\eta_{A}}{2}|\sigma(x)z|^2,\mbox{ for $(x,z,\a)\in\R^{2}\x\R^2\x[\alpha_{min},\alpha_{max}]$.}
\end{align}
For completeness, we provide a derivation of the class of revealing contracts in Appendix \ref{appendix : Formal derivation of the revealing contracts}.\\

	\no We denote by $\hat{\a} : \R^2\rightarrow [\alpha_{min},\alpha_{max}]$ the maximizer of $h$ which can be easily computed 
	\begin{align}\label{eq : recommended effort}
	\hat{\a}\brak{x^C,z^C}:=\alpha_{min}\lor\brak{\Frac{z^Cx^C-\kappa_1\underline{x}^C}{(\underline{x}^C)^2\kappa_2}}\land \alpha_{max},\mbox{ for $\brak{x^C,z^C}\in\R^{2}$ with $x^C> 0$,}
	\end{align}
where $x^C$ (respectively $z^C$) denotes the first component of $x$ (respectively $z$). The function $\hat{\a}$ will be referred to as the ``recommended effort''--\cite{Sannikov08}--, and can be reasonably approximated when $\abs{\alpha_{min}},\alpha_{max},x_{\infty}\rightarrow +\infty$ as
\begin{align}\label{eq : recommended effort approx}
	\hat{\a}\brak{x^C,z^C}\approx\Frac{z^C-\kappa_1}{x^C\kappa_2}\mbox{ for $\brak{x^C,z^C}\in\R^{2}$ with $x^C> 0$}.
\end{align}
We will stick to the expression \eqref{eq : recommended effort} for the mathematical proofs, and use the approximation \eqref{eq : recommended effort approx} for the interpretations.\\
 
Remark that the process $(Y_t^{Y_0,Z})_{t\in\edg{0,T}}$ depends only on observations of $X$ (consumption and available capacities) which are observable by Principal, as opposed to the effort $\alpha$ and the Brownian motion $W^{\alpha}$. This is consistent with the moral-hazard of this problem.\\

The revealing contracts class $\Zc$ plays a central role in principal-agent problems. Not only does it allow principal to predict agent's optimal control, but also to overcome the main difficulty; the \mbox{\emph{non-Markovianity} of $\xi$,} \ie the dependence of the payment on the whole paths of the demand $X^D$ and the capacity $X^C$.\\

We can interpret the revealing contract as a performance index, closely related to agent's continuation value, which comes with a recommended effort $\hat{\a}$ defined in \eqref{eq : recommended effort}; recall that the actual effort provided by agent is neither observable, nor contractible, and therefore principal can only propose $\hat{\a}$ as a recommendation and not an obligation.\\

Nevertheless, Proposition \ref{Agent's response} proves that whenever agent (producer) is rational -which is a reasonable assumption- he will follow the recommended effort since it maximizes his expected utility, and so principal (consumer) can predict agent's (producer's) effort. Furthermore, Proposition \ref{Proposition : existence BSDE}, identifies $\Zc$ to $\Xi$, meaning that any admissible contract can be represented as a revealing one. We present these two results and provide their proofs in the Appendix. 

\begin{Proposition} \label{Agent's response}
	For every contract $Y_T^{Y_0,Z}$ in the class $\Zc$, producer's value function is characterized as 
	\begin{align}
		V_{0}^{A}\brak{Y_T^{Y_0,Z}}  =  U_{A}\brak{Y_0},
	\end{align}
	and his optimal control is given by consumer's recommended effort $\brak{\hat{\a}\brak{X^C_t,Z^C_t}}_{t\in[0,T]}$. 
\end{Proposition}  
 \no The proof is reported in Appendix \ref{Appendix : proof of Agent's response}.	

	\begin{Proposition}\label{Proposition : existence BSDE}
		Let $\xi\in\Xi$. Then there exists a pair $\brak{Y_0,Z}\in \Zc$ such that 
		\begin{align}\label{eq : equation BSDE}
		\begin{cases}
			Y_T^{Y_0,Z} = \xi,\\
			dY_t^{Y_0,Z}= Z_t\cdot dX_t - H\brak{X_t,Z_t}dt.
		\end{cases}
		\end{align}
		Furthermore, 
		\begin{align}\label{eq : estimate on Y}
		\E^{\P^{\hat{\alpha}}}\edg{e^{(\etaa\lor\etap) (1+\delta) \sup_{t\in[0,T]}\abs{Y_t}}}<+\infty.
		\end{align}
		In particular, $\Zc=\Xi$, and therefore  
		\begin{align}\label{eq : reduced principal problem}
		V_0^{P}
		&=\sup_{Y_0\geq \Rc}\sup_{Z \in \Vc}J_{0}^{P}\brak{Y_T^{Y_0,Z},\P^{\hat{\a}}}.
		\end{align}
	\end{Proposition}
	\no  The proof is reported in Appendix \ref{appendix Proposition : existence BSDE}.\\

	The main conclusion is that there is no loss of generality in restricting consumer's problem to contracts in $\Zc$, which are general enough (since $\Zc=\Xi$), and properly parameterized to make producer's response predictable (by proposition \ref{Agent's response}) . So consumer only needs to solve the reduced problem \eqref{eq : reduced principal problem}, \ie to maximize his objective function over the set $\Zc$ (with the two new control variables $Y_0$ and $Z$). This corresponds to a Markovian stochastic control problem which can be solved by standard techniques, and is the object of Proposition \ref{proposition : verification Principal reduced} reported in Appendix \ref{appendix : decomposition contrat}.\\  By virtue of Proposition \ref{proposition : verification Principal reduced}, we can provide a straightforward decomposition of Principal's optimal control in different parts as in the following Corollary \ref{corollary : decomposition contrat}.
	
	\begin{Corollary}\label{corollary : decomposition contrat}
		Under the assumptions of Proposition \ref{proposition : verification Principal reduced}, Principal's optimal contract $\xi^{\star}$ can be written as $\xi^{\star}:=Y_T^{\Rc, Z^{\star}}$ with the following decomposition: 
		\begin{align}
		Y_T^{\Rc, Z^{\star}} &= \Rc+\int_0^TZ^{\star}_t\cdot dX_t -\int_0^TH\brak{X_t,Z^{\star}_t}dt,
		\end{align}
		or equivalently, 
		\begin{align}\label{optimal contract}
		Y_T^{\Rc, Z^{\star}} + S_T
		&=\Rc+\int_0^Tc^{A}\brak{X_t,\hat{\a}\brak{X^C_t,Z^{\star,C}_t}}dt+\int_{0}^{T}Z^{\star}_t\cdot\sigma\brak{X_t}dW_t^{\P^{\hat{\a}}}+\frac{\eta_{A}}{2}\int_0^T|\sigma\brak{X_t}Z^{\star}_t|^2dt.
		\end{align}
		 with $\brak{Z^{\star}_t}_{t\in[0,T]}$ defined in \eqref{eq : hat_z} and $\hat{\a}$ the recommended effort function defined in \eqref{eq : recommended effort}.
	\end{Corollary} 
	This optimal contract consists in a terminal payment to the producer of the random amount $Y_T^{\Rc,Z^{\star}}$, which incites him to follow the recommended effort. Note that any different effort (from the producer) would be sub-optimal in terms of his utility by Proposition \ref{Agent's response}.

	\begin{Remark}\label{remark : summary all}
		{\upshape
		We can make some observations on the remuneration of the producers and the optimal recommended effort:\\
		\rmi Under smoothness assumptions on consumer's certainty equivalent $u$ (the solution of PDE \eqref{eq : pde_u}), the recommended effort $\brak{\hat{\a}(X^C_t,Z^{\star,C}_t)}_{t\in[0,T]}$ is a feedback control as a function of $X_t^C$ and $\partial_{x^C}u(t,X^C_t,X^D_t)$. Recall 
		 \begin{align}
		 	\hat{\a}\brak{X^C_t,Z^{\star,C}_t} &= \alpha_{min}\lor\brak{\Frac{Z^{\star,C}_tX_t^C-\kappa_1\underline{X}_t^C}{(\underline{X}_t^C)^2\kappa_2}}\land \alpha_{max},\tag{\ref{eq : recommended effort}}\\
		 	&\approx\Frac{Z^{\star,C}_t-\kappa_1}{X_t^C\kappa_2}\mbox{ for }\abs{\alpha_{min}},\alpha_{max},x_{\infty}\rightarrow +\infty,\tag{\ref{eq : recommended effort approx}}
		 \end{align}	
		and from the definitions \eqref{eq : hat_z} and \eqref{eq : control xc}, and the approximation \eqref{eq : hat z approx} when $\abs{\alpha_{min}},\alpha_{max},x_{\infty}\rightarrow +\infty$,
		\begin{align}
		\begin{split}
			Z^{\star,C}_t &=\frac{\etap(\sigma^Cx^C)^2 + \frac{1}{\kappa_2}\brak{\frac{X_t^C}{\underline{X}_t^C}}^2}{(\etaa+\etap)(\sigma^CX_t^C)^2+\frac{1}{\kappa_2}\brak{\frac{X_t^C}{\underline{X}_t^C}}^2}\partial_{x^C}u(t,X_t^C,X_t^D) + \frac{\kappa_1}{\kappa_2}\frac{\brak{\frac{X_t^C}{\underline{X}_t^C}}-\brak{\frac{X_t^C}{\underline{X}_t^C}}^2}{(\etaa+\etap)(\sigma^CX_t^C)^2+\frac{1}{\kappa_2}\brak{\frac{X_t^C}{\underline{X}_t^C}}^2},\\
			&\approx \frac{\etap\brak{\sigma^{C}X_t^{C}}^{2}+\frac{1}{\kappa_2}}{\brak{\etaa+\etap}\brak{\sigma^{C}X_t^{C}}^{2}+\frac{1}{\kappa_2}} \partial_{x^C}u(t,X_t^C,X_t^D).
		\end{split}
		\end{align}	
		Therefore, for $t\in [0,T]$ and a given position $(X_t^C,X_t^D)$ the rate of investment $\hat{\alpha}$ recommended by consumer is a function of the capacity $X_t^C$ and $\partial_{x^C}u(t,X_t^C,X_t^D)$; the sensitivity of his certainty equivalent with respect to capacity at time $t$. The latter term depends on the triplet $(t,X_t^C,X_t^D)$, and so the recommended effort depends obviously on $X^C$, and implicitly on $t$ and $X^D$ through the sensitivity of consumer $\partial_{x^C}u(t,X_t^C,X_t^D)$. This point will be further highlighted with numerical experiments.\\
		\rmii Since the effort $\brak{\hat{\a}(X^C_t,Z^{\star,C}_t)}_{t\in[0,T]}$ is defined as a rate (in $\edg{\mbox{Year}}^{-1}$), it can be better understood by observing the process $\brak{X^C_t \hat{\a}(X^C_t,Z^{\star,C}_t)}_{t\in[0,T]}$ which corresponds to the actual construction or dismantling of power plants (in $\edg{\mbox{GW/Year}}$) which is given from the previous expressions as
		\begin{align}\label{eq : forme approximative de alpha}
			\begin{split}
			X^C_t\hat{\a}\brak{X^C_t,Z^{\star,C}_t} &\approx \frac{1}{\kappa_2}\frac{\etap\brak{\sigma^{C}X_t^{C}}^{2}+\frac{1}{\kappa_2}}{\brak{\etaa+\etap}\brak{\sigma^{C}X_t^{C}}^{2}+\frac{1}{\kappa_2}} \partial_{x^C}u(t,X_t^C,X_t^D) -\frac{\kappa_1}{\kappa_2},\\
			&\approx\frac{w(X_t^C)\partial_{x^C}u(t,X_t^C,X_t^D)-\kappa_1}{\kappa_2},
			\end{split}
		\end{align}
		with $w(X_t^C)\in (0,1)$ a positive weight function related to the risk aversion of producer (because $\etaa>0$). In the extreme case $\etaa\rightarrow 0$, $w(X_t^C)=1$ and the recommended effort becomes clear; the (linear) marginal cost of construction is $\kappa_1$, and so the optimal control is to construct if $\partial_{x^C}u(t,X_t^C,X_t^D)>\kappa_1$ and dismantle power plants if $\partial_{x^C}u(t,X_t^C,X_t^D)<\kappa_1$. In particular, $\hat{\alpha}$ depends only on  the certainty equivalent of consumer (and $X^C$), not producer, since the latter is compensated by the contract instead. The general case $\etaa>0$, with $0<w(X_t^C)<1$ needs further considerations in terms of the magnitude of the sensitivity $\partial_{x^C}u(t,X_t^C,X_t^D)$ to offset the weight $w(X_t^C)$, but leads to similar results. \\
		\no\rmiii $Y_T^{\Rc,Z^{\star}}$ covers all the costs the producer has to pay to follow the recommended capacity policy \eqref{eq : recommended effort} and to produce electricity to match the demand. Therefore, the optimal contract compensates those costs taking into account what the producer is earning on the spot market. We recall below producer's costs (we omit the truncation function for exposition clarity) 
		\begin{align}
		\begin{split}
		\underbrace{\int_0^Tc^{A}\brak{X_t,\hat{\a}\brak{X^C_t,Z^{\star,C}_t}}dt}_{\mbox{{Producer's costs}}}=&\underbrace{\int_0^T\kappa_1 \hat{\a}\brak{X^C_t,Z^{\star,C}_t}X_t^{C} + \kappa_2\frac{\brak{\hat{\a}\brak{X^C_t,Z^{\star,C}_t} X_t^{C}}^{2}}{2}dt}_{\mbox{Construction costs}}\\&+\underbrace{\int_0^Ta X_t^{C}dt}_{\mbox{Maintenance costs}}+\underbrace{\int_0^Tb (X_t^{C}\land X_t^{D})dt}_{\mbox{Production costs}}, 
		\end{split}
		\end{align}
		\rmvi $Y_T^{\Rc,Z^{\star}}$ shares the risk (realized uncertainties on demand and capacity) between producers and consumers, by transferring part of the randomness to the agent, while providing him with a risk compensation at the same time, to overcome his risk-aversion: 
		\begin{align}
		\mbox{{Risk part}}=\underbrace{\int_{0}^{T}Z^{\star}_t\cdot\sigma\brak{X_t}dW_t^{\hat{\a}}}_{\mbox{{Risk shared}}} + \underbrace{\frac{\eta_{A}}{2}\int_0^T|\sigma\brak{X_t}Z^{\star}_t|^2dt}_{\mbox{Risk compensation}},
		\end{align}
		and the risk shared can be interpreted as a ``reward for good luck'' and a ``punishment for bad luck'' as in \cite{hoffmann2010reward}, for both of the external noises $W^C$ and $W^D$, which is possible since the contract is defined \emph{ex-ante} and paid \emph{ex-post}. In particular, our model accounts for the risk on the uncontrolled demand $X^D$ in two ways; through the recommended effort (as explained in \rmi\!\! and \rmii\!\!\!), and the optimal contract via the ``Risk shared'' term.\\
		\rmv $Y_T^{\Rc,Z^{\star}}$ is a random variable which depends on the scenario. In particular, its value changes as the uncertainties change, and might even become negative. This means that agent might earn less or more than his total costs, depending on the outcome of uncertainties (for example very sunny or windy years might lead to low spot prices and therefore to a higher capacity remuneration). Nevertheless, in expectation, agent (producer) is guaranteed to earn $\Rc$; the cash equivalent of the reservation utility.\\  
		\rmvi We can rewrite the decomposition \eqref{optimal contract} as follow:
		\begin{align}\label{eq : decomposition en lettres}
		\mbox{Capacity remuneration}+\mbox{Spot compensation}=\Rc+\mbox{{Producer's costs}}+\mbox{{Risk shared}}+\mbox{Risk compensation}.
		\end{align} }
	\end{Remark}
	
	\subsection{Producer's participation constraint: the problem without capacity payment}
	
	\no In absence of a capacity payment, producer's only income is the spot compensation and therefore his problem is a standard Markovian stochastic control problem:
	\begin{align}\label{eq : uncontrolled producer}
	\hat{V}_0^A:=V^A_0\brak{0}=\sup_{\P^{\a}\in\Pc}\E^{\P^{\a}}\edg{U_A\brak{S_T-\int_0^Tc^A\brak{X_t,\a_t}dt}},\mbox{ for $\P^{\a}\in\Pc$}.
	\end{align}
	In this case, the consumer has no bargaining power and no control on investment decisions in capacity, and is then subject to shortage risk. The producer does not care anymore about consumer's value function, and there is no guarantee that a criteria such as the LoLE constraint is satisfied (recall that the Loss of Load Expectation (LoLE) is the targeted maximum number of hours of shortage per year, set at 3 hours per year for most European countries).\\
	
	Because of the structure of the spot function (decreasing in $x^C-x^D$), the producer makes a compromise between having few installed capacities (less than in the case with a capacity payment) to save maintenance costs and increase the spot prices, and enough capacities to satisfy (part of) the demand, to earn more on the spot market (since only sold energy generates a cash flow).  Remark that this kind of arbitrage can be seen in practice even among producers in perfect competition.\\
	
	The resolution of problem \eqref{eq : uncontrolled producer} is the object of Proposition \ref{proposition : verification pb sans contrat} reported in Appendix \ref{appendix : participation constraint}. In the absence of a contract, producer's value function $\hat{V}^A_0$, \ie the solution  to problem \eqref{eq : uncontrolled producer} given in Proposition \ref{proposition : verification pb sans contrat} provides a good proxy for the participation constraint, which we define as follow :
	\begin{align}\label{eq : reservation utility}
	\Rc:= U_A^{-1}\brak{\hat{V}_0^A} \lor 0.    
	\end{align}
	
	\no The maximum is taken in the previous equation between 0 and $U_A^{-1}\brak{\hat{V}_0^A}$ as the producer has two choices:  to operate the power plants if $U_A^{-1}\brak{\hat{V}_0^A}\geq 0$ and earn the spot price which provides a utility $\hat{V}_0^A$, or (if  $U_A^{-1}\brak{\hat{V}_0^A}<0$), to stop all activities which would lead \mbox{to 0} earnings (assuming that we neglect any agency costs related to bankruptcy).	
	\section{Numerical results and interpretations}
		In this section our model is numerically solved for a stylized system, based on the French electricity power system.	We implement the optimal capacity contract and optimal policy, by numerically solving the PDE \eqref{eq : pde_u} describing consumer's value function,
	with parameters calibrated on the French power system. Then we observe multiple scenarios and the evolution of state variables under this policy. A more precise description of the numerical resolution procedure is provided in Appendix \ref{Appendix : Numerical method}.
	\subsection{Case study: the French power system}
	
	We consider a time horizon $T = 5\edg{\mbox{Years}}$, and we discretize it with a time step $\Delta t = \frac{1}{400}\edg{\mbox{Years}}$ for the diffusion of state variables, roughly speaking, over one time step per day.

	\no The state variables $X^{C}$ and $X^{D}$ are expressed in GW, and the contract and costs
	(quantities inside of the utility function) are in $10^6$\euro, ($\edg{\mbox{M\euro}}$). As stated earlier, $X^C$ is the instantaneous overall available capacity, and the control is only on the peak power plants which are assumed to be gas turbine.
	
	\subsubsection{Capacity and demand}
	
	 We use the generation capacity, demand and spot prices available online\footnote{The French TSO RTE website for capacity and demand \mbox{\url{https://clients.rte-france.com}}, and the EPEX SPOT website \url{https://www.epexspot.com} for spot prices.} for the time period 29/06/2009-15/12/2014. Remark that we stop at $2014$ because the latest available capacity data is provided in that year, as the French TSO stopped publishing available capacity records on an aggregated basis per technology. Nevertheless, the French production mix did not change a lot in the past period, and we can reasonably assume that the uncertainties on capacity generation remain unchanged too.\\ 
	 We start by calibrating the parameters of SDE \eqref{eq:dyn-SDE} modeling the dynamics of $X^C$ and $X^D$. As mentioned in \mbox{Remark \ref{seasonality},} we only consider deseasonalized state variables in our model. Therefore, to deseasonalize the input data, we use a locally weighted scatterplot smoothing algorithm implemented in the software R; the function ``STL'' which decomposes time series into three components: a trend, a seasonal component and a residual noise. This algorithm extracts the trend 
	by averaging locally, then computes the seasonality on residuals by averaging across a given frequency. Once the seasonal component is computed, it is subtracted from the original
	time series to get the deseasonalized data. We apply this procedure twice; once for the annual seasonality and another time for the weekly seasonality.\\

	The demand $X^D$ modeled as the exponential of an Orstein--Uhlenbeck process is calibrated by linear regression of the returns of daily data fixed at 7 p.m., the hour of the day with the highest demand.\\ 
	As for the capacity $X^C$, we take the daily sum of the different generation technology capacities: nuclear, gas, coal, fuel, hydro-power (reservoir and run-of-the-river) and then $\sigma^C$ is calibrated as to have simulated trajectories with similar behaviour with historical (observed) capacity data.\\

	Table \ref{tab : capa et dem params} summarizes our estimated parameters for capacity and demand processes, and we can see in figure \ref{fig : Simu vs hist} a comparison between historical (deseasonalized) data with generated scenarios of demand and capacity with our calibrated parameters. 
\begin{center}
\begin{minipage}[c]{.7\textwidth}
	\begin{table}[H]
		\centering{}%
		\begin{tabular}{|c|c|c|c|}
			\cline{2-4} 
			\multicolumn{1}{c|}{} & Parameter & Value & Unit\tabularnewline
			\hline 	    	\multicolumn{1}{|c|}{\begin{tabular}[c]{@{}c@{}}Available\\ generation capacity\end{tabular}} & $x_{0}^{C}$ & 90 & $\edg{\mbox{GW}}$\tabularnewline
			\cline{2-4} 
			& $\sigma^{C}$ & 0.1 & $\edg{\mbox{Year}}^{-\frac{1}{2}}$\tabularnewline
			\hline 
			\multirow{4}{*}{Demand} & $x_{0}^{D}$ & 60 & $\edg{\mbox{GW}}$\tabularnewline
			\cline{2-4} 
			& $\mu^{D}$ & 61.92 & $\edg{\mbox{Year}}^{-1}$\tabularnewline
			\cline{2-4} 
			& $\exp\brak{m^{D}}$ & 60 & $\edg{\mbox{GW}}$\tabularnewline
			\cline{2-4} 
			& $\sigma^{D}$ & 0.86 & $\edg{\mbox{Year}}^{-\frac{1}{2}}$\tabularnewline
			\hline 
		\end{tabular}
		\caption{Calibrated parameters for capacity and demand. This table provides the set of parameters for which our model fits the deseasonalized data of the French generation capacity and demand over the time period $29/06/2009-15/12/2014$.}
		\label{tab : capa et dem params}
	\end{table}		
\end{minipage}
\end{center}
	\graphicspath{ {Plots/} }
	\begin{figure}[H]
		\begin{subfigure}{.5\textwidth}
			\centering
			\includegraphics[scale=0.4]{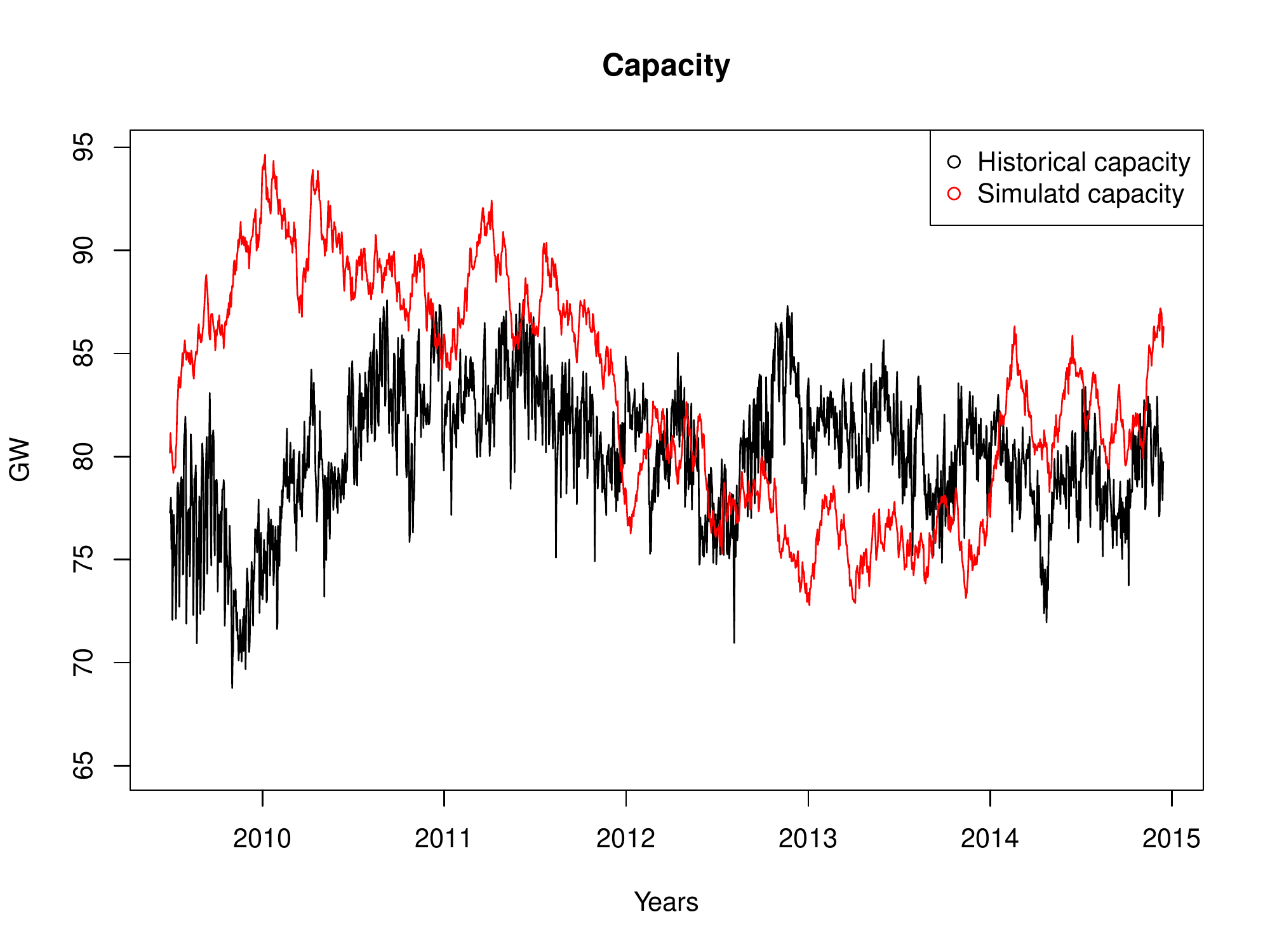}
			\caption{Capacity process}
		\end{subfigure}%
		\begin{subfigure}{.5\textwidth}
			\centering
			\includegraphics[scale=0.4]{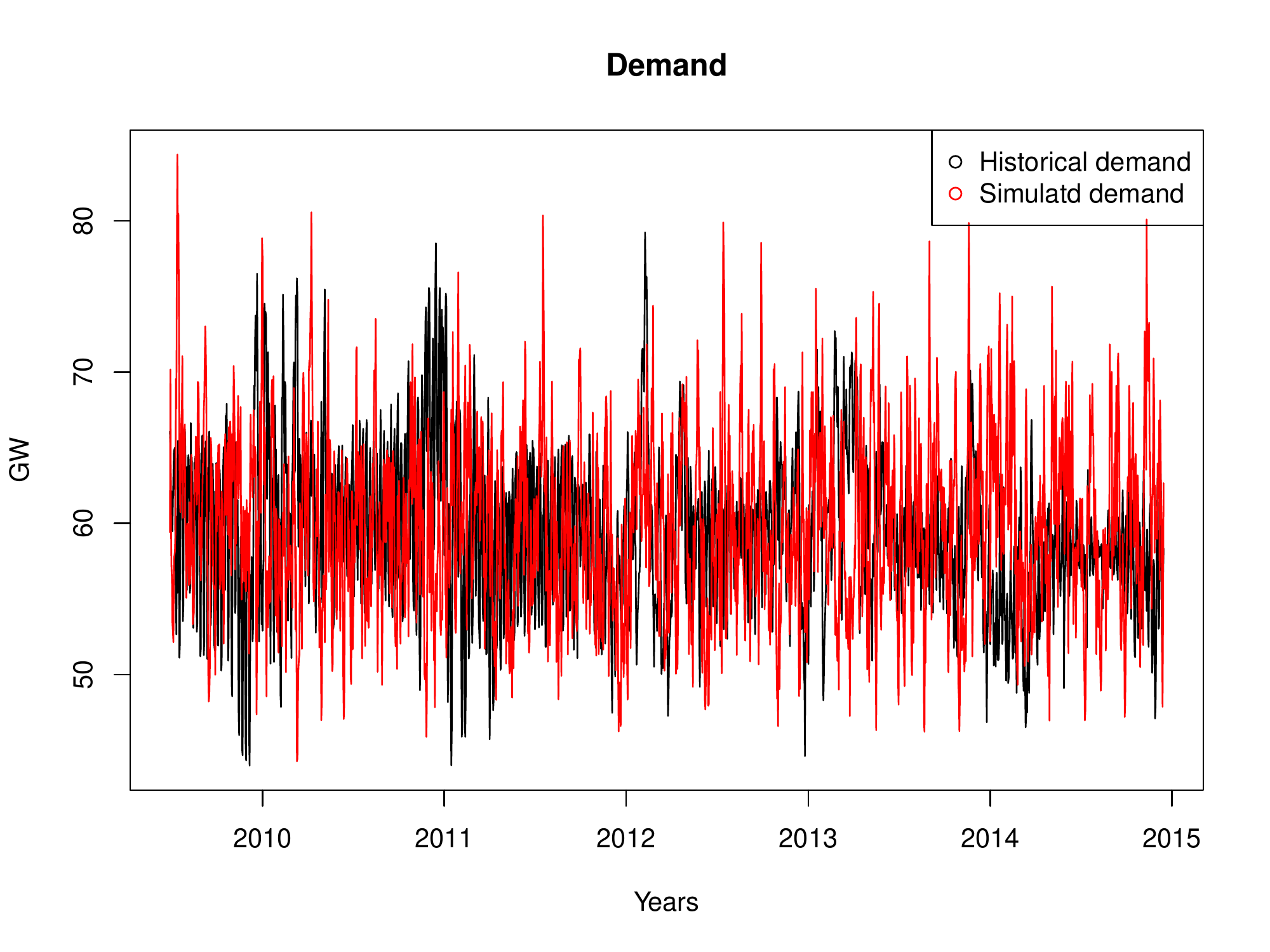}
			\caption{Demand process}
		\end{subfigure}
		\captionsetup{width=.7\linewidth}
		\caption{Comparison between historical data and simulated processes. The black line represents the evolution of the historical data (generation capacity (GW) in the left figure and demand (GW) in the right figure), plotted against a simulated trajectory (in red) with our model and the calibrated parameters of table \ref{tab : capa et dem params}.}
		\label{fig : Simu vs hist}
	\end{figure}

	\subsubsection{Spot price function}
	
	 We calibrate our spot price function $P$ defined in \eqref{eq : prix spot} using historical data, and taking one price per day, at 7 p.m.; the same as for demand data (the hour of daily demand peak).
    The calibration is simply done by taking the $\log$ of the time series, and then applying a linear regression. Remark that the spot price function $P$ is completely characterized by the capacity and demand, so the seasonality is naturally accounted for.
	Table \ref{tab : spot function} summarizes our choice of the spot price function $P$ with its calibrated parameters, and the figure \ref{fig : calibration_spot} represents a comparison between historical and simulated spot prices.    
	
	\begin{table}[H]
		\centering{}%
		\begin{tabular}{|c|c|c|c|c|}
			\cline{2-5} 
			\multicolumn{1}{c|}{} & Model:  $\frac{\edg{\mbox{M\euro}}}{\edg{\mbox{GW}}\edg{\mbox{Year}}}$ & Parameter & Value & Unit\tabularnewline
			\hline 
			\multirow{2}{*}{Spot price function} & \multirow{2}{*}{$P(x)=\beta_{0}e^{-\beta_{1}\brak{x^{C}-x^{D}}}$} & $\beta_{0}$ & 102.8 &$\edg{\mbox{\euro/MWh}}$ \tabularnewline
			\cline{3-5} 
			&  & $\beta_{1}$ & $335.3\x10^{-4}$ & $\edg{\mbox{GW}}^{-1}$\tabularnewline
			\hline 
		\end{tabular}
		\captionsetup{width=.7\linewidth}
		\caption{Spot price model and its calibrated parameters. The first column recalls the model of electricity spot price as a function of the capacity margin $(x^C-x^D)$, and the last three columns provide the parameters of this model calibrated on the French electricity spot prices.}
		\label{tab : spot function}
	\end{table}

	\begin{figure}[H]
		\centering
		\includegraphics[width=.5\textwidth]{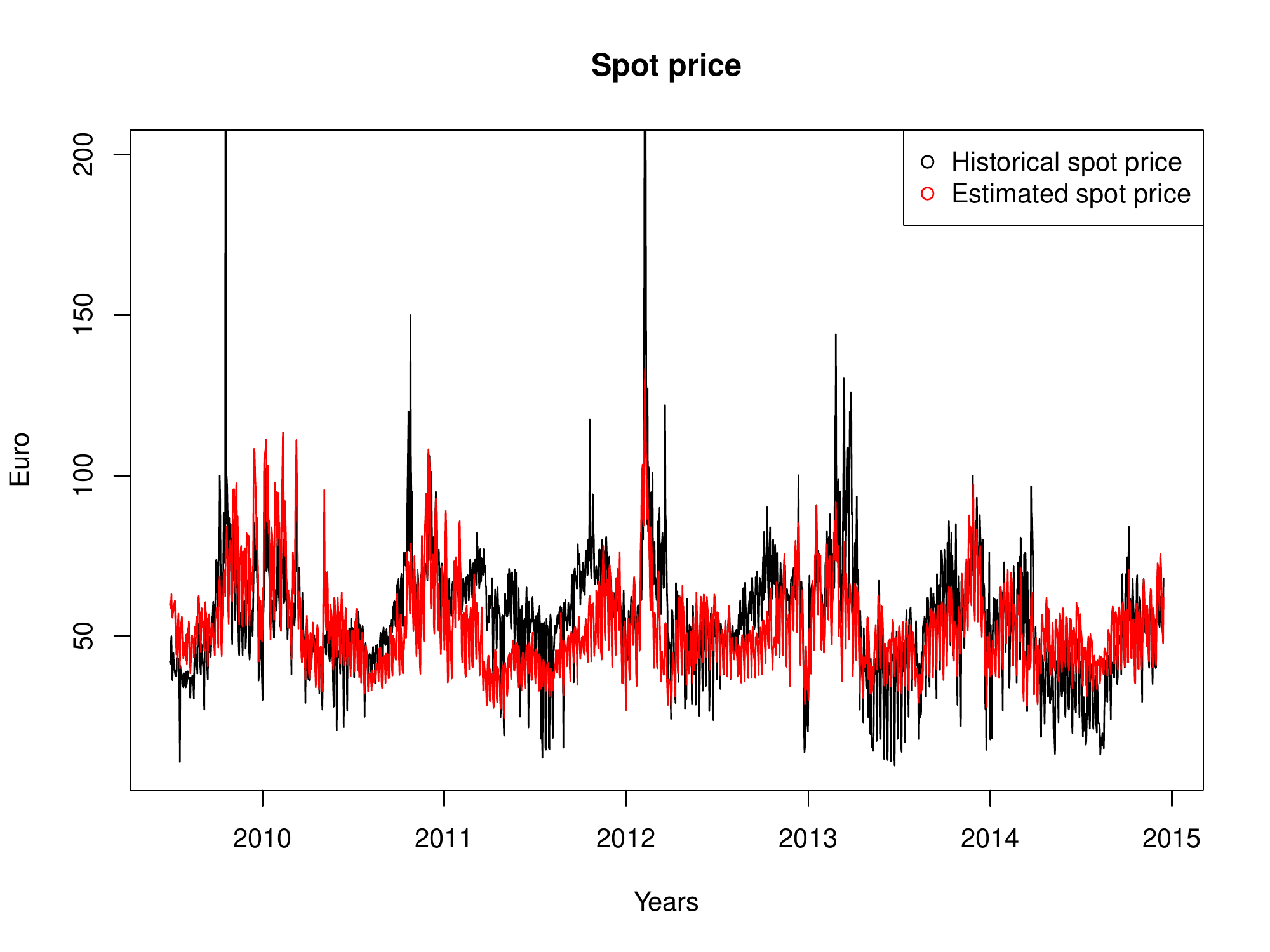} 
		\captionsetup{width=.5\linewidth}
		\caption{Historical and estimated spot prices in euros. In black
			the historical spot prices, and in red the reconstruction of the
			spot price using the function $P$ from table \ref{tab : spot function} and the historical
			realizations of demand and capacity generation.}
		\label{fig : calibration_spot}
	\end{figure}
	 Remark that our model reproduces quite well the behaviour of spot prices but without the largest peaks. This is coherent with a market with a low price cap (a price cap which can be seen for example in many European countries).
	
	\subsubsection{Costs}
	Electricity supplier has to take into account the construction, maintenance and production costs of different power plants. These costs are provided by the French TSO and WEO 2018\footnote{See the \emph{``WEO 2018 report''} and \emph{``Impact assessment of the French Capacity Market, 2018.''} by the French TSO RTE.}.
	The cost of maintenance $a$ and the cost of production $b$ in $\edg{\mbox{\euro/MWh}}$ are estimated as weighted averages between the different costs of technologies, where the weights used for maintenance are the same as for the installed capacities, and those for production costs are taken as the proportions of production; cf Table \ref{tab : weights of technologies}.\\
	
	\begin{table}[H]
	\centering{}
    \begin{tabular}{|c|c|c|}
    \hline
    Technology                    & Installed capacities & Percentage of production \\ \hline
    Nuclear                       & 48 \%                & 72 \%                    \\ \hline
    Coal                          & 2 \%                 & 1 \%                     \\ \hline
    Gas and Fuel                 & 8 \%                 & 4 \%                     \\ \hline
    Wind turbines and Photovoltaic & 18 \%                & 7 \%                     \\ \hline
    Hydropower                    & 19 \%                & 12 \%                    \\ \hline
    Bioenergy                     & 6 \%                 & 4 \%                     \\ \hline
\end{tabular}
\captionsetup{width=.7\linewidth}
\caption{Different technologies and their weights. The first column lists the technologies present in the French mix. The second column represents the percentages in terms of installed capacities for each technology. The third column represents the percentages in terms of energy produced.}
\label{tab : weights of technologies}
\end{table}
	
The cost of construction per unit $\kappa_1$ in $\edg{\mbox{\euro/MWh}}$ is taken as the equivalent annual cost of a gas turbine power plant 
instead of a weighted average since only peak power plants are used for the control (construction or dismantling). The equivalent annual cost of a gas turbine power plant with a total cost of investment $C_{\mbox{{\tiny Total cost}}} = 550\edg{\mbox{\euro/kW}}$-- which should not be confused with the levelized cost of energy, and is computed by dividing the annual cost of investment by the total number of hours per year (8760 hours)-- , a lifetime $T_{\mbox{{\tiny Gas Turbine}}} = 30\edg{\mbox{Years}}$, and a discount rate $r = 8\%$ is computed as $\kappa_1=\frac{nr C_{\mbox{{\tiny Total cost}}}}{1-\brak{1+r}^{T_{\mbox{{\tiny Gas Turbine}}}}}=122.13\edg{\mbox{\euro/kW}}$, where $n$ is the number of upcoming annuities approximated by $n\simeq 2.5$. A more precise computation requires to take $n$ as the number of annuities left to pay during the contract time (between 1 and 5 years), but we chose to simplify and take an average value $n = 2.5$.\\
The adjustment coefficient $\kappa_2$ is taken as $\kappa_2= 2\x \kappa_1$ where $\kappa_2$ is in $\edg{\mbox{\euro/(MWh$\x$MW)}}$. Different sensitivities with respect to this parameter are then performed and ensure that results are stable within a reasonable range of parameters.\\  

A summary of our calibrated parameters for producer's costs can be found in table \ref{tab : costs}.
	
\begin{table}[H]
	\centering{}%
	\begin{tabular}{|c|c|c|c|}
		\cline{2-4} 
		\multicolumn{1}{c|}{} & Parameter & Value & Unit\tabularnewline
		\hline 
		\multirow{4}{*}{Producer's costs} & $\kappa_1$ & 122.13
			& $\edg{\mbox{\euro/kW}}$\tabularnewline
			\cline{2-4} 
			& $\kappa_2$ & $31.8\x 10^{-4}$ & $\edg{\mbox{\euro/(MWh$\x$MW)}}$ \tabularnewline
			\cline{2-4}
			& $a$ & 75.35 & $\edg{\mbox{\euro/(kW$\x$year)}}$\tabularnewline
			\cline{2-4} 
			& $b$ & 17.6 & $\edg{\mbox{\euro/MWh}}$\tabularnewline
			\hline 
		\end{tabular}
	\captionsetup{width=.7\linewidth}
	\caption{Calibrated parameters for construction, maintenance and production costs. $\kappa_1$ and $\kappa_2$ are set to the equivalent annual cost of a gas turbine, and $a$ and $b$ are weighted
		averages (with weights from table \ref{tab : weights of technologies}) of the costs of maintenance and production of all technologies.}
	\label{tab : costs}
\end{table}

	\subsubsection{Utility, disutility and risk aversion}
	
	Risk aversions and utility preferences, $\theta, k,\eta_{A}$, $\eta_{P}$ and the participation constraint $\Rc$, are parameters less straightforward to calibrate.\\
	
	The participation constraint (or cash equivalent of the reservation utility) $\Rc$ is defined as a function of the solution to producer's problem in absence of capacity payment, and given by \ref{eq : reservation utility}. $\Rc$ is computed by numerically solving the PDE \eqref{eq : pde agent solo}.\\
	
	The parameter $\theta$ is expressed in $\edg{\mbox{\euro/MWh}}$, and reflects consumers' satisfaction in consumption per GW over time, or how much they are willing to pay for the electricity. This parameter is calibrated as the \emph{Value of Lost Load} (VoLL)\footnote{See again \emph{``Impact assessment of the French Capacity Market, 2018.''} by RTE.}.\\
	
	We assume that producer should be more risk averse than consumer. The reason of this assumption is that consumer's risk aversion embeds the aversion to shortage represented by the term in $k$. Roughly speaking, as soon as there is a shortage, consumer's utility starts to decrease because of the term $\etap \x k (X^D-X^C)^+$. We can interpret it by saying that consumer is willing to accept more financial risk than producer, in exchange of offsetting the risk of having a shortage.
	We choose therefore $\etaa>\etap$ and we take values inspired by the calibration in \cite{aidpt}. A sensitivity analysis is then performed on these parameters and we find that our results are not affected if producer is more risk averse than the consumer or the contrary.\\
	
	Regarding the parameter $k$, we use a further constraint which is that the average number of shortage hours per year should be reasonable. Indeed, recall that for most European electricity systems, the targeted maximum number of hours of shortage (LoLE), is 3 hours \cite{Newberry16}. Therefore we calibrate these parameters by an iterative procedure, \ie by repeatedly solving the problem, diffusing the state variables, computing the total period of shortage and adjusting the parameters until we attain a reasonable number of shortage hours per year. A possible set of consumer's preferences and  risk aversions of both parties is given in the table \ref{tab : Principal pref and risk aversions}. Of course, one could always argue that this set of parameters is not unique because of the degrees of liberties compared to the number of constraints, but
	this set seems quite reasonable and produces stable numerical results, which we present in the next section along with a sensitivity analysis.\\
	
	\begin{table}[H]
		\centering{}%
		\begin{tabular}{|c|c|c|c|}
			\cline{2-4} 
			\multicolumn{1}{c|}{} & Parameter & Value & Unit\tabularnewline
			\hline 
			\multirow{2}{*}{Consumer's preferences} & $\theta$ & $20000$ & $\edg{\mbox{\euro/MWh}}$\tabularnewline
			\cline{2-4} 
			& $k$ & $200000$ & $\edg{\mbox{\euro/MWh}}$\tabularnewline
			\hline 
			\multirow{2}{*}{Risk aversions} & $\eta_{A}$ & $0.852\x10^{-4}$ & $\edg{\mbox{M\euro}}^{-1}$\tabularnewline
			\cline{2-4} 
			& $\eta_{P}$ & $0.8094\x10^{-5}$ & $\edg{\mbox{M\euro}}^{-1}$\tabularnewline
			\hline 
			Participation constraint & $\mathcal{R}$ & $2.8$ & $\edg{\mbox{\euro/MWh}}$ \tabularnewline
			\hline
		\end{tabular}
		\captionsetup{width=.7\linewidth}
		\caption{Choice of risk aversions, and calibrated preferences. This table provides estimates for risk aversions $\etaa$ and $\etap$, and consumers' utility for consumption $\theta$ (the VoLL), together with their aversion to shortage $k$, and the minimal payment required by producer $\Rc$. This set of parameters generates scenarios with a reasonable number of shortage hours.}
		\label{tab : Principal pref and risk aversions}
	\end{table}

	\subsection{Numerical results : Comparison between the system with and without a CRM}
	
	Once the parameters fixed, we simulate $N=5000$ scenarios and compare three different cases; one ``without a CRM'', where producer adjusts capacities to maximize his utility, another ``with a CRM''; using the optimal policy for both consumer (optimal compensation \eqref{optimal contract}) and producer (recommended effort \eqref{eq : recommended effort}), and a third one with no capacity adjustment (``No adjustment''), \ie no building or dismantling of capacities; leaving them subject to external uncertainties. Table \ref{table : Comparison between policies} summarizes the results of our simulations.
	
	\begin{table}[H]
		\centering
		\resizebox{0.8\linewidth}{!}{
			\begin{tabular}{lrrrrrr}
				\toprule
\multicolumn{1}{c}{ } & \multicolumn{2}{c}{Without CRM} & \multicolumn{2}{c}{With a CRM} & \multicolumn{2}{c}{No adjustment} \\
\cmidrule(l{3pt}r{3pt}){2-3} \cmidrule(l{3pt}r{3pt}){4-5} \cmidrule(l{3pt}r{3pt}){6-7}
  & mean & sd & mean & sd & mean & sd\\
\midrule
\addlinespace[-1em]
\multicolumn{7}{l}{\textbf{}}\\
\hspace{1em}Shortage hours per year [Hours] & 6165.4 & 1459.1 & 2.2 & 7.4 & 178.4 & 594.3\\
\hspace{1em}Average Spot price [euro/MWh] & 146.5 & 30 & 37.6 & 9.6 & 43 & 16.4\\
\hspace{1em}Average Margin [GW] & -6 & 8.2 & 32.4 & 9.0 & 28.1 & 11.4\\
\hspace{1em}Spot revenues [euro/MWh] & 134.6 & 25.6 & 37.8 & 9.6 & 43.1 & 16.2\\
\hspace{1em}Capacity payment [euro/MWh] & NA & NA & 12.3 & 13.2 & NA & NA\\
\hspace{1em}Spot + Capacity payment[euro/MWh] & 134.6 & 25.6 & 50.1 & 10.9 & 43.1 & 16.2\\
\hspace{1em}\hspace{1em}\hspace{1em}Participation constraint [euro/MWh] & NA & NA & 2.8 & 0.0 & NA & NA\\
\hspace{1em}\hspace{1em}\hspace{1em}Risk shared [euro/MWh] & NA & NA & -0.2 & 11.7 & NA & NA\\
\hspace{1em}\hspace{1em}\hspace{1em}Risk compensation [euro/MWh] & NA & NA & 14.9 & 3.9 & NA & NA\\
\hspace{1em}\hspace{1em}\hspace{1em}Total costs [euro/MWh] & 77.4 & 9.3 & 32.7 & 3.2 & 30.5 & 1.6\\
\hspace{1em}\hspace{1em}\hspace{1em}\hspace{1em}\hspace{1em}Construction and dismantling [euro/MWh] & 50.3 & 9.7 & 1.8 & 3.7 & 0 & 0\\
\hspace{1em}\hspace{1em}\hspace{1em}\hspace{1em}\hspace{1em}Maintenance [euro/MWh] & 9.5 & 0.7 & 13.2 & 1.3 & 12.9 & 1.6\\
\hspace{1em}\hspace{1em}\hspace{1em}\hspace{1em}\hspace{1em}Production [euro/MWh] & 17.6 & 0 & 17.6 & 0.0 & 17.6 & 0\\
\bottomrule
		\end{tabular}}
	\captionsetup{width=.8\linewidth}
	\caption{Comparison between different policies. This table regroups the average and standard deviation of the most relevant quantities over $N=5000$ simulations. The first two columns represent the scenarios generated without a CRM, \ie where producer controls the capacity and his only income is the spot revenue. The next two columns represent our proposed CRM, with scenarios generated following the recommended effort and a capacity payment. The last columns provide the results for scenarios generated without any control on capacity and without capacity payment.}
	\label{table : Comparison between policies}
	\end{table}
	
	\subsubsection{The system evolution without a capacity payment}
	
	We start by analyzing the system without a capacity payment. As mentioned before, the producer has market power and no incentive to satisfy the LoLE constraint, since his goal is to maximize his utility function instead of just offsetting his marginal costs. On the contrary, as his only compensation is from the spot market, his optimal strategy consists in finding the equilibrium between high enough spot prices (corresponding to low or even negative capacity margins) and high enough available capacity as spot compensation is $\int_0^TP(X_t)X_t^C\land X_t^Ddt$. We can see from table \ref{table : Comparison between policies} that producer settles for a -6 GW average, which corresponds to an average spot price of 146 euros/MWh. It follows from this negative equilibrium average margin that the system is in shortage situation most of the period $\edg{0,T}$, which is confirmed by the numerical results (6165 shortage hours per year).\\
   
    In figure \ref{fig : violent_cap_et_dem_ag_sol}, we can observe one scenario without a CRM. As stated before, producer's optimal strategy is to decrease the capacity level which leads to the spot prices increasing, see figure \ref{fig : violent_spot_price_Ag_sol}. This decrease in capacity continues even after reaching the average demand level (about $60$ GW) and attains an equilibrium (around 20 GW in the scenario which is more severe than average). This explains the high construction and dismantling costs mainly due to dismantling actions. 
	
	\begin{figure}[H]
		\centering
		\begin{subfigure}{.4\textwidth}
			\centering
			\includegraphics[scale=0.4]{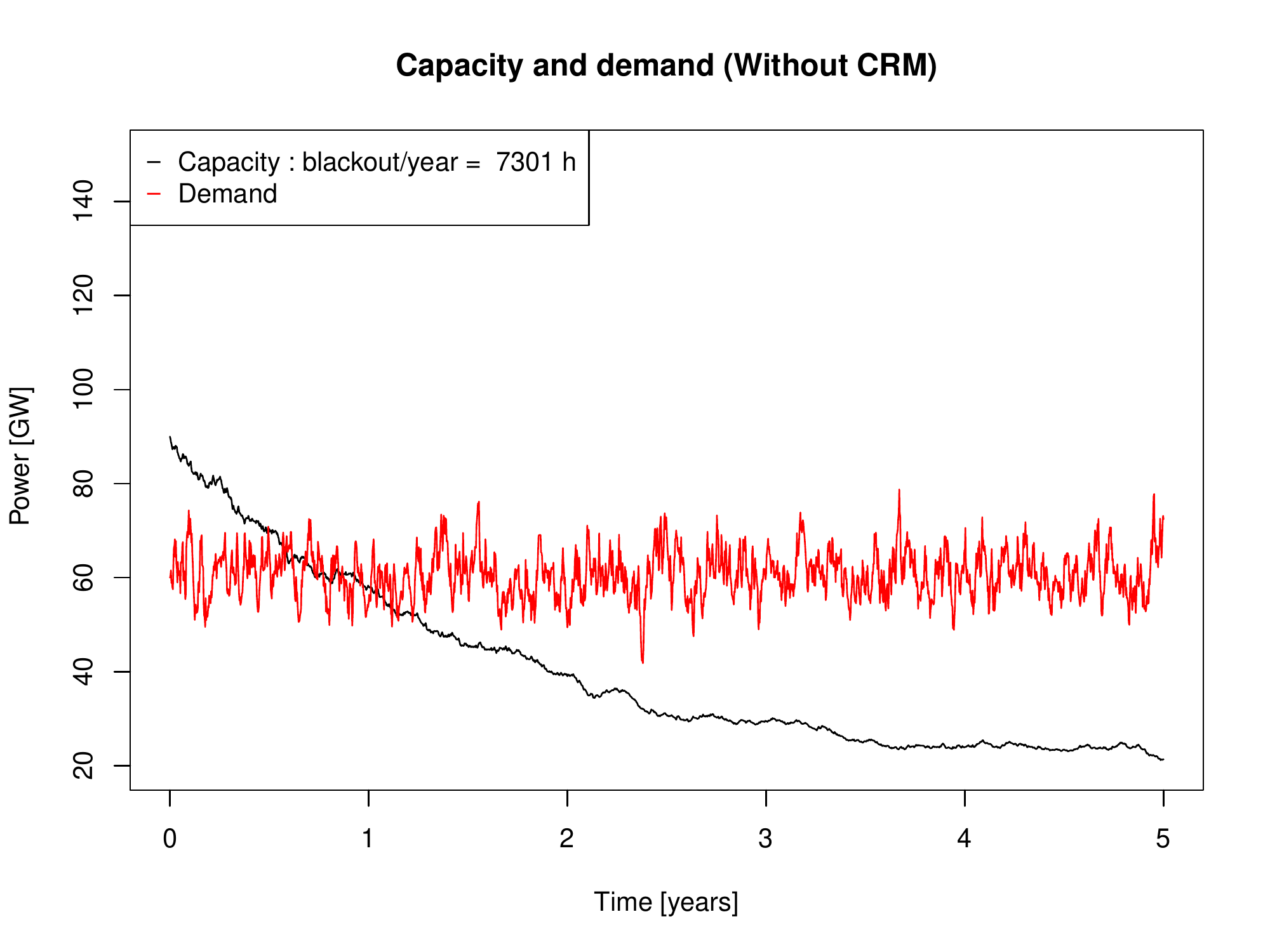}
			\captionsetup{width=.8\linewidth}
			\caption{Capacity and demand without a CRM. In red the demand with a constant average, and in black the capacity generation decreasing over time and suggesting consistent dismantling of powerplants and an increasing number of shortage hours per year.}
			\label{fig : violent_cap_et_dem_ag_sol}
		\end{subfigure}
		\begin{subfigure}{.4\textwidth}
			\centering
			\includegraphics[scale=0.4]{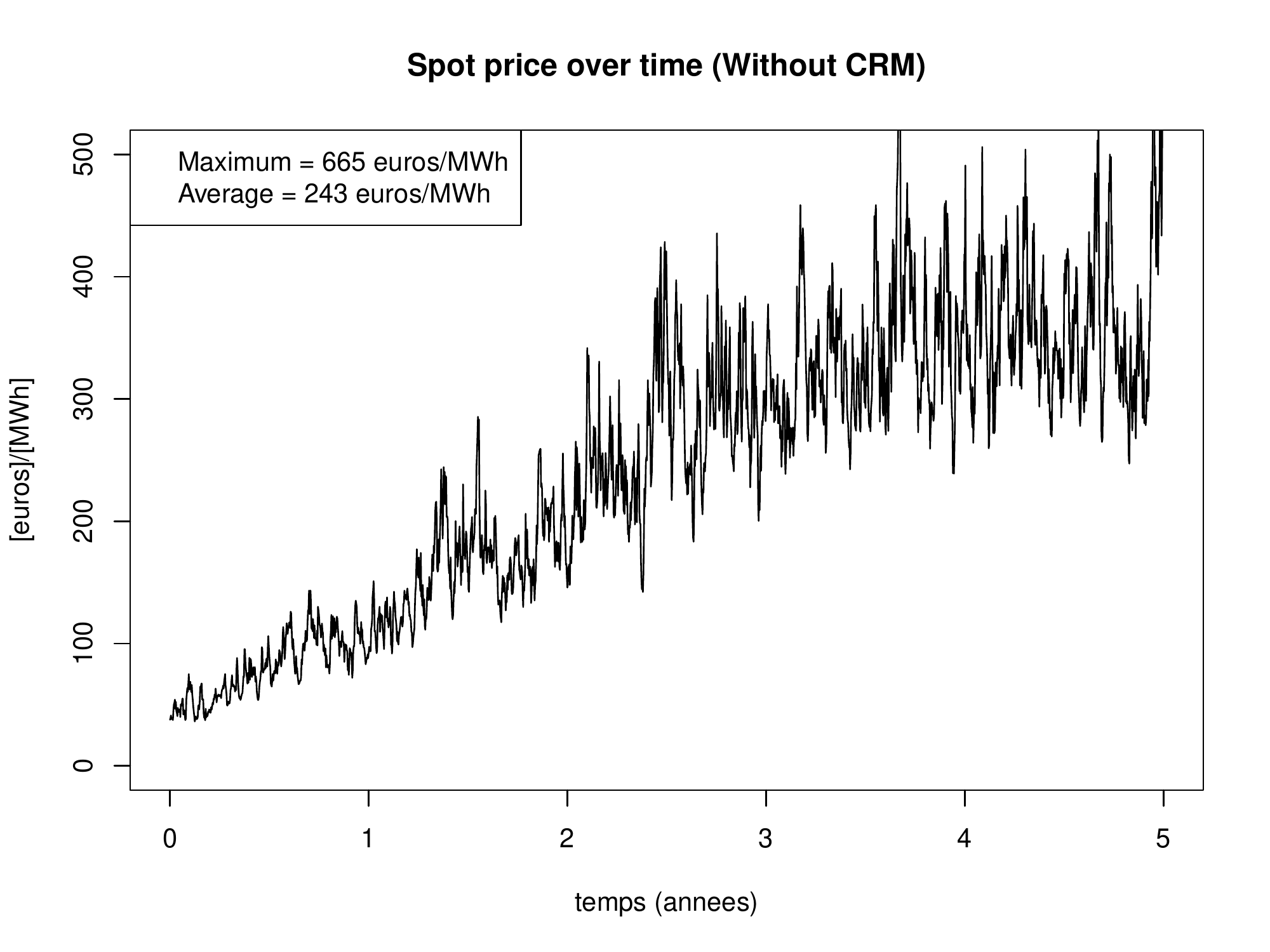}
			\captionsetup{width=.8\linewidth}
			\caption{Spot price without a CRM. The evolution of spot price is the reverse of the margin (capacity-demand), increasing over time and reaching a maximum of 665 euros/MWh near maturity and an average of 243 euros/MWh on the whole period.}
			\label{fig : violent_spot_price_Ag_sol}
		\end{subfigure}
		\captionsetup{width=.8\linewidth}
		\caption{Evolution over time of the state variables without a CRM. The capacity and demand (GW) and the corresponding spot price (Euros/MWh), for a scenario where producer controls the generation capacity, and his only income is the spot revenue.}
		\label{fig : state variables ag solo}
	\end{figure}
	
	The figures \ref{fig : violent_cap_et_dem_ag_sol} and \ref{fig : violent_spot_price_Ag_sol} illustrate that without a CRM, producer will be better off with low capacity and high spot revenues, and therefore if we aim at keeping a reasonable level of available capacity (which implies lower spot prices), it is necessary to provide him with a complementary compensation to replace his losses in spot revenues. We highlight the fact that the absence of a capacity payment would have much less drastic impacts in the real life than what we observe in our numerical simulations. This is due from one side to the regulation authorities which would not allow for such levels of shortage to occur, and from the other side because of the presence of multiple producers in competition who might decide to invest more --breaking the market power--, or even new actors (producers) willing to invest and enter the market in such favorable conditions (\ie with a spot price much higher than costs.)\\	
	
	The loss of spot revenues incurred by producer when keeping high capacity levels is partly captured by his participation constraint, since $\Rc$ also accounts for the change in construction, maintenance and production costs. In fact, the more profitable the system without a CRM to the producer, the higher $\Rc$, and the more inciting the contract (CRM) needs to be. In our setting $\Rc=2.8\mbox{\euro/MWh}$, and is decreasing in $\sigma^C$ and increasing in $\sigma^D$. When the volatility of capacity is high, the efforts of producer have less and less impact on $X^C$ -and therefore on the system- and his utility (or its cash equivalent $\Rc$) is lower. On the other side, whenever the volatility of demand is high, the probability of shortage increases and when capacity margin becomes low this drives the spot prices up in an amplified manner which gives the producer more utility.\\ 
	
	Similar to the volatility of capacity $\sigma^C$, the parameters $\kappa_2$ and $x_0^C$ have the same impact on $\Rc$; the cost of control (in this case, the cost of shutting down powerplants) becomes higher with $\kappa_2$ which lowers producer's utility, while a positive change in the initial value $x_0^C$ increases the capacity margin and decreases the spot prices and $\Rc$ as a result. Finally, and obviously, higher spot prices (because of higher spot levels $\beta_0$) increase $\Rc$.\\
	
	Note that it is not possible to infer from the first column of table \ref{table : Comparison between policies} the participation constraint $\Rc=2.8\mbox{\euro/MWh}$, since the utility function is concave we have $\Rc= U_A^{-1}\brak{\E \edg{U_A \brak{\text{Spot}+ \text{Capacity payment} -  \text{Total Costs}}}}$ which is lower than what can be read from the table and which corresponds to $\E \edg{\text{Spot} }+ \E \edg{\text{Capacity payment}} -  \E \edg{\text{Total Costs}} $.
	\subsubsection{Analysis of the system evolution under the optimal policy}
    Coming back to table \ref{table : Comparison between policies}, we can see that introducing the CRM drastically improves the security of the system; (an average of 2 hours shortage per year--\emph{respecting the LoLE constraint}-- compared to 178 hours per year when there is no capacity adjustment and 6165 hours per year when producer has market power.) Remark that this also reduces consumer's payments: it is less costly for the consumer to pay for capacity and the spot prices (which is in average rather low because the system margin is high) than paying only the spot prices ``without CRM'' where spot prices are very high.\\
   
    Observe also that when comparing the system with a CRM and without capacity adjustment, we see that the average margin is positive in both cases and quite high (32 GW and 28 GW) so one would expect that these two settings would be quite similar. However, we see that we obtain a substantial gain in the average number of shortage hours per year when following the dictated policy (with 32 GW capacity margin), going from 178 hours per year to only 2 hours per year. This owes to the design of the contract in the CRM taking into account the magnitude of uncertainties and other characteristics of the system.\\
	
	In order to better interpret producer's optimal policy, we select and analyze two of the $5000$ simulated scenarios; a severe scenario --the one with the highest number of shortage hours over the period of simulation, selected a posteriori-- and a Favorable scenario. Table \ref{table : Different scenarios outcomes} provides the outcomes of these scenarios compared with the average scenario.	
	
	\begin{table}[H]
		
		\centering
		\resizebox{0.8\linewidth}{!}{
			\begin{tabular}{lrrr}

\toprule
  & Average scenario & Favorable scenario & Severe scenario\\
\midrule
\addlinespace[-1em]
\multicolumn{4}{l}{\textbf{}}\\
\hspace{1em}Shortage hours per year [Hours] & 2.2 & 0.0 & 113.9\\
\hspace{1em}Average Spot price [euro/MWh] & 37.6 & 31.9 & 59.4\\
\hspace{1em}Average Margin [GW] & 32.4 & 36.0 & 17.3\\
\hspace{1em}Spot revenues [euro/MWh] & 37.8 & 32.1 & 59.7\\
\hspace{1em}Capacity payment [euro/MWh] & 12.3 & 57.9 & 21.7\\
\hspace{1em}Spot + Capacity payment[euro/MWh] & 50.1 & 90.0 & 81.4\\
\hspace{1em}\hspace{1em}\hspace{1em}Participation constraint [euro/MWh] & 2.8 & 2.7 & 2.7\\
\hspace{1em}\hspace{1em}\hspace{1em}Risk shared [euro/MWh] & -0.2 & 43.4 & -40.7\\
\hspace{1em}\hspace{1em}\hspace{1em}Risk compensation [euro/MWh] & 14.9 & 12.7 & 52.1\\
\hspace{1em}\hspace{1em}\hspace{1em}Total costs [euro/MWh] & 32.7 & 31.2 & 67.3\\
\hspace{1em}\hspace{1em}\hspace{1em}\hspace{1em}\hspace{1em}Construction and dismantling [euro/MWh] & 1.8 & 0.0 & 38.6\\
\hspace{1em}\hspace{1em}\hspace{1em}\hspace{1em}\hspace{1em}Maintenance [euro/MWh] & 13.2 & 13.6 & 11.1\\
\hspace{1em}\hspace{1em}\hspace{1em}\hspace{1em}\hspace{1em}Production [euro/MWh] & 17.6 & 17.6 & 17.6\\
\bottomrule
		\end{tabular}}
	\captionsetup{width=.8\linewidth}
	\caption{Comparison between different scenarios with CRM. This table regroups the results of the diffusion of $N=5000$ scenarios with the CRM and following the recommended effort to assess the extreme outcomes that might occur. The first column represents the average over all scenarios for reference, the second column provides the results for one ``favorable'' scenario with 0 shortage hours, and the third column provides the results of a ``severe'' scenario; the one with the maximal number of shortage hours in our simulations.}
	\label{table : Different scenarios outcomes}
	\end{table}
	
	We plot first the evolution of state variables in figure \ref{fig : capa et demande}. We can see in red the evolution of demand, and in black the available capacity. The demand process is by construction a mean-reverting process. However, the capacity is a geometric Brownian motion. So the capacity has a priori no reason to exhibit a mean-reverting behavior which is nevertheless observed in \ref{fig : capa et demande} on the right figure (b). This mean-reversion can be explained by the effort rate $\hat{\alpha}$ which readjusts the capacity depending on the randomness, and the level of security fixed by consumers' preferences. This readjustment can also be seen in the difference between construction costs in table \ref{table : Different scenarios outcomes}; the severe scenario having the highest cost suggesting a policy with intensive construction.   
	
	\begin{figure}[H]
		\centering
		\begin{subfigure}{.45\textwidth}
			\centering
			\includegraphics[scale=0.4]{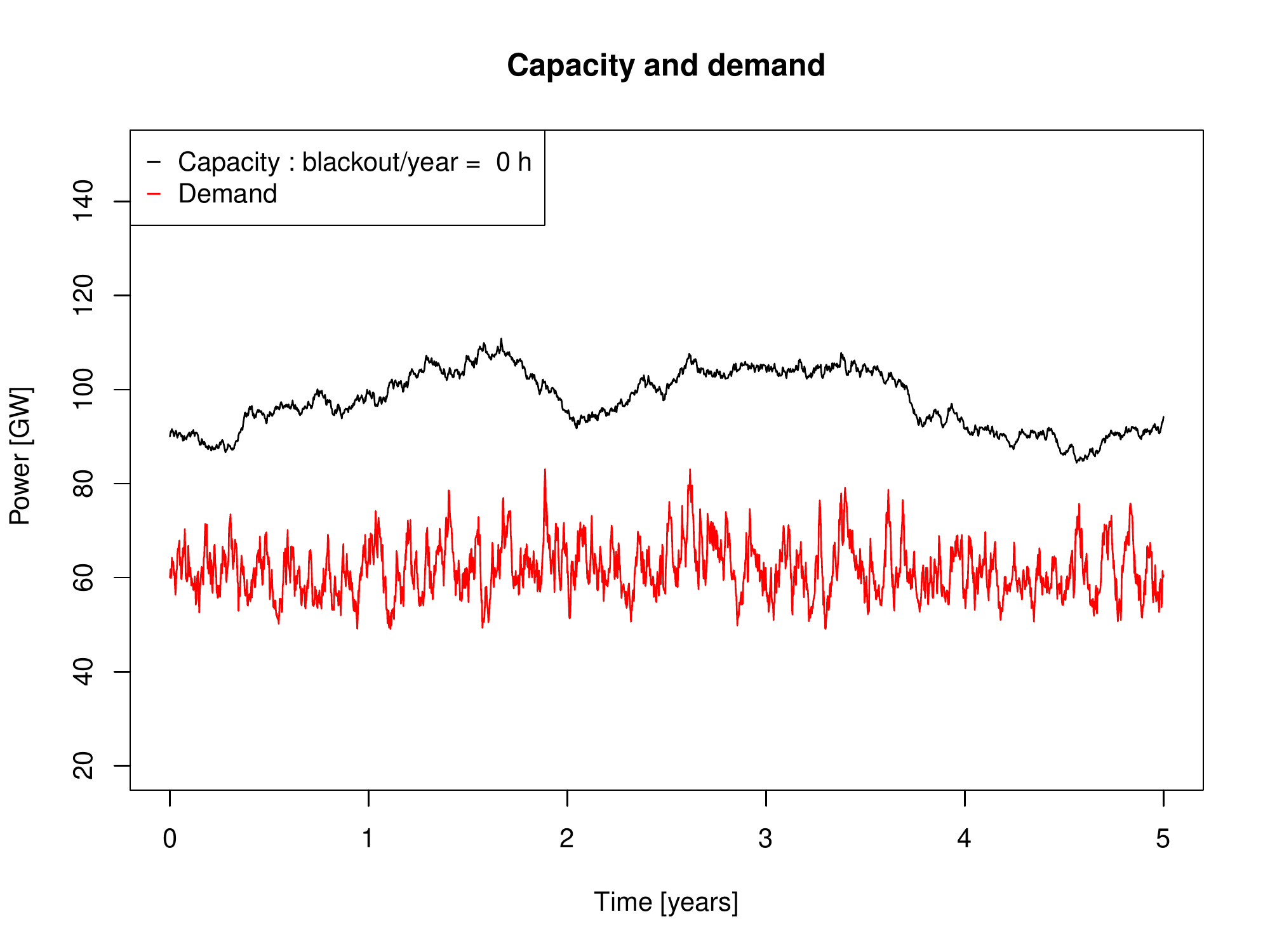}
			\caption{Favorable scenario}
		\end{subfigure}
		\begin{subfigure}{.45\textwidth}
			\centering
			\includegraphics[scale=0.4]{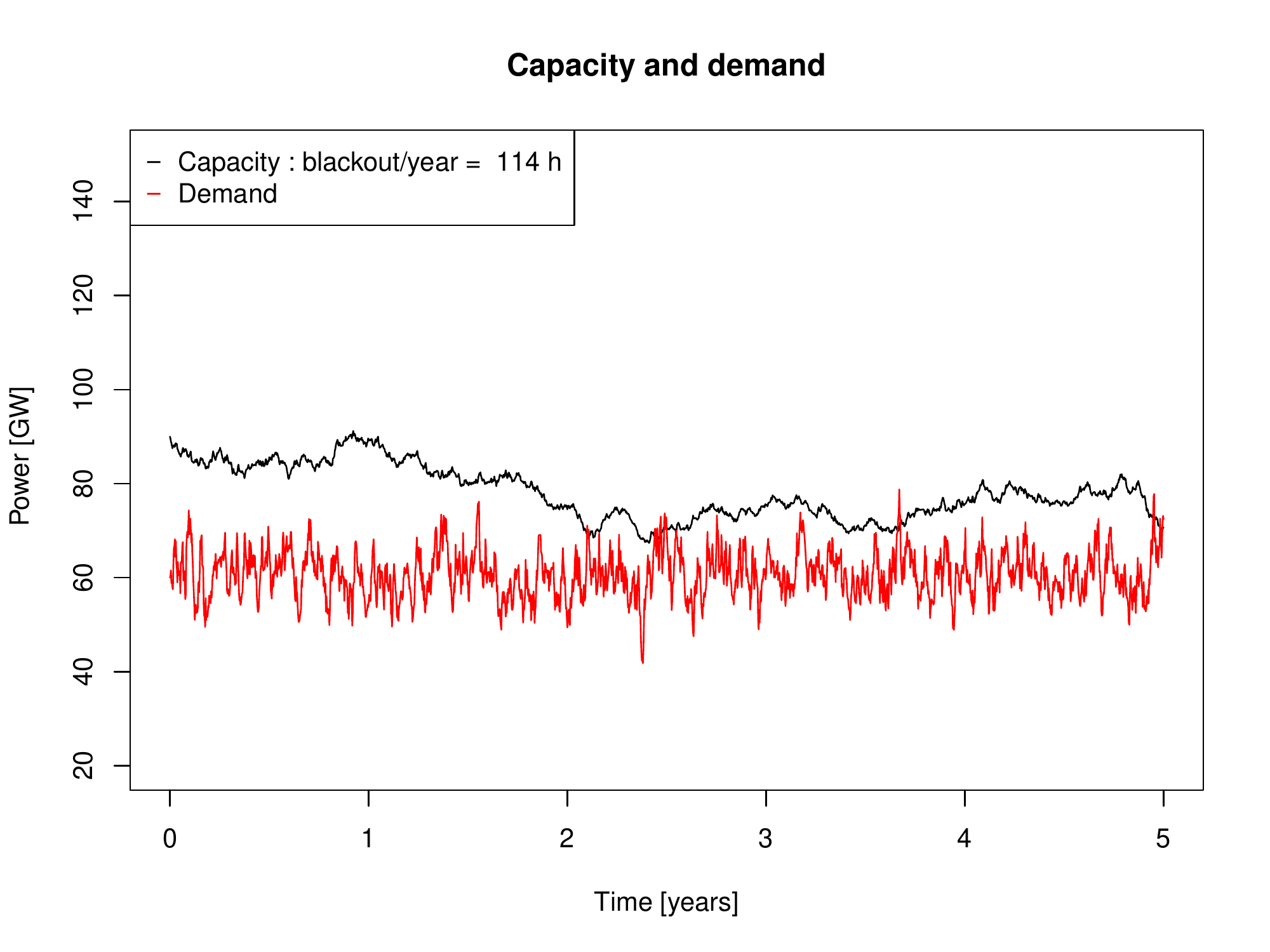}
			\caption{Severe scenario}
		\end{subfigure}
		\captionsetup{width=.8\linewidth}
		\caption{Evolution of capacity and demand under the recommended effort $\hat{\a}$. We compare the evolution of state variables for the ``favorable'' and ``severe'' scenarios. In the favorable scenario (left), the capacity (black) remains above the demand (red) with a high margin, as opposed to the severe scenario (right) where capacity is drawn down from external uncertainties and goes beyond the demand process a few times corresponding to the occurrences of shortage.}
		\label{fig : capa et demande}
	\end{figure}
	
	Figure \ref{fig:capa_non_optimal} provides an interpretation of producer's optimal control, by comparing two policies
	. In black, we see the capacity evolution obtained with producer's optimal policy, and in red the capacity without adjustment, \ie $\a_t=0$ for $t\in\edg{0,T}$. This is interpreted as a comparison between our model and a the ``No adjustment'' policy model in which producer sets an initial capacity margin (30 GW in this case) and the system is then only impacted by the uncertainties. 
	
	\begin{figure}[H]
		\centering
		\begin{subfigure}{.45\textwidth}
			\centering
			\includegraphics[scale=0.4]{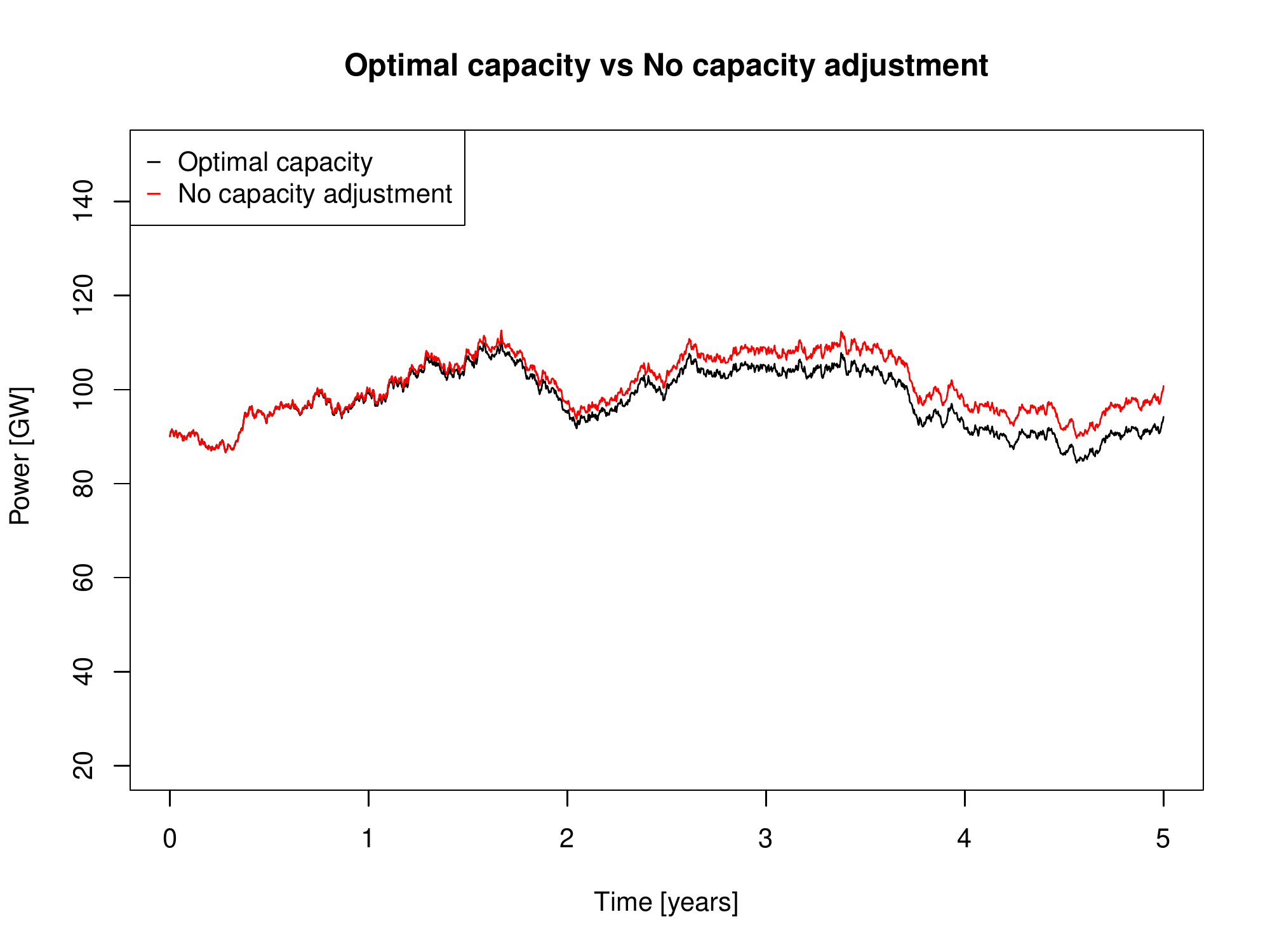}
			\caption{Favorable scenario}
		\end{subfigure}%
		\begin{subfigure}{.45\textwidth}
			\centering
			\includegraphics[scale=0.4]{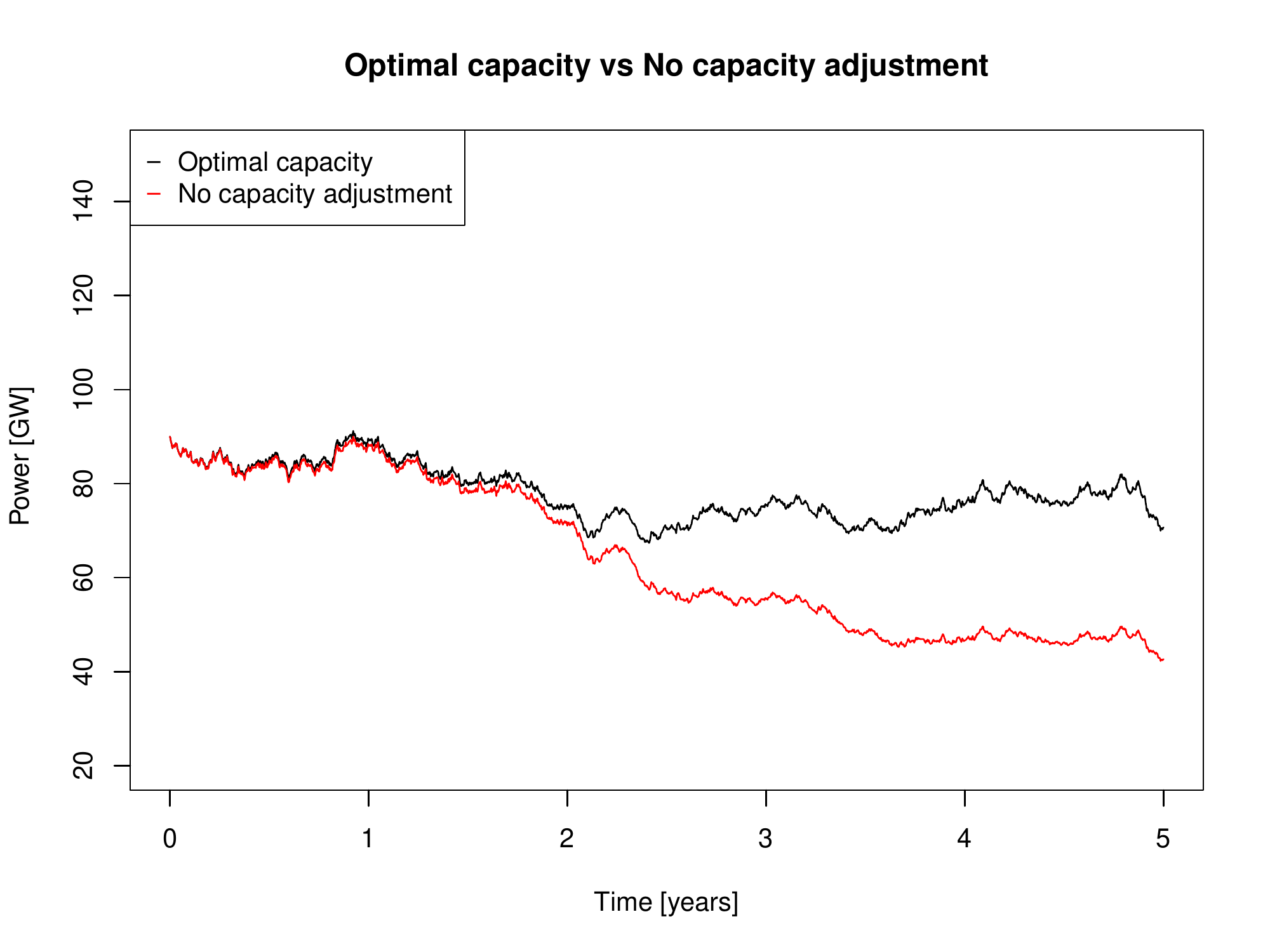}
			\caption{Severe scenario}
		\end{subfigure}
		\captionsetup{width=.8\linewidth}
		\caption{Comparison between producer's optimal policy (the recommended effort) and a ``no adjustment'' policy for two different scenarios. The black line represents the controlled capacity for the favorable (left) and severe (right) scenarios. The externalities in these scenarios are made explicit with the uncontrolled capacity processes (red), and the impact of the control can be seen from the difference between the black line and the red line. As expected, the favorable scenario has positive externalities and the control is slightly negative, while the severe scenario has negative externalities and a positive control.}
		\label{fig:capa_non_optimal}
	\end{figure}
	
	The figure shows that the Favorable scenario is a scenario where the capacities naturally experience favorable outcomes. For example, the inputs of hydro-powerplants, the load factor of wind and photovoltaic should have been very high, or no major failure of power plant should have been observed. On the contrary, the severe scenario is a scenario where capacities uncertainties are very unfavorable (strong drop of capacity after year 2, which can be seen in the ``no capacity adjustment policy''). We can see that optimal policy absorbs this shock, at least partially because of the costs of construction preventing producer from restoring a higher capacity margin.\\
	
	A heuristic observation of consumer's value function suggests that whenever the capacity margin is tight (when $x^C - x^D$ is small), there is a high risk of shortage and a low satisfaction from consumption. We can guess then that $\partial_{x^C}u\gg 0$ as a higher capacity level will make consumer's situation better. It makes sense then for consumer to recommend a positive control $\hat{\a}$ as suggested from \eqref{eq : forme approximative de alpha}, which is confirmed numerically as $\hat{\a}$ takes higher values, and pushes the capacity process up. This represents a typical situation where we can see the implicit dependence of $\hat{\alpha}$ on $x^D$ (in particular on $x^C-x^D$) as mentioned in Remark \ref{remark : summary all} (\rmi\!\! and \rmii\!\!\!). The same pattern is observed at year 2 and 2.5 in the severe scenario where the producer invests to counteract a negative shock in capacity. However, between year 3 and 4.5, as it becomes very expensive to keep a positive capacity margin, the optimal control does not follow the shock and a serie of shortages occurs. This implies that starting from some threshold, consumers are willing to accept a shortage instead of paying a very high price to avoid it.\\
	
	Finally, remark that from consumer's perspective, the total payment (capacity remuneration + spot) is higher than average in the extreme scenarios, whether favorable or severe. Indeed, when the scenario is severe, the spot prices are high, and the capacity compensation is also high because of construction costs, and so consumer has to pay a lot for both. On the other hand, when the scenario is favorable, the spot prices are low but the consumer still needs to incentivize the producer with a high capacity compensation and share the positive risks (the ``reward for good luck'').

	\subsection{Analysis of the optimal contract}
	\subsubsection{Decomposition of capacity payment}
	
	To understand how the contract is designed, we use the decomposition \eqref{eq : decomposition en lettres} suggested in section \ref{ section : contract decomposition}, recalled below: 	
	\begin{align}
	\mbox{Capacity remuneration ($\xi^{\star}$)}+\mbox{Spot compensation}&=\Rc+\mbox{{Producer's costs}}+\mbox{{Risk shared}}+\mbox{Risk compensation}.\tag{\ref{eq : decomposition en lettres}}
	\end{align}

	This decomposition is represented for the favorable and severe scenarios in figures \ref{fig : best_xi_rand_Dec_xi_costs} and \ref{fig : violent_Dec_xi_costs}. In each of these figures, the first bar represents the total compensation (the left hand side terms of equality \eqref{eq : decomposition en lettres}), and the second bar corresponds to the right hand side. Remark that in the favorable scenario all the components are positive, and add up to the total compensation, while in the severe scenario, there is a negative component which is the risk shared. Nevertheless in this case also the algebraic sum of the components is equal to the total compensation.
	
	\begin{figure}[H]
		\centering
		\begin{subfigure}{.5\textwidth}
			\centering
			\includegraphics[scale=0.4]{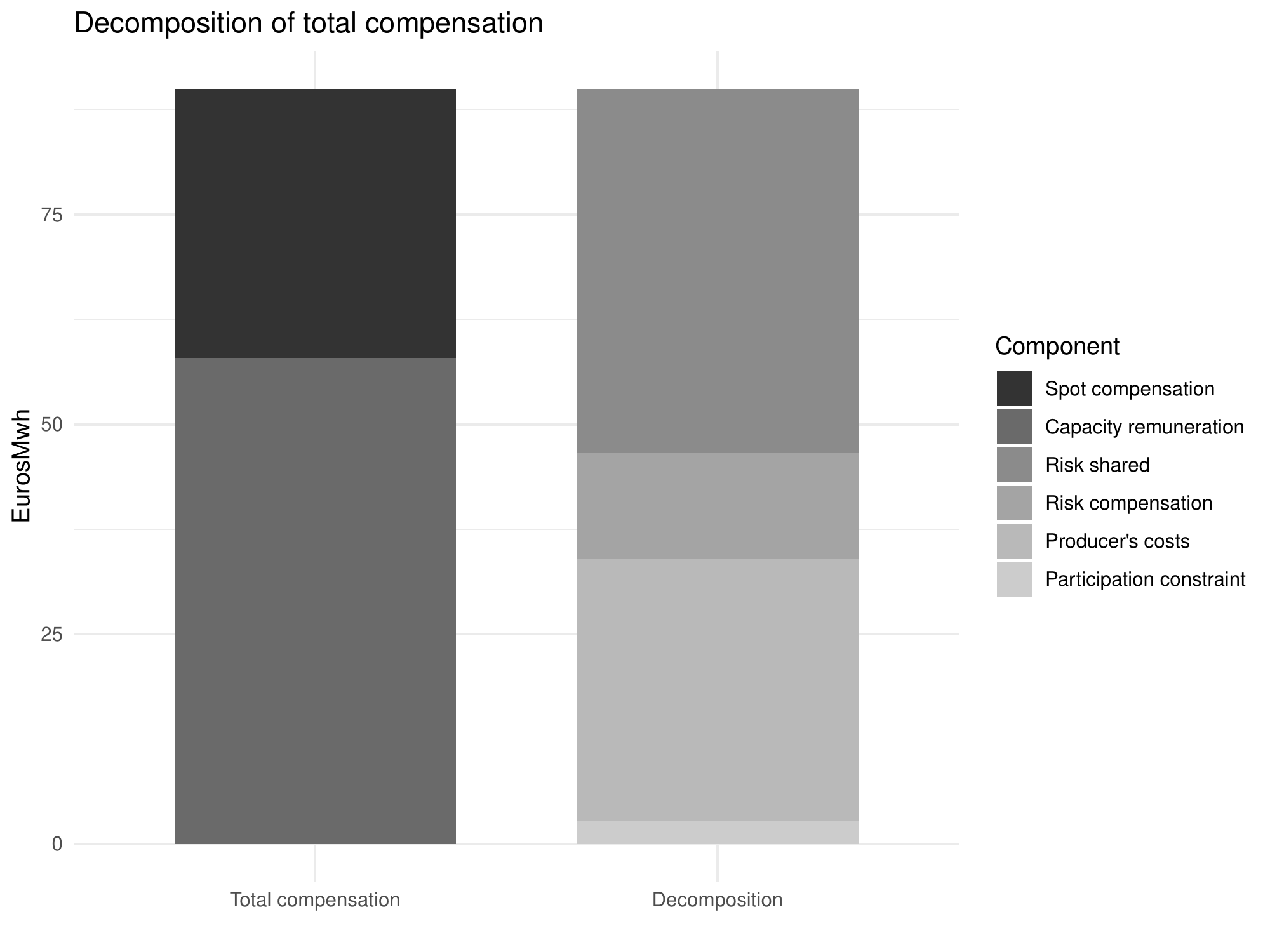}
			\captionsetup{width=.8\linewidth}
			\caption{Favorable scenario. The left bar aggregates the total compensation : (spot revenues + capacity remuneration), and is equal to the right bar (participation constraint + producer's costs + risk shared + risk compensation) where we can see that all these terms are positive for the favorable scenario.}
			\label{fig : best_xi_rand_Dec_xi_costs}
		\end{subfigure}%
		\begin{subfigure}{.5\textwidth}
			\centering
			\includegraphics[scale=0.4]{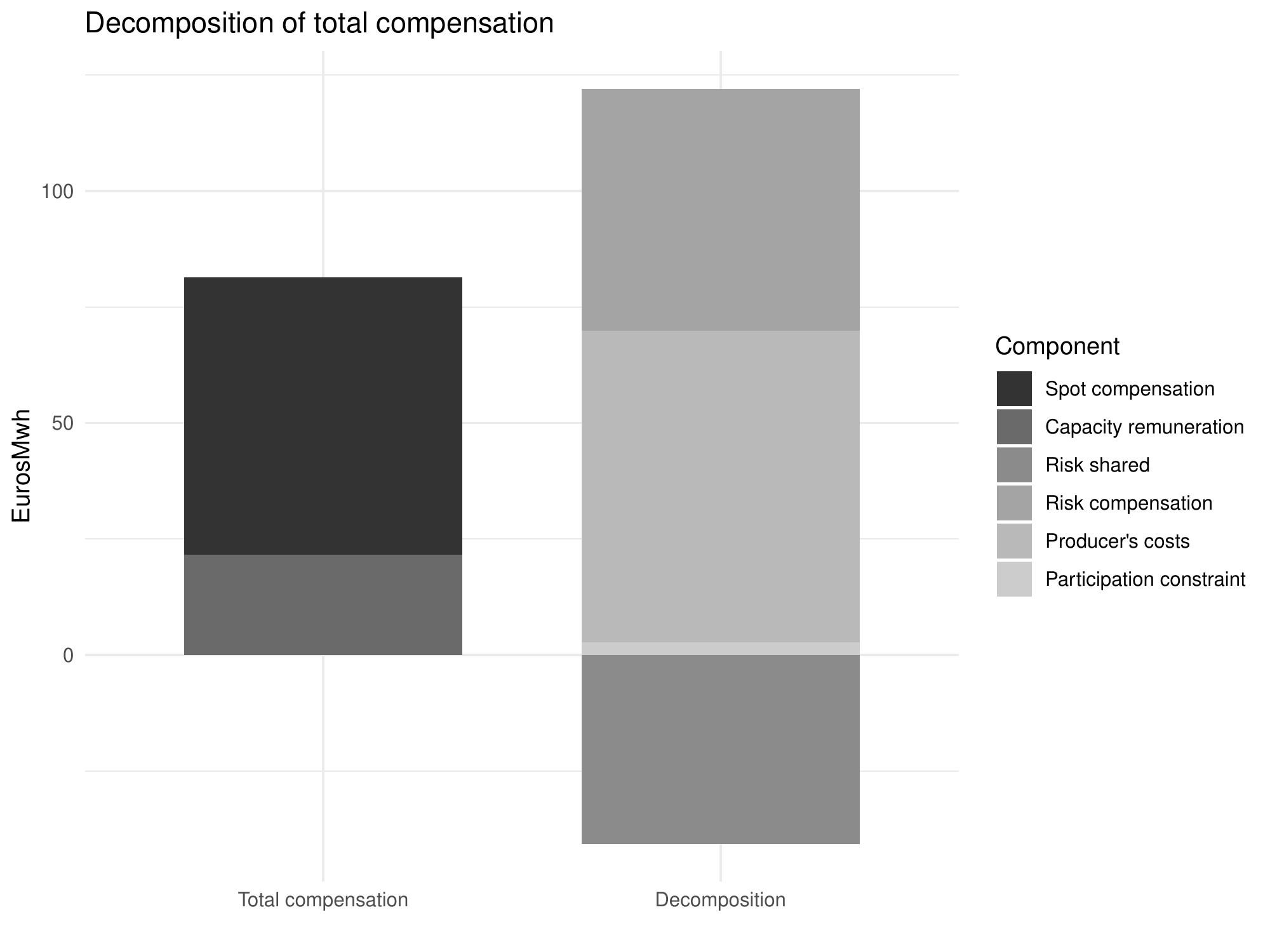}
			\captionsetup{width=.8\linewidth}
			\caption{Severe scenario. The left bar is the sum (spot revenues + capacity remuneration). The right bar has a part above zero representing the sum (participation constraint + producer's costs+risk compensation), and a negative ``risk shared'' component because of the unfavorable uncertainties.} 			
			\label{fig : violent_Dec_xi_costs}
		\end{subfigure}
		\captionsetup{width=.8\linewidth}
		\caption{Decomposition of $\xi^{\star}$ for the favorable and severe scenarios highlighting the equality \mbox{(spot + capacity) = (participation constraint + producer's costs + risk shared + risk compensation).}}
		\label{fig:decomposition contract costs}
	\end{figure}
	
%
	In the severe scenario --figure \ref{fig : violent_Dec_xi_costs}-- the realized randomness is very unprofitable to the system. Our contract automatically shares this negative randomness with the producer (negative ``risk shared''). On the contrary, in the favorable scenario --figure \ref{fig : best_xi_rand_Dec_xi_costs}-- where capacity outcomes are naturally high and profitable, this positive randomness is also shared with the producer but positively.\\
	
	In addition, the contract accounts for the costs needed to implement the optimal policy, the remuneration from spot and the risk compensation. Remark that, as expected, the risk compensation is positive in both cases, even when the shared randomness is positive.  This helps to offset the impact of the risk shared, for example in the severe scenario the negative risk shared (-40 euros/MWh) is completely canceled by the risk compensation (52 euros/MWh), see Table \ref{table : Different scenarios outcomes}.\\

	In the severe scenario the compensation for costs (the difference between blue and green) is quite high (this is mainly linked to the high costs needed to follow the optimal policy, consisting in investing a lot), whereas under the favorable scenario, this part is much limited (investment to be made are small). However, the remuneration obtained from the spot market in the favorable scenario (with a high capacity margin) is much less than in the severe scenario (with a low capacity margin). This leads to a low capacity payment under the severe scenario compared to the favorable scenario.
	
	\subsubsection{Link between capacity payment and spot compensation}\label{subsection : Link between capacity payment and spot compensation}
	 Table \ref{table : Comparison between policies ratios} provides another comparison between the system without a CRM, with a CRM and with no adjustment, but this time with regard to the occurrences of missing money (when total costs are more than spot revenues) and the scenarios with negative net revenues for producer, along with the role played by spot in total compensation in the case of a CRM, and the percentage of scenarios with a negative capacity remuneration.

	\begin{table}[H]
		
		\centering
		\resizebox{0.8\linewidth}{!}{
			\begin{tabular}{lrrr}
				\toprule
				& Without CRM & With a CRM & No adjustment\\
\midrule
Spot revenues/Total revenues [\%] & 100 & 70 & 100\\
Scenarios with missing money [\%] & 2 & 28 & 26\\
Scenarios with negative capacity remuneration [\%] & NA & 18 & NA\\
Scenarios with negative total compensation [\%] & NA & 0 & NA\\
Scenarios with negative net revenues [\%] & 2 & 5 & 26\\
\bottomrule
		\end{tabular}}
	\captionsetup{width=.8\linewidth}
	\caption{Comparison between policies outcomes for producer with and without CRM. Each column of this table provides the average ratio (Spot revenues/Total revenues) on a sample of $N=5000$ simulated scenarios, together with the percentages of scenarios with missing money (total costs>spot revenues), the scenarios with (capacity remuneration<0) and the scenarios with (spot revenue+capacity remuneration)<0, and those with negative net revenues (total compensation - total costs)<0. The first column represents the scenarios without a CRM, where producer controls the capacity and his only income is the spot revenue. The second column represents our CRM, with scenarios following the recommended effort and a capacity payment. The third column provides the results for scenarios generated without any control on capacity and without capacity payment.}
	\label{table : Comparison between policies ratios}
	\end{table}
	
	 Regarding the system with a CRM, we obtain an average number of shortage hours per year less than three hours, which satisfies the LoLE constraint. The revenue provided by the capacity payment is about $30\%$ of the total compensation (compared to $70\%$ for spot market).\\

	In our simulations $\E^{\P^{\a^{\star}}}\edg{\frac{S_T+\xi^{\star}}{S_T}}\in\edg{1.42,1.43}$ with a $95\%$ confidence level and a standard deviation of $0.53$. This corresponds roughly to a partition of total revenues into $70\%$ from spot market and $30\%$ from capacity compensation. However, as we can see in figure \ref{fig : rat_tot_sur_spot}, the distribution of $\frac{S_T+\xi^{\star}}{S_T}$ can take values less than 1  (even negative theoretically, but not observed in practice (cf. Table \ref{table : Comparison between policies ratios})) when $\xi^{\star}$ is negative. Whenever they occur, the negative capacity prices could be interpreted as a reimbursement from producers when their revenues from the spot market are high, similar to reliability options used in Italy and Ireland, and which by definition require such a money transfer when the spot price exceeds a certain level \cite{bhagwat2019reliability}.\\
	
	\begin{figure}[H]
		\centering
		\includegraphics[width=.4\textwidth]{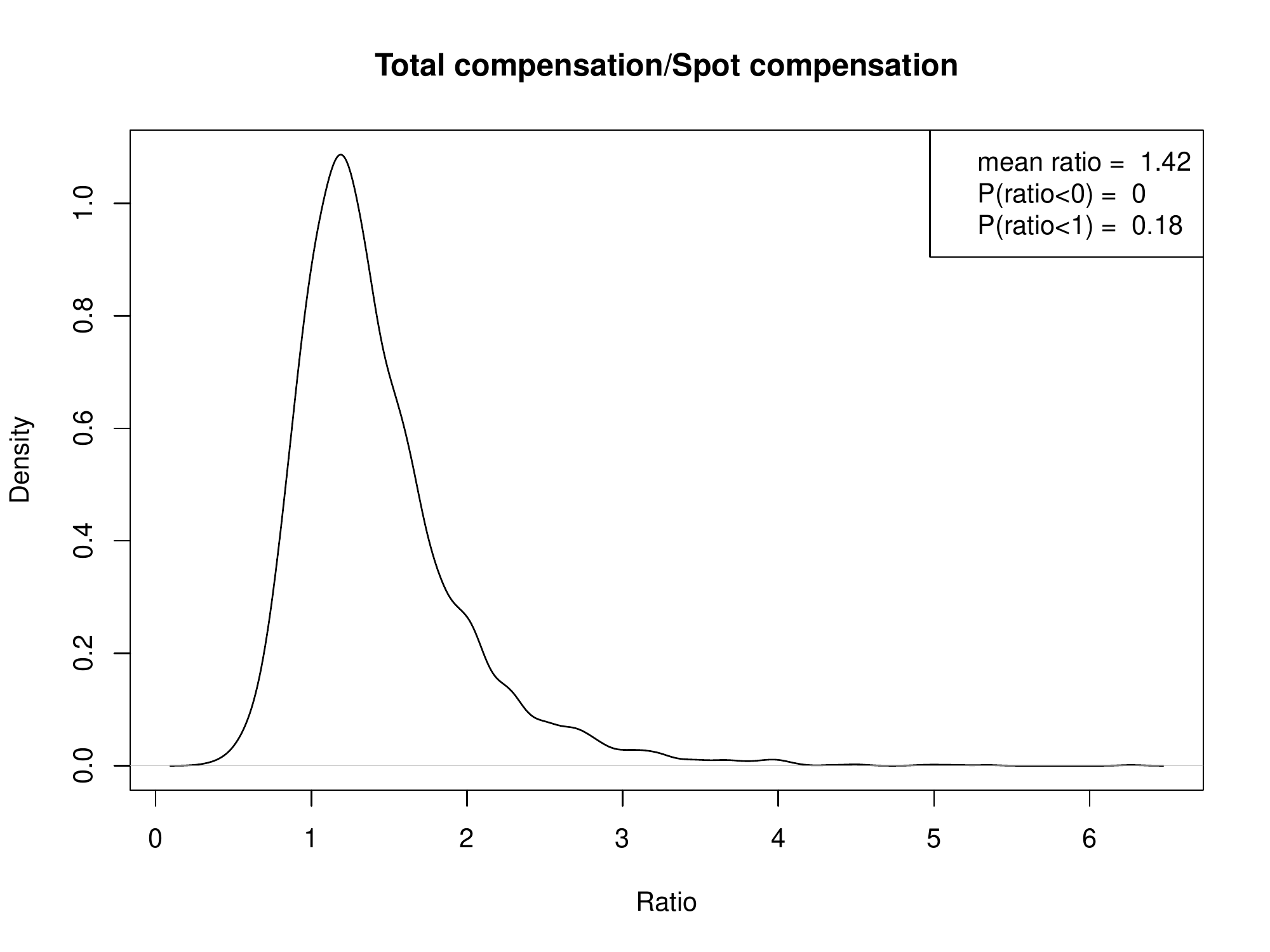}
		\captionsetup{width=.8\linewidth}
		\caption{Distribution of $\frac{S_T+\xi^{\star}}{S_T}$ with a CRM. The ratio (total compensation/spot revenues) is observed on $N=5000$ scenarios, with an average of $1.42$, and a positive probability ($0.18$) of being lower than $1$, \ie of having (total compensation < spot revenues).}
		\label{fig : rat_tot_sur_spot}
	\end{figure}
	
	We investigate more this ratio in figure \ref{fig:decomposition contract spot} by decomposing the total compensation into spot revenues and capacity remuneration for the favorable and severe scenarios, together with another extreme scenario.
	
	\begin{figure}[H]
		\centering
		\begin{subfigure}{.5\textwidth}
			\centering
			\includegraphics[scale=0.4]{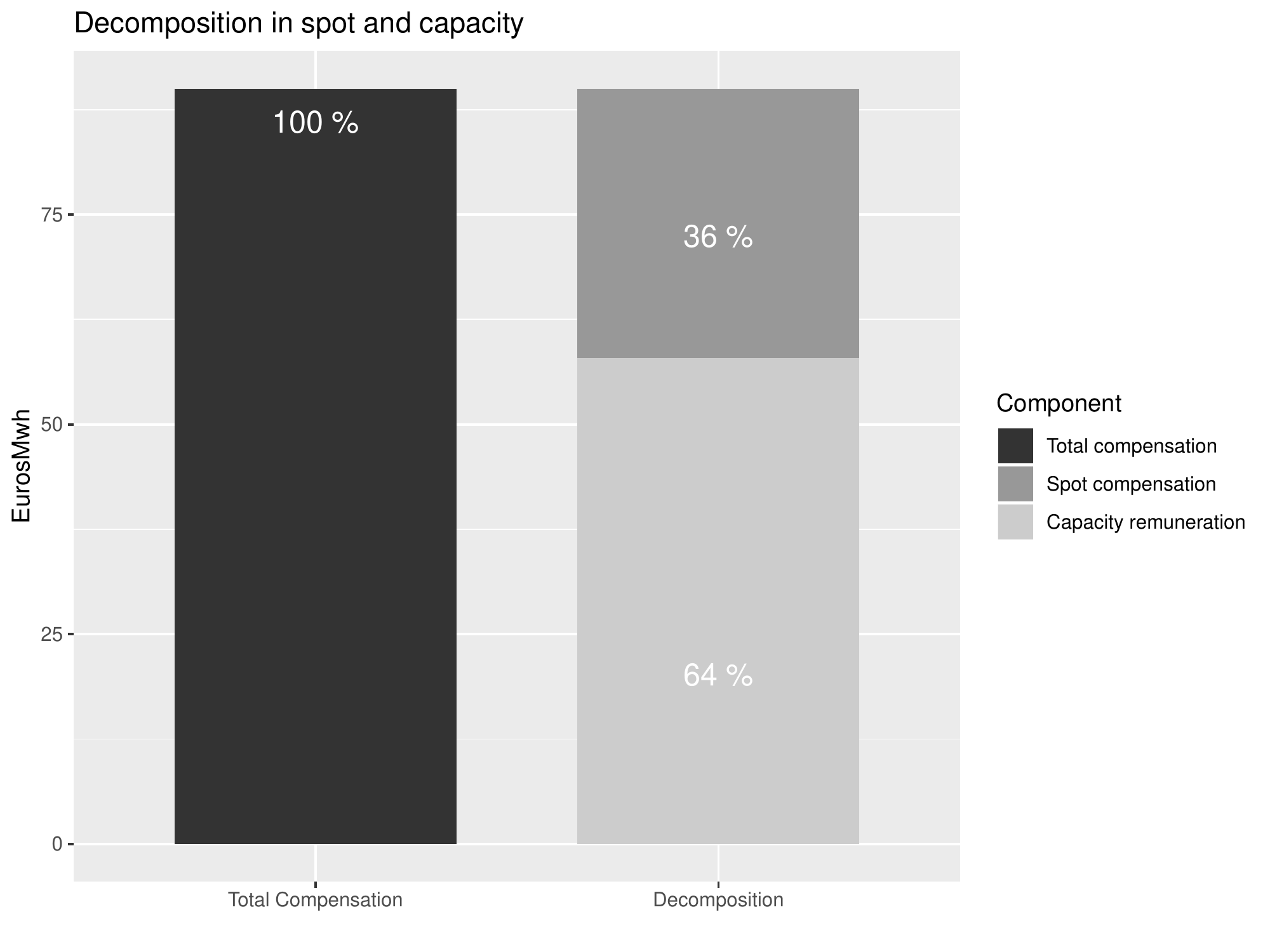}
			\captionsetup{width=.7\linewidth}
			\caption{Favorable scenario. This scenario corresponds to high capacity levels, with low spot prices, so the capacity remuneration revenues have the biggest share in the total compensation.}
			\label{fig:decomposition contract spot random}
		\end{subfigure}%
		\begin{subfigure}{.5\textwidth}
			\centering
			\includegraphics[scale=0.4]{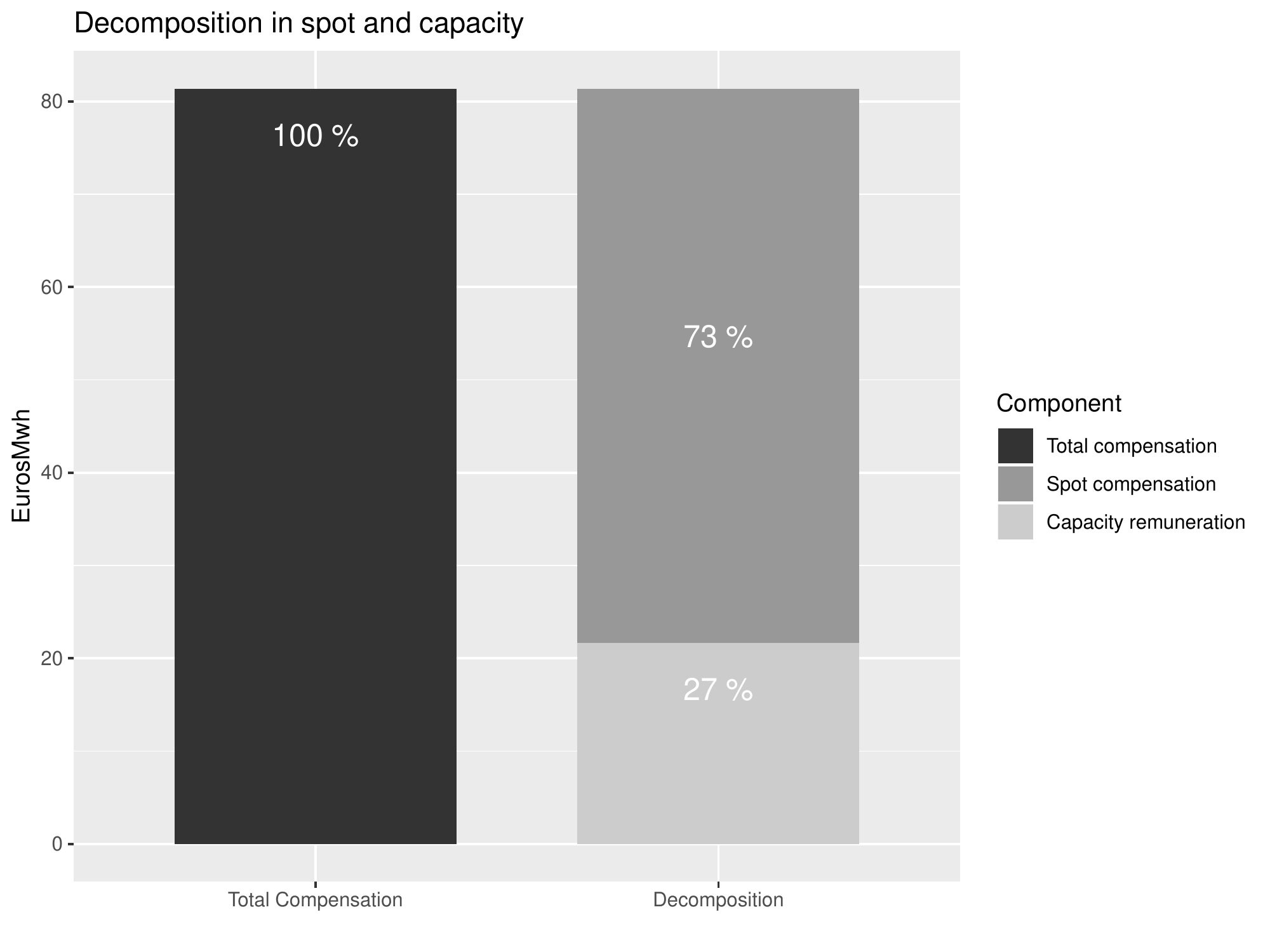} 	
			\captionsetup{width=.7\linewidth}
			\caption{Severe scenario. This scenario is characterized by low capacities, and tight generation margins, which increase the spot revenues, whereas the capacity remuneration is low because of the (negative) risk shared.}
			\label{fig:decomposition contract spot random severe}
		\end{subfigure}
		\begin{subfigure}{.5\textwidth}
			\centering
			\includegraphics[scale=0.4]{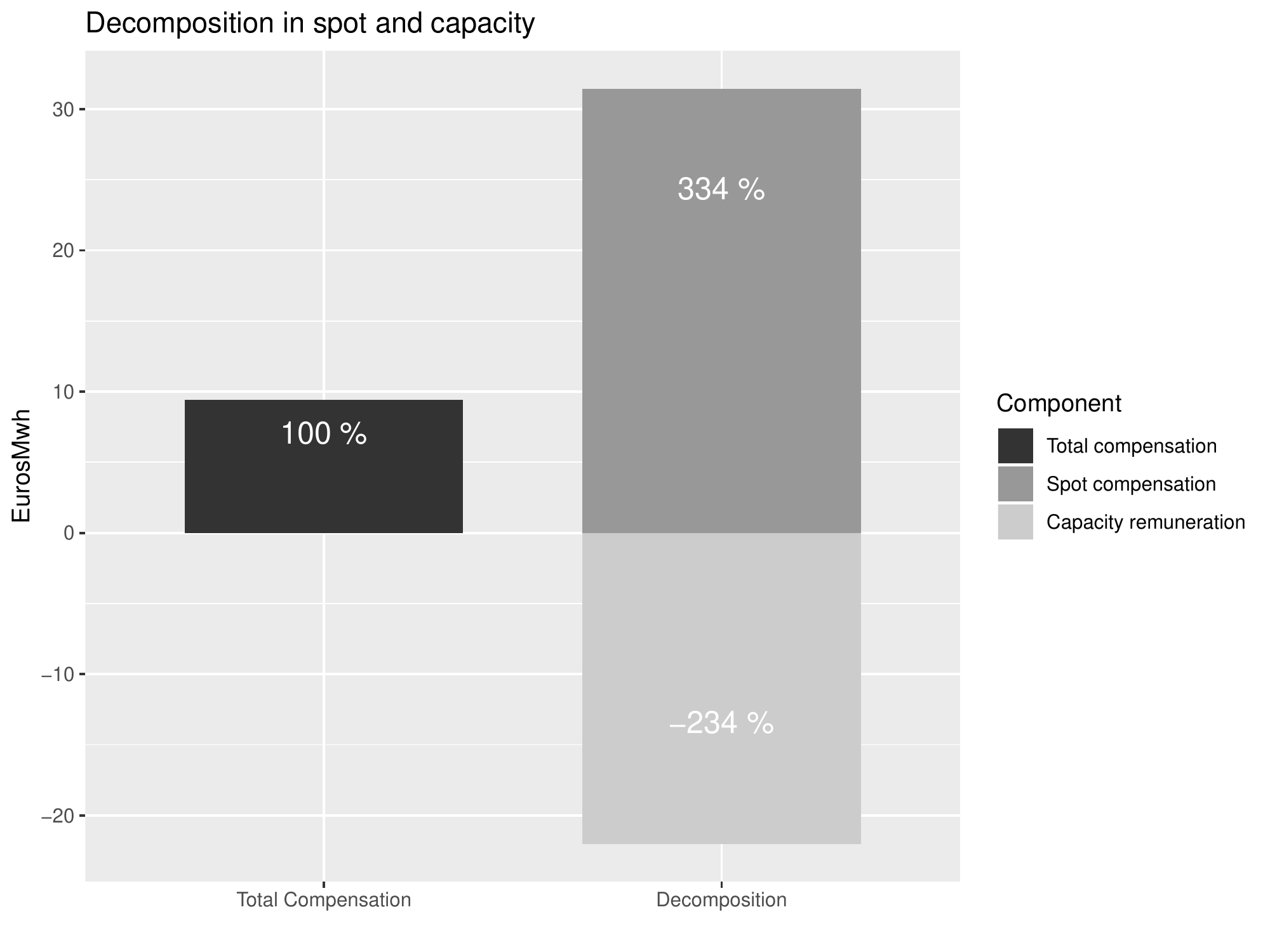}
			\captionsetup{width=.8\linewidth}
			\caption{Scenario with the lowest $\frac{S_T+\xi^{\star}}{S_T}$ ratio. This particular extreme scenario shows the ``punishement for bad luck'' where the capacity remuneration is negative (because of the risk shared term), which considerably decreases the total compensation compared to the spot compensation.}
			\label{fig:decomposition contract spot extreme highest}
		\end{subfigure}
		\captionsetup{width=.8\linewidth}
		\caption{Decomposition of $\xi^{\star}$ in terms of spot and capacity payment for different scenarios. In the three figures, the left bar represents the total compensation, and the right bar represents the decomposition total compensation = (spot compensation + capacity remuneration).}
		\label{fig:decomposition contract spot}
	\end{figure}

	In the favorable scenario previously analyzed, figure \ref{fig:decomposition contract spot random} shows that the capacity payment is going to be more active than in the severe scenario \ref{fig:decomposition contract spot random severe} to complement spot remuneration. For the favorable scenario, the capacity payment represents 64\% of the total remuneration of the producers, compared to 27\% in the severe scenario. This is explained by the capacity margin, which dictates the level of spot prices and the need for a complementary compensation. Figure \ref{fig:decomposition contract spot extreme highest} illustrates a scenario where the capacity margin is very little during all the period leading to high spot prices, accompanied by a negative capacity remuneration.\\ 

     Observe finally that the probability of getting negative net revenues drops from 26\% in the case of no adjustment (very uncertain for producer) to 5\% in the case with a CRM. It remains howerver slightly more than in the case without a CRM (2\%) which is natural since in this scenario the producer's only goal is to maximize his utility.
	
	\subsubsection{When is there missing money or a negative capacity payment}
	
	We investigate in this section ``unfavorable'' scenarios. Over the 5000 scenarios simulated earlier with a CRM, we select the ones where there is a missing money (MM) (where the spot remuneration is less than producer's total costs (28\% of scenarios)), and the scenarios where the capacity remuneration is negative (NCR); (18\% of scenarios), and the ones where there is a missing money and the capacity remuneration is negative (MM and NCR); (1\% of scenarios). 
	
	\begin{table}[H]
		\centering
		\resizebox{0.4\linewidth}{!}{
		\begin{tabular}{lrrr}
			\toprule
			& MM & NCR & MM and NCR\\
			\midrule
			Percentage [\%] & 28 & 18 & 1\\
			\bottomrule
		\end{tabular}}
	\captionsetup{width=.8\linewidth}
	\caption{Missing money and negative capacity payment with a CRM. This table focuses on our model for CRM, and presents the percentages of ``pathological'' scenarios over $N=5000$; the scenarios with missing money (spot compensation < total costs), those with a negative capacity remuneration, and those with both missing money and negative capacity remuneration.}
	\label{table : Missing money and negative capacity payment}
	\end{table}
	
	We compute the same indicators as before for two groups of the selected scenarios (the ones with missing money ``MM'' and the ones with a negative capacity remuneration ``NCR''), and we discard the third group (with both ``MM and NCR'') since it concerns only 1\% of the scenarios (average quantities are indeed meaningless in that case as the number of scenarios are very low). We summarize the results in table \ref{table : Singular policies}. Note that these indicators are computed with a Monte-Carlo method using different sizes of samples, and so do not have the same accuracy because of the different confidence intervals.	
	\begin{table}[H]
		\centering
		\resizebox{0.8\linewidth}{!}{
			\begin{tabular}{lrrrrrrrr}
				\toprule
\multicolumn{1}{c}{ } & \multicolumn{2}{c}{Reference} & \multicolumn{2}{c}{MM} & \multicolumn{2}{c}{NCR} \\
\cmidrule(l{3pt}r{3pt}){2-3} \cmidrule(l{3pt}r{3pt}){4-5} \cmidrule(l{3pt}r{3pt}){6-7}
  & mean & sd & mean & sd & mean & sd\\
\midrule
\addlinespace[-1em]
\multicolumn{7}{l}{\textbf{}}\\
\hspace{1em}Shortage hours per year [Hours] & 2.2 & 7.4 & 0.4 & 5.1 & 3.5 & 8.7\\
\hspace{1em}Average Spot price [euro/MWh] & 37.6 & 9.6 & 25.8 & 4.9 & 44.6 & 6.5\\
\hspace{1em}Average Margin [GW] & 32.4 & 9.0 & 43.9 & 6.9 & 26.2 & 4.7\\
\hspace{1em}Spot revenues [euro/MWh] & 37.8 & 9.6 & 26.0 & 4.9 & 44.8 & 6.5\\
\hspace{1em}Capacity payment [euro/MWh] & 12.3 & 13.2 & 22.7 & 11.6 & -6.4 & 5.3\\
\hspace{1em}Spot + Capacity payment[euro/MWh] & 50.1 & 10.9 & 48.7 & 10.4 & 38.4 & 7.7\\
\hspace{1em}\hspace{1em}\hspace{1em}Participation constraint [euro/MWh] & 2.8 & 0.0 & 2.8 & 0.0 & 2.8 & 0.0\\
\hspace{1em}\hspace{1em}\hspace{1em}Risk shared [euro/MWh] & -0.2 & 11.7 & 0.1 & 11.2 & -14.6 & 8.2\\
\hspace{1em}\hspace{1em}\hspace{1em}Risk compensation [euro/MWh] & 14.9 & 3.9 & 13.1 & 2.8 & 16.7 & 4.7\\
\hspace{1em}\hspace{1em}\hspace{1em}Total costs [euro/MWh] & 32.7 & 3.2 & 32.8 & 2.7 & 33.5 & 4.1\\
\hspace{1em}\hspace{1em}\hspace{1em}\hspace{1em}\hspace{1em}Construction and dismantling [euro/MWh] & 1.8 & 3.7 & 0.3 & 2.8 & 3.5 & 4.5\\
\hspace{1em}\hspace{1em}\hspace{1em}\hspace{1em}\hspace{1em}Maintenance [euro/MWh] & 13.2 & 1.3 & 14.9 & 1.0 & 12.4 & 0.7\\
\hspace{1em}\hspace{1em}\hspace{1em}\hspace{1em}\hspace{1em}Production [euro/MWh] & 17.6 & 0.0 & 17.6 & 0.0 & 17.6 & 0.0\\
\bottomrule
		\end{tabular}}
	\captionsetup{width=.8\linewidth}
	\caption{Missing money and negative capacity remuneration with a CRM. This table regroups the average and sd of the relevant quantities for a simulation of $N=5000$ scenarios with the recommended policy. The first two columns provide the averages and sds over all the scenarios for reference. The second two columns provide the average and sd of the same quantites but only on the scenarios with missing money (spot revenues<total costs), and the last two columns provide the averages and sds over the scenarios where the capacity payment is negative.}
	\label{table : Singular policies}
	\end{table}

	The major effect which explains missing money and negative capacity remuneration is the margin of the system. Missing money often comes with high margin while a negative capacity remuneration happens mostly with a low margin.\\
	
	A close look at table \ref{table : Singular policies} shows that there is missing money whenever the average margin is high (43 GW) in average, and so spot revenues are low (26 euros/MWh), and are not enough to cover total costs especially since maintenance costs are higher than average (14.9 euros/MWh compared to 13.2 euros/MWh) which overrules the fact that construction costs are close to zero. The CRM completes producer's earnings on spot market since it compensates him systematically for his costs (whether high or low), and this can be seen in table \ref{table : Singular policies} by the fact that average total compensation is above costs.\\
	
    The scenarios with negative capacity remuneration are more subtle to understand. In fact, they occur because our CRM is designed to take into account the spot compensation and complement it only when needed. In other words, if we recall again the decomposition of total compensation \eqref{eq : decomposition en lettres}: 	
    \begin{align}
    \mbox{Capacity remuneration}+\mbox{Spot compensation}&=\Rc+\mbox{{Producer's costs}}+\mbox{{Risk shared}}+\mbox{Risk compensation},\tag{\ref{eq : decomposition en lettres}}
    \end{align}
    we can see that negative capacity remuneration is equivalent to 
    \begin{align}\label{eq : inequality negative capacity}
    \mbox{Spot compensation}&>\Rc+\mbox{{Producer's costs}}+\mbox{{Risk shared}}+\mbox{Risk compensation}.
    \end{align}
    
    This happens whenever the risk shared is negative (as confirmed by table \ref{table : Singular policies} : -14.6 euros/MWh), suggesting unfavorable uncertainties and low capacity margins (26 GW compared to 32 GW in average), which increases spot revenues (44.8 euros/MWh) and consolidates the inequality \eqref{eq : inequality negative capacity}. A typical scenario with a negative capacity remuneration would be one with a consistent low demand and high capacity without much realized volatility. This would keep the risk compensation low, with negative risk shared and high spot compensation. Nevertheless, observe that even in such conditions (unfavorable uncertainties), producer manages to keep an acceptable number of shortage hours per year (3.5 hours per year).
	
	\subsection{The optimal capacity payment in other setups}\label{section : sensitivities}
	We test our system under different conditions, in the case where there are more renewable energies; and thus more uncertainties in the system (associated with higher $\sigma^C$ ($50\%$ higher)), or when there is a demand response program applied on consumers, \ie assuming demand volatility $\sigma^D$ is lower ($50\%$ lower).\\
	We also test different risk aversions; with more risk averse consumers or more risk averse producers (multiplying by 10 the risk aversion parameter each time).
		
	\subsubsection{More renewable energies or a demand response program ($\sigma^C$ and $\sigma^D$)}
	
	A brief summary of the numerical results with perturbations of volatilites can be found in table \ref{table : ratios sensitivities}. 
	
	\begin{table}[H]
		\centering
		\resizebox{0.8\linewidth}{!}{
			\begin{tabular}{lrrr}
\toprule
  & Reference & More renewables & Demand response program\\
\midrule
Spot revenues/Total revenues [\%] & 70 & 41 & 98\\
Scenarios with missing money [\%] & 28 & 74 & 25\\
Scenarios with negative capacity remuneration [\%] & 18 & 0 & 60\\
Scenarios with negative total compensation [\%] & 0 & 0 & 0\\
Scenarios with negative net revenues [\%] & 5 & 0 & 13\\
\bottomrule
		\end{tabular}}
	\captionsetup{width=.8\linewidth}
	\caption{Impacts of a shock in volatility on the CRM. This table presents the variation of total compensation composition and the percentages of scenarios with missing money, negative capacity remuneration, negative total compensation and negative net revenues. The first column recalls the reference results with our calibrated parameters. The second column provides the results when the volatility of capacity is $50\%$ higher, interpreted as an increase of the proportion of renewable energies (which are more variable). The third column provides the results when the volatility of demand is $50\%$ lower, interpreted as the introduction of some demand-response program driving consumers to have less variable demand.}
	\label{table : ratios sensitivities}
	\end{table}
	
	We can see that introducing more renewable energies in the system yields a higher percentage of scenarios with missing money (74\% as opposed to 28\%), and to positive net revenues all the time. In this context, the capacity compensation is never negative (in the observed scenarios), and plays a much more important role in complementing producer revenues (59\% of total revenues instead of only 30\%). This is similar to what we observed in section \ref{subsection : Link between capacity payment and spot compensation}, when comparing a severe and favorable scenario.\\
	
	By looking further into table \ref{table : Different volatility}, we can explain the missing money by the high costs of construction and dismantling (7.6 euros/MWh compared to 1.8 euros/MWh), and the role of capacity payment by the considerable amount of risk compensation (34.5 euros/MWh compared to 14.9 euros/MWh).\\
	
	So in summary the more uncertainties in production, the more the need for a capacity mechanism in order to cope with the random electricity demand and the random available capacity, since otherwise the number of shortage hours would increase. The capacity mechanism results in increasing the capacity margin of the system. As the system is longer in terms of capacity, spot prices are lower. Therefore, in average, consumers have to pay a higher capacity remuneration when the capacity volatility increases, and producers receive a higher total compensation, and higher earnings even though the spot compensation decreases.
	
	\begin{table}[H]
		
		\centering
		\resizebox{0.8\linewidth}{!}{
			\begin{tabular}{lrrr}
				\toprule
				& Reference & More renewables & Demand response program\\
				
\midrule
\addlinespace[-1em]
\multicolumn{4}{l}{\textbf{}}\\
\hspace{1em}Shortage hours per year [Hours] & 2.2 & 2.8 & 0.7\\
\hspace{1em}Average Spot price [euro/MWh] & 37.6 & 32.5 & 40.9\\
\hspace{1em}Average Margin [GW] & 32.4 & 39.1 & 29.9\\
\hspace{1em}Spot revenues [euro/MWh] & 37.8 & 32.6 & 40.7\\
\hspace{1em}Capacity payment [euro/MWh] & 12.3 & 41.0 & -3.2\\
\hspace{1em}Spot + Capacity payment[euro/MWh] & 50.1 & 73.6 & 37.4\\
\hspace{1em}\hspace{1em}\hspace{1em}Participation constraint [euro/MWh] & 2.8 & 0.0 & 2.2\\
\hspace{1em}\hspace{1em}\hspace{1em}Risk shared [euro/MWh] & -0.2 & -0.4 & -0.1\\
\hspace{1em}\hspace{1em}\hspace{1em}Risk compensation [euro/MWh] & 14.9 & 34.5 & 4.0\\
\hspace{1em}\hspace{1em}\hspace{1em}Total costs [euro/MWh] & 32.7 & 39.4 & 31.3\\
\hspace{1em}\hspace{1em}\hspace{1em}\hspace{1em}\hspace{1em}Construction and dismantling [euro/MWh] & 1.8 & 7.6 & 0.8\\
\hspace{1em}\hspace{1em}\hspace{1em}\hspace{1em}\hspace{1em}Maintenance [euro/MWh] & 13.2 & 14.2 & 12.9\\
\hspace{1em}\hspace{1em}\hspace{1em}\hspace{1em}\hspace{1em}Production [euro/MWh] & 17.6 & 17.6 & 17.6\\
\bottomrule
		\end{tabular}}
	\captionsetup{width=.8\linewidth}
	\caption{Impacts of a shock in volatility on the number of shortage hours, and the costs and revenues of the CRM. The first column recalls the reference results with our calibrated parameters. The second column provides the results when the volatility of capacity is $50\%$ higher, interpreted as an increase of the proportion of renewable energies (which are more variable). The third column provides the results when the volatility of demand is $50\%$ lower, interpreted as the introduction of some demand-response incentive driving consumers to have less variable demand.}
	\label{table : Different volatility}
	\end{table}
	
The third column of tables \ref{table : ratios sensitivities} and \ref{table : Different volatility} summarizes numerical results when the demand volatility is lower, which is a model for a demand response program; \ie we assume that consumer is somehow incentivized to behave in a more predictable manner, so that the uncertainty on demand fluctuations ($\sigma^D$) is lower.\\
	
In this case, the security of the system can be ensured with a lower capacity margin (29.9 GW instead of 32.4 GW) since there are less uncertainties. Therefore the average spot price and the spot revenues are higher (40.9 euros/MWh instead of 37.6 euros/MWh), so we are likely to obtain the inequality \eqref{eq : inequality negative capacity} which we recall 
\begin{align}
\mbox{Spot compensation}&>\Rc+\mbox{{Producer's costs}}+\mbox{{Risk shared}}+\mbox{Risk compensation},\tag{\ref{eq : inequality negative capacity}}
\end{align} 
which is equivalent to having a negative capacity remuneration especially since none of the terms on the right hand side should be high. This explains the high percentage of scenarios with negative capacity remuneration (60\%).\\
    
These conditions suggest an auto-regulated market and less need for a CRM, since spot revenues represent $98\%$ of total compensation, even though there is missing money $25\%$ percent of the time, and often a negative capacity remuneration (60\% of the time). We can guess that taking away the capacity remuneration would lead to an equilibrium situation with more uncertainties for producer: he incurs losses when there is a missing money, but these losses are balanced by the scenarios where the spot price is high, and he does not have to ``pay'' for capacity remuneration which is now a cost for him rather than a revenue, making his net revenues negative in 13\% of the scenarios. Nevertheless, the total remuneration of producer (spot + capacity payment) is still above in average the total costs: in average with the CRM, the producer is going to earn money.

	\subsubsection{More risk aversion producer or consumer ($\etaa$ and $\etap$)}
	We analyze the impacts of the CRM when producer or consumer is more risk averse. The numerical results are summarized in tables \ref{table : More risk aversions} and \ref{table : More risk aversions ratios}.
	
	\begin{table}[H]
		\centering
		\resizebox{0.8\linewidth}{!}{
			\begin{tabular}{lrrr}
				\toprule
  & Reference & Risk averse producer & Risk averse consumer\\
\midrule
\addlinespace[-1em]
\multicolumn{4}{l}{\textbf{}}\\
\hspace{1em}Shortage hours per year [Hours] & 2.2 & 0.8 & 0.5\\
\hspace{1em}Average Spot price [euro/MWh] & 37.6 & 32.8 & 34.4\\
\hspace{1em}Average Margin [GW] & 32.4 & 36.6 & 34.9\\
\hspace{1em}Spot revenues [euro/MWh] & 37.8 & 33.0 & 34.6\\
\hspace{1em}Capacity payment [euro/MWh] & 12.3 & 51.2 & 380.6\\
\hspace{1em}Spot + Capacity payment[euro/MWh] & 50.1 & 84.2 & 415.3\\
\hspace{1em}\hspace{1em}\hspace{1em}Participation constraint [euro/MWh] & 2.8 & 0.0 & 2.7\\
\hspace{1em}\hspace{1em}\hspace{1em}Risk shared [euro/MWh] & -0.2 & -0.1 & -0.3\\
\hspace{1em}\hspace{1em}\hspace{1em}Risk compensation [euro/MWh] & 14.9 & 49.7 & 378.7\\
\hspace{1em}\hspace{1em}\hspace{1em}Total costs [euro/MWh] & 32.7 & 34.6 & 34.1\\
\hspace{1em}\hspace{1em}\hspace{1em}\hspace{1em}\hspace{1em}Construction and dismantling [euro/MWh] & 1.8 & 3.2 & 2.9\\
\hspace{1em}\hspace{1em}\hspace{1em}\hspace{1em}\hspace{1em}Maintenance [euro/MWh] & 13.2 & 13.8 & 13.6\\
\hspace{1em}\hspace{1em}\hspace{1em}\hspace{1em}\hspace{1em}Production [euro/MWh] & 17.6 & 17.6 & 17.6\\
\bottomrule
		\end{tabular}}
	\captionsetup{width=.8\linewidth}
	\caption{Impacts of a shock in risk aversion parameters on the number of shortage hours, and the costs and revenues of the CRM. The first column recalls the reference results with our calibrated parameters. The second column provides the results when producer's risk aversion is higher ($10\x\etaa$). The third column provides the results when consumer's risk aversion is higher ($10\x\etap$).}
	\label{table : More risk aversions}
	\end{table}
	
	The first observation is that a more risk averse producer has a zero participation constraint, which means that the utility gained from a system without a CRM is negative.\\
	
	Apart from this observation, we have in both cases the same impacts but with different magnitudes. We can see that the costs of construction are slightly higher, with lower spot revenues. This leads to an increase in occurrences of missing money scenarios, and a larger part of total revenues coming from capacity compensation (always positive, and representing 41\% and 92\% from total compared to 30\%). This capacity compensation comes mainly from the risk compensation (51.2 euros/MWh and 380 euros/MWh compared to 15 euros/MWh in the reference case).    
	
	 \begin{table}[H]
		\centering
		\resizebox{0.8\linewidth}{!}{
			\begin{tabular}{lrrr}
				\toprule
  & Reference & Risk averse producer & Risk averse consumer\\
\midrule
Spot revenues/Total revenues [\%] & 70 & 39 & 8\\
Scenarios with missing money [\%] & 28 & 53 & 41\\
Scenarios with negative capacity remuneration [\%] & 18 & 0 & 0\\
Scenarios with negative total compensation [\%] & 0 & 0 & 0\\
Scenarios with negative net revenues [\%] & 5 & 0 & 0\\
\bottomrule
		\end{tabular}}
	\captionsetup{width=.8\linewidth}
	\caption{Impacts of a shock in risk aversion parameters on the CRM. This table presents the variation of total compensation composition and the percentages of scenarios with missing money, negative capacity remuneration, negative total compensation and negative net revenues. The first column recalls the reference results with our calibrated parameters. The second column provides the results when producer's risk aversion is higher ($10\x\etaa$). The third column provides the results when consumer's risk aversion is higher ($10\x\etap$).}
	\label{table : More risk aversions ratios}
	\end{table}
	
	To summarize, whenever one of the two parties is more risk averse, it becomes very costly for consumer to pay for capacity remuneration. A risk averse producer would require more risk compensation (the risk compensation is proportional to producer's risk aversion by definition). From the other hand, a risk averse consumer would be ready to spend a lot to avoid potential shortage or blackout. The consequence is that the producer gets positive total compensation and positive net revenues 100\% of the time.

	\section{Conclusion}
		In this paper, we provide some insights on how electricity producers and consumers could share the financial and physical risks (and uncertainties) to ensure the security of the system. We propose a CRM based on contract theory, which incentivizes producers to perform an optimal level of effort to maintain and develop new power-plants. It takes the form of a contract and a recommended effort, with payment adjusted to the uncertainty of outcomes (weather, outages...) ensuring to producers the right level of average earnings and financial risks while accounting for the spot revenues, as long as they follow the recommended effort.\\
	
	This means that the CRM does not disrupt the spot market operations, which remains therefore able to ensure the short-term optimal economical dispatch. Given a predefined level of security for the system, the capacity mechanism we propose provides the right level of investment needed to insure it, and gives us insights to challenge real implementations of capacity markets.\\
	
	One of our main results is that we point out the necessity of a CRM. This is reinforced with the level of randomness of the capacity and demand. Higher share of random renewable production in the electrical system means that a higher fraction of the compensation of the producers needs to come from a capacity remuneration system. As a matter of fact, in that case the volume of installed capacities should be more important to ensure the security of the system. Meanwhile, the spot price decreases (the spot price is decreasing with respect to installed capacities, which is consistent with the increase of the  supply curve), and so it is essential to support installed capacities. \\
	
	In the meantime, the higher the volatility of consumption or supply, the higher the volatility of the capacity market. The mechanism we proposed also enables to study how risks should be share between the producers and the consumers. The increase of the financial risk is principally supported by the producers who are then compensated by a higher average revenue. However, even when there is a capacity payment, the number of hours of shortage increases when consumption or production become more volatile. In fact, if consumers do not modify the virtual value they associate to shortage, it is economically optimal for them to accept more hours of shortage instead of increasing suppliers' compensation. The capacity payment enables to share the financial risk between the producers and consumers depending on their risk aversion. The producers accept to take more financial risk if it comes with an increase of their average revenue. This is also what happens when the consumers want to reduce the physical risk.\\
	
	Finally, we propose some variants for further research. It would be interesting to challenge our results by using other spot functions, especially ones that could reach higher peaks than the function \eqref{eq : prix spot}. Furthermore, we assumed that the demand process is completely exogenous and uncontrolled. We explored this aspect of the model to a certain extent in section \ref{section : sensitivities} by studying the sensitivities of our results with respect to the volatility of demand. Adding a direct control on the demand process would be relevant but would also change drastically the resolution methodology and is therefore an open research question. The same can be said about the assumption that dismantling power plants has the same cost as building new ones.\\  
	
	In addition, representing several technologies to produce electricity would also make a lot of sense. This could be the starting point for a future work by considering for example the possibility to control the volatility of the capacity process or to have several producers with different technologies instead of one. We also used a continuous time setting, which is convenient for modeling and computations, but requires using small time steps to discretize and approximate the optimal control (in our simulations we used a daily time step). This is equivalent to ignoring the delay needed to build new power plants which is not completely realistic, but remains quite common in such models because decision-making delay is hard to capture. Finally, our model only represents one design of CRM (a kind of sophisticated capacity payment). Further work may include other possible market designs.

	\section{Appendix}
		\subsection{Mathematical framework and weak formulation}\label{appendix : Rigorous mathematical framework and weak formulation}
	This section is devoted to the mathematical formulation of the problem. For $T\in\brak{0,\ +\infty}$ a fixed maturity, we denote by $\Omega:=\Cc\brak{\edg{0,T}, \mathbb{R}^{2}}$
	the space of continuous functions from $\edg{0,\ T}$ to $\R^{2}$. The system is described by the state variable $X:=\brak{X^{C},X^{D}}^{T}$ which
	is the canonical process on $\Omega$.  Finally we endow $\Omega$ with its Borel $\sigma$-algebra $\Fc_{T}$ and define the completed filtration $\F$
	generated by the process $X$.\\
	
	We define the reference probability measure as the weak solution of the controlled equation \eqref{eq:dyn-SDE} with a constant control set to zero, \ie $\alpha_t=0$ for $t\in \edg{0,T}$, on the space $\brak{\Omega,\Fc_{T}}$ and we denote it $\P^{0}$. It is characterized as the unique probability measure such that $\P^{0}\circ\brak{X_0}^{-1}=\delta_{x_0}$ for some $x_0\in\R_+^2$  and the processes $\brak{X_t-\int_0^t\mu\brak{X_s,0}ds}_{t\in\edg{0,T}}$ is a $\brak{\P^{0} , \F}$--martingale with $\left<X\right>_t=\int_0^t\sigma\brak{X_s}\sigma^{\intercal}\brak{X_s}ds$ for $t\in \edg{0,T}$. Note that existence and uniqueness of this measure are insured by the existence of a unique strong solution to the corresponding SDE. 
	Therefore there exists a 2--dimensional $\P^{0}$--Brownian motion $W^{\P^{0}}$ such that
	\begin{align}
	X_t = X_0 +\int_0^t\mu\brak{X_s,0}ds + \int_0^t\sigma\brak{X_s}dW^{\P^{0}}_s, \mbox{ for $t\in\edg{0,T}$}.
	\end{align}
	Notice that under $\P^{0}$ the component $X^C$ is a martingale (the drift part of $X^C$ is zero when the control is constantly equal to $0$).\\  
\begin{Definition}[Admissible controls]\label{def : admissible controls}
		Recall from section \ref{subsection Model, state variables and control} that $\Uc$ is the set of $\F$--predictable processes valued in $[\alpha_{min},\alpha_{max}]$. For $\a\in \Uc$, the Novikov's criterion is satisfied (from the boundedness of $\a$):
		\begin{align}\label{eq : controle admissible}
		\E^{\P^{0}}\edg{e^{\frac{1}{2}\int_0^T\brak{\frac{\alpha_t}{\sigma^C}}^2dt}} <+\infty,
		\end{align}	
		and so the process 
		\begin{align}
		\Ec\brak{\int_0^t\frac{\alpha_s}{\sigma^C}dW^{C,\P^{0}}_s}_{t\in[0,T]}\mbox{ is a }\brak{\P^{0} , \F}\mbox{-UI martingale,}
		\end{align}
		and we can define the induced probability measure $\P^{\a}$ on $\brak{\Omega,\Fc_{T}}$ as the equivalent measure to $\P^{0}$ with its Radon–Nikodym derivative
		\begin{align}
		\begin{cases}
		\frac{\mathrm{d} \P^{\a}}{\mathrm{d} \P^{0}} := \Ec\brak{\int_0^T\frac{\alpha_s}{\sigma^C}dW^{C,\P^{0}}_s},\\
		L_t^{\a}:=\left.\frac{\mathrm{d} \P^{\a}}{\mathrm{d}\P^{0}}\right|_{\Fc_t},
		\end{cases}
		\end{align}
		where $W^{C,\P^{0}}$ is the capacity component of the Brownian motion. We denote by $\Pc$ the collection of probability measures $\P^{\alpha}$ induced by the set of controls $\a\in\Uc$.
	\end{Definition}
	\no By applying Girsanov's theorem, we can see that for an admissible control $\a$, the process $X$ has the law of SDE \eqref{eq:dyn-SDE} under $\P^{\a}$.\\ Remark that the demand component is not affected by this change of probability, \ie \mbox{$\P^{0}\circ\brak{X^D}^{-1}=\P^{\alpha}\circ\brak{X^D}^{-1}$}, for every admissible $\a$, which is consistent with our model: the control is only on (the drift) of the capacity.\\
		\begin{Definition}
		The set $\Xi$ of admissible contracts is defined as the collection of $\Fc_T$--measurable random variables $\xi$ satisfying
		\begin{align}\label{eq : admissible contracts}
		\begin{cases}
		V_{0}^{A}\brak{\xi}\geq U_A\brak{\Rc},\\
		{\Pc}^{\star}\brak{\xi} \neq \emptyset,\\
		\Sup_{\P^{\a}\in\Pc}\E^{\P^{\a}}\edg{e^{(\etaa\lor\etap) (1+\delta)\abs{\xi}}}<\infty,\mbox{ for some }\delta>0.
		\end{cases}
		\end{align} 
	\end{Definition}
	\no The definition of admissible contracts imposes the existence of an optimal control for agent which satisfies the participation constraint, and non-degeneracy conditions for both principal and agent problems.
%
%
\subsection{Definition and derivation of the class of revealing contracts $\Zc$}
\subsubsection{Definition}\label{appendix : definition of the revealing contracts}
\no Let $\Vc$ be the set of $\F$-predictable processes $Z$ valued in $\R^2$ satisfying
\begin{align}
\begin{cases}
\E^{\P^{\a}}\edg{\Ec\brak{-\etaa\int_0^T Z_t\cdot \sigma\brak{X_t}dW_t^{\a}}}=1,\mbox{ for }\P^{\a}\in\Pc,\\
\Sup_{\P^{\a}\in\Pc}\E^{\P^{\a}}\edg{e^{(\etaa\lor\etap) (1+\delta)\abs{Y_T^{0,Z}}}}<\infty,
\end{cases}
\end{align}
where we recall 
\begin{align}
Y_{t}^{Y_0,Z}=Y_{0}+\int_{0}^{t}Z_{s}\cdot dX_{s}-\int_{0}^{t}H\brak{X_{s},Z_{s}}ds,\mbox{ for  }t\in \edg{0,T},\tag{\ref{eq : contrat revelateur}}
\end{align}
and $H$ defined as in \eqref{eq big hamiltonian}.
The class of revealing contracts $\Zc$ is then the subset of $\Xi$ containing the terminal values of $\brak{Y_{t}^{Y_0,Z}}_{t\in[0,T]}$ parameterized by the two control variables $(Y_0,Z)\in(\R\x\Vc)$;
\begin{align}
\Zc = \crl{Y_T^{Y_0,Z}\mbox{ for some } \brak{Y_0,Z}\in\R\x\Vc\mbox{ with }Y_T^{Y_0,Z}\in\Xi}.\tag{\ref{eq : ensemble des contrats revelateur}}
\end{align}
\subsubsection{Derivation of the class $\Zc$}\label{appendix : Formal derivation of the revealing contracts}
The aim of this subsection is to explain the intuition behind the class of revealing contracts $\Zc$ in the spirit of \cite{Sannikov08}, which is a key ingredient in the resolution of our problem. Rather than providing a rigorous treatment of this question (for which we refer the reader to \cite{cvitanicpt}), we only present here the main ideas behind it.\\
	\no The main goal of introducing the class $\Zc$ which--we recall--plays the role of a performance index, is to overcome the non-markovianity of the principal-agent problem, and to make agent's response ``predictable'' by principal. This is achieved by the martingale optimality principle.\\
	
	\no We start by recalling agent's value function:
	\begin{align}
	V_{0}^{A}\brak{\xi} = \sup_{\P^{\a}\in\Pc}\E^{\P^{\a}}\edg{U_{A}\brak{\xi+\int_{0}^{T}s(X_t)dt-\int_{0}^{T}c^{A}\brak{X_{t},\a_{t}}dt}}.
	\end{align}
	We first restrict the contracts to the ones which are terminal values of some diffusion process, \ie of the form $\xi = Y_T^{Y_0,Z}$ where $Y_0$ and $Z$ are new control variables, and $Y_T^{Y_0,Z} = Y_0+\int_0^TZ_t\cdot dX_t + \int_0^T g\brak{t,X_t,Y_t,Z_t}dt$, where $g$ is a deterministic function to be determined.\\
	
	Roughly speaking, this allows us to reduce to the markovian case by recapturing principal's ``missing'' information and plugging it into the new process $Y_T^{Y_0,Z}$. This is done through the change of control variables from $\xi$ to $Y_0$ and $Z$. We stress here that this not a mathematical proof but just an explanation, since $Y_T^{Y_0,Z}$ is a solution to some SDE and has a priori no reason to exist (so far we didn't impose any restrictions on the function $g$).\\
	
	Now that we are in a markovian setting, we want to get rid of the moral-hazard, \ie to have a predicable response from the agent. This is possible by a careful choice of the function $g$. In fact, by the structure of the exponential utility $U_A(x)=-e^{-\etaa x}$, we have by a simple application of It\^o formula that the process defined by $U_{A}\brak{Y_t^{Y_0,Z}+\int_{0}^t\brak{s(X_r) -c^{A}(X_r,\a_r)}dr}_{t\in\edg{0,T}}$, will be a $\brak{\P^{\a},\F}$-supermartingale and a martingale only for some $\hat{\a}$ if we choose $g$ such that   
	\begin{align}
	g\brak{t,x,y,z} = - H\brak{x,z}\mbox{ for }(t,x,y,z)\in\edg{0,T}\x\R_+^2\x\R\x\R^2,     
	\end{align}
	with $H$ the hamiltonian defined in \eqref{eq big hamiltonian}. This gives us first an upper bound on agent's optimal control $V_{0}^{A}\brak{Y_T^{Y_0,Z}}\leq U_A\brak{Y_0}$ by the supermartingale property, and we have that this bound is attained for the control induced by $\P^{\hat{\a}}$ by construction.\\
	
	Therefore, since agent is rational and aims at maximizing his utility, he chooses the control $\hat{\a}$, which is a deterministic function of the pair $\brak{X_t,Z_t}_{t\in\edg{0,T}}$, both observable by principal.
	
	\subsection{Solving Producer's problem: proof of Proposition \ref{Agent's response}}\label{Appendix : proof of Agent's response}	
	Let $Y_T^{Y_0,Z}\in\Zc$ and $\P^{\alpha}\in\Pc$. By definition
	\begin{align}
	\begin{split}
	J_0^A\brak{Y_T^{Y_0,Z}, \P^{\a}}&=\E^{\P^{\a}}\edg{U_{A}\brak{Y_T^{Y_0,Z}+\int_{0}^{T}\brak{s(X_t)-c^A(X_{t},\a_{t})}dt}},\\
	&=\E^{\P^{\a}}\edg{U_{A}\brak{Y_{0}+\int_{0}^TZ_t\cdot dX_t-\int_{0}^TH\brak{X_t,Z_t}dt+\int_{0}^{T}\brak{s(X_t)-c^A(X_{t},\a_{t})}dt}}.
	\end{split}
	\end{align}
	Using \eqref{eq : small hamiltonian} :
	\begin{align}\label{eq: int de zdx}
	\int_{0}^TZ_{s}\cdot dX_{s}&= \int_{0}^Th\brak{X_t,Z_t,\alpha_t}dt+\int_{0}^T\sigma\brak{X_t}Z_t\cdot dW_t^{\a}-\int_{0}^{T}\brak{s(X_t)-c^A\brak{X_t,\a_t}}dt+\frac{\etaa}{2}\int_{0}^T|\sigma\brak{X_t}Z_t|^2dt, 
	\end{align} 
	and therefore replacing $U_A$ with its expression and injecting \eqref{eq: int de zdx} we get
	\begin{align}
	\begin{split}
	J_0^A\brak{Y_T^{Y_0,Z}, \P^{\a}}&=\E^{\P^{\a}}\edg{-e^{-\etaa\brak{Y_{0}+\int_{0}^Th\brak{X_t,Z_t,\alpha_t}-H\brak{X_t,Z_t}dt+\int_{0}^T\sigma\brak{X_t}Z_t\cdot dW_t^{\a}+\frac{\etaa}{2}\int_{0}^T|\sigma\brak{X_t}Z_t|^2dt}}}\\
	&=U_A\brak{Y_0}\E^{\P^{\a}}\edg{e^{\etaa\int_{0}^T\crl{H\brak{X_t,Z_t}-h\brak{X_t,Z_t,\alpha_t}}dt}\Ec\brak{-\etaa\int_0^T Z_t\cdot \sigma\brak{X_t}dW_t^{\a}}}.
	\end{split}
	\end{align}
	Since $Z\in\Vc$ we have that 
	\begin{align}
		\E^{\P^{\alpha}}\edg{\Ec\brak{-\etaa\int_0^T Z_t\cdot \sigma\brak{X_t}dW_t^{\a}}}=1,
	\end{align}
	and we can define the probability measure $\widetilde{\P^{\a}}$ equivalent to $\P^{\a}$ with the Radon-Nikodym derivative
	\begin{align}
	\Frac{\diff\widetilde{\P^{\a}}}{\diff\P^{\a}}:=\Ec\brak{-\etaa\int_0^T Z_t\cdot \sigma\brak{X_t}dW_t^{\a}},
	\end{align}
	\no and so 
	\begin{align}
	J_0^A\brak{Y_T^{Y_0,Z}, \P^{\a}}&=U_A\brak{Y_0}\E^{\widetilde{\P^{\a}}}\edg{e^{\etaa\int_{0}^T\crl{H\brak{X_t,Z_t}-h\brak{X_t,Z_t,\alpha_t}}dt}}.
	\end{align}
	Recalling that $U_A\brak{Y_0}<0$, and $H\brak{X_t,Z_t}-h\brak{X_t,Z_t,\alpha_t}\geq 0$ with equality if and only if $\alpha_t=\hat{\a}\brak{X^C_t,Z^C_t}$ for all $t$ in $\edg{0,T}$, we obtain the upper bound 
	\begin{align}
	J_0^A\brak{Y_T^{Y_0,Z},\P^{\a}}&\leq U_A\brak{Y_0}\mbox{ for all }\P^{\a}\in \Pc,
	\end{align} 
	which is attained for $\brak{\hat{\a}\brak{X^C_t,Z^C_t}}_{t\in[0,T]}$ since it is an admissible control (as it is a progressively measurable process valued in $[\alpha_{min},\alpha_{max}]$) and so
	\begin{align}
	J_0^A\brak{Y_T^{Y_0,Z},\P^{\hat{\a}}}&= U_A\brak{Y_0},
	\end{align}
	which yields
	\begin{align}
	V_{0}^{A}\brak{Y_T^{Y_0,Z}}  =  U_{A}\brak{Y_0},
	\end{align}
	and $\brak{\hat{\a}\brak{X^C_t,Z^C_t}}_{t\in[0,T]}$ is producer's optimal response given the contract $Y^{Y_0,Z}_T$. 
	\qed
		\subsection{Proof of Proposition \ref{Proposition : existence BSDE}\label{appendix Proposition : existence BSDE}}
	
	For readers familiar with BSDE theory, Proposition \ref{Proposition : existence BSDE} can be seen as an existence result for BSDEs with a quadratic generator. The following proof is largely inspired by \cite{eliep} and \cite{el2018optimal}, is classical in the non-Markovian stochastic control theory and relies on the \emph{Agent's continuation utility} as a natural candidate for the solution of the BSDE.\\
	
	We start by defining Agent's continuation utility, and prove that it satisfies the Dynamic Programming Principle. Then we use the assumptions on the set of admissible contracts and the properties of Agent's continuation utility to conclude the proof of Proposition \ref{Proposition : existence BSDE}. 
	\begin{Definition}
		Let $\tau$ be a stopping time valued in $\edg{t,T}$. We denote by $\Uc_{\tau}$ the restriction of (Agent's) controls to $\edg{\tau,T}$.
		We define the dynamic version of Agent's objective function for a given $\xi\in\Xi$ as 
		\begin{align}
		J_{\tau}^{A}\brak{\xi,\P^{\a}}:=\E_{\tau}^{\P^{\a}}\edg{U_{A}\brak{\xi+\int_{\tau}^{T}s(X_r)dr-\int_{\tau}^{T}c^{A}\brak{X_r,\a_r}dr}}\mbox{ and } \Jc^A_{\tau}\brak{\xi}
		:=\brak{J_{\tau}^{A}\brak{\xi,\P^{\a}}}_{\a\in\Uc_\tau},
		\end{align}
		and his continuation utility 
		\begin{align}\label{eq : continuation utility}
		V_{\tau}^A\brak{\xi} := \esssup_{\a\in\Uc_{\tau}}J_{\tau}^{A}\brak{\xi,\P^{\a}}.
		\end{align}
	\end{Definition}
	
	Remark that for any $\P^{\a}\in\Pc$, the conditional expectation $\E_{\tau}^{\P^{\a}}$ depends only on the restriction of $\a$ on $\edg{\tau,T}$. It is then defined without ambiguity for $\a\in\Uc_{\tau}$.
	
	\begin{Lemma}\label{lemma : DPP}
		For $\xi\in\Xi$, $t\in\edg{0,T}$, and $\tau$ an $\F$-stopping time in $\edg{t,T}$, we have that\\
		\rmi The family $\Jc^A_{\tau}\brak{\xi}$ satisfies the lattice property, therefore the limiting sequence approaching $V_{\tau}^A\brak{\xi}$ can be chosen to be non-decreasing, \ie there exists a sequence of $\brak{\P^{{\a}^n}}_{n\geq 0}$ such that
		\begin{align}
		V_{\tau}^A\brak{\xi} = \lim_{n\rightarrow+\infty}\uparrow J_{\tau}^{A}\brak{\xi,\P^{{\a}^n}}.
		\end{align}\\
		\rmii The dynamic programming principle for Agent's value function holds, i.e. for $\tau_1$ and $\tau_2$ two stopping times such that $0\leq \tau_1\leq \tau_2\leq T$:
		\begin{align}
		V_{\tau_1}^A\brak{\xi} = \esssup_{\P^{\a}\in\Pc}\E_{\tau_1}^{\P^{\a}}\edg{V_{\tau_2}^A\brak{\xi}e^{\etaa\int_{\tau_1}^{\tau_2}\brak{c^{A}(X_t,\a_t)-s(X_t)}dt}}.
		\end{align}
		
	\end{Lemma}  
	\begin{proof}
		\rmi We consider two controls $\a$ and $\a'$ in $\Uc_{\tau}$. We define then 
		\begin{align}
		\tilde{\a}:=\a\mathds{1}_{\crl{J_{\tau}^{A}\brak{\xi,\P^{{\a}}}\geq J_{\tau}^{A}\brak{\xi,\P^{{\a'}}}}} + \a'\mathds{1}_{\crl{J_{\tau}^{A}\brak{\xi,\P^{{\a}}}< J_{\tau}^{A}\brak{\xi,\P^{{\a'}}}}}
		\end{align}
		Then $\tilde{\a}\in\Uc_{\tau}$ and from the definition of $\tilde{\a}$ we have the inequality
		\begin{align}
		J_{\tau}^{A}\brak{\xi,\P^{{\tilde{\a}}}}&\geq \max\brak{J_{\tau}^{A}\brak{\xi,\P^{{\a}}},J_{\tau}^{A}(\xi,\P^{{\a'}})},
		\end{align}
		which proves the lattice property, implying \rmi\! \cite[Proposition VI.I.I,
		p121]{neveu72}.\\

		\no \rmii The proof of this part is similar to the one in \cite[Proposition 6.2]{CK93}. We proceed in two steps proving each of the two inequalities. The first inequality is a direct consequence of the tower property. In fact, for $0\leq \tau_1\leq \tau_2\leq T$, we have by definition
		
		\begin{align}
		\begin{split}
		V_{\tau_1}^A\brak{\xi} &= \esssup_{\a\in\Uc_{\tau_1}}\E_{\tau_1}^{\P^{\a}}\edg{-e^{-\etaa\brak{\xi-\int_{\tau_1}^{T}(c^{A}(X_r,\a_r)-s(X_r))dr}}},\\
		&=\esssup_{\a\in\Uc_{\tau_1}}\E_{\tau_1}^{\P^{\a}}\edg{-e^{-\etaa\brak{\xi-\int_{\tau_1}^{\tau_2}\brak{c^{A}(X_r,\a_r)-s(X_r)}dr-\int_{\tau_2}^{T}\brak{c^{A}(X_r,\a_r)-s(X_r)}dr}}}.
		\end{split}
		\end{align}
		By the tower property of the expectation we write
		\begin{align}
		V_{\tau_1}^A\brak{\xi} &= \esssup_{\a\in\Uc_{\tau_1}}\E_{\tau_1}^{\P^{\a}}\edg{e^{\etaa\int_{\tau_1}^{\tau_2}\brak{c^{A}(X_r,\a_r)-s(X_r)}dr}\E_{\tau_2}^{\P^{\a}}\edg{-e^{-\etaa \brak{\xi -\int_{\tau_2}^{T}\brak{c^{A}(X_r,\a_r)-s(X_r)}dr}}}}.
		\end{align}
		
		\no Using Bayes rule and remarking that $\E_{\tau_2}^{\P^{\a}}\edg{-e^{-\etaa \brak{\xi -\int_{\tau_2}^{T}\brak{c^{A}(X_r,\a_r)-s(X_r)}dr}}}$ depends only on values of $\a$ after $\tau_2$, we have that for an arbitrary $\a\in\Uc$    
		\begin{align}
		\begin{split}
		\E_{\tau_2}^{\P^{\a}}\edg{-e^{-\etaa \brak{\xi -\int_{\tau_2}^{T}\brak{c^{A}(X_r,\a_r)-s(X_r)}dr}}}
		&\leq \esssup_{\a\in\Uc}\E_{\tau_2}^{\P^{\a}}\edg{U_{A}\brak{\xi-\int_{\tau_2}^{T}\brak{c^{A}(X_r,\a_r)-s(X_r)}dr}},\\
		&=\esssup_{\a\in\Uc_{\tau_2}}\E_{\tau_2}^{\P^{\a}}\edg{U_{A}\brak{\xi-\int_{\tau_2}^{T}\brak{c^{A}(X_r,\a_r)-s(X_r)}dr}},\\
		&= V_{\tau_2}^A\brak{\xi},
		\end{split}
		\end{align}
		and then 
		\begin{align} \label{eq : first ineq}
		V_{\tau_1}^A\brak{\xi} &\leq \esssup_{\P^{\a}\in\Pc}\E_{\tau_1}^{\P^{\a}}\edg{V_{\tau_2}^A\brak{\xi}e^{\etaa\int_{\tau_1}^{\tau_2}\brak{c^{A}(X_r,\a_r)-s(X_r)}dr}}.
		\end{align}
		We proceed next to prove the second inequality. Consider $\a\in\Uc$ and $\nu\in\Uc_{\tau_2}$. Define then the concatenation of the two controls for $0\leq t \leq T$ as $\brak{\a\otimes_{\tau_2}\nu}_t:= \a_t\mathds{1}_{0\leq t< \tau_2} + \nu_t\mathds{1}_{\tau_2 \leq t \leq T}$, where $\tau_2$ is an $\F$--stopping time.\\
		
		\no We have then $(\a\otimes_{\tau_2}\nu)\in\Uc$ and by definition of the essential supremum (where we denote $\E_{\tau_1}^{\a\otimes_{\tau_2}\nu}$ instead of $\E_{\tau_1}^{\P^{\a\otimes_{\tau_2}\nu}}$):
		\begin{align}
		\begin{split}
		V_{\tau_1}^A\brak{\xi} &\geq \E_{\tau_1}^{\a\otimes_{\tau_2}\nu}\edg{-e^{-\etaa\brak{\xi-\int_{\tau_1}^T\brak{c^{A}(X_r,(\a\otimes_{\tau_2}\nu)_r)-s(X_r)}dr}}},\\
		&=\E_{\tau_1}^{\a\otimes_{\tau_2}\nu}\edg{-e^{-\etaa\brak{-\int_{\tau_1}^{\tau_2}\brak{c^{A}(X_r,\a_r)-s(X_r)}dr-\int_{\tau_2}^T\brak{c^{A}\brak{X_r,\nu_r}-s(X_r)}dr}}e^{-\etaa\xi}},\\
		&=\E_{\tau_1}^{\a\otimes_{\tau_2}\nu}\edg{e^{\etaa\int_{\tau_1}^{\tau_2}\brak{c^{A}(X_r,\a_r)-s(X_r)}dr}\E_{\tau_2}^{\a\otimes_{\tau_2}\nu}\edg{-e^{-\etaa\brak{\xi-\int_{\tau_2}^T\brak{c^{A}\brak{X_r,\nu_r}-s(X_r)}dr}}}}.
		\end{split}
		\end{align}
		\no Using again Bayes formula on the conditional expectation w.r.t $\Fc_{\tau_2}$, we have that 
		\begin{align}
		\E_{\tau_2}^{\a\otimes_{\tau_2}\nu}\edg{-e^{-\etaa\brak{\xi-\int_{\tau_2}^T\brak{c^{A}\brak{X_r,\nu_r}-s(X_r)}dr}}}&=
		\E_{\tau_2}^{0}\edg{-\frac{L_{T}^{\a\otimes_{\tau_2}\nu}}{L_{\tau_2}^{\a\otimes_{\tau_2}\nu}}e^{-\etaa\brak{\xi-\int_{\tau_2}^T\brak{c^{A}\brak{X_r,\nu_r}-s(X_r)}dr}}}.
		\end{align}
		Now notice that $\frac{L_{T}^{\a\otimes_{\tau_2}\nu}}{L_{\tau_2}^{\a\otimes_{\tau_2}\nu}}=\frac{L_{T}^{\nu}}{L_{\tau_2}^{\nu}}$ (as stated earlier the change of measure applied to the conditional expectation depends only on the control after $\tau_2$). We have therefore 
		\begin{align}
		\begin{split}
		\E_{\tau_2}^{\a\otimes_{\tau_2}\nu}\edg{-e^{-\etaa\brak{\xi-\int_{\tau_2}^T\brak{c^{A}\brak{X_r,\nu_r}-s(X_r)}dr}}}&=
		\E_{\tau_2}^{0}\edg{-\frac{L_{T}^{\nu}}{L_{\tau_2}^{\nu}}e^{-\etaa\brak{\xi-\int_{\tau_2}^T\brak{c^{A}\brak{X_r,\nu_r}-s(X_r)}dr}}},\\
		&=J_{\tau_2}^A\brak{\xi,\P^{\nu}}.
		\end{split}
		\end{align} 
		Thus we obtain the following inequality 
		\begin{align}
		V_{\tau_1}^A\brak{\xi} &\geq \E_{\tau_1}^{\a\otimes_{\tau_2}\nu}\edg{e^{\etaa\int_{\tau_1}^{\tau_2}\brak{c^{A}(X_r,\a_r)-s(X_r)}dr}J_{\tau_2}^A\brak{\xi,\P^{\nu}}}.
		\end{align}
		We use again Bayes Formula for the change of measure and the tower property of conditional expectation leading to 
		\begin{align}
		\begin{split}
		V_{\tau_1}^A\brak{\xi} &\geq \E_{\tau_1}^{0}\edg{\frac{L_T^{\a\otimes_{\tau_2}\nu}}{L_{\tau_1}^{\a\otimes_{\tau_2}\nu}}e^{\etaa\int_{\tau_1}^{\tau_2}\brak{c^{A}(X_r,\a_r)-s(X_r)}dr}J_{\tau_2}^A\brak{\xi,\P^{\nu}}},\\
		&=\E_{\tau_1}^{0}\edg{\E_{\tau_2}^{0}\edg{\frac{L_T^{\a\otimes_{\tau_2}\nu}}{L_{\tau_1}^{\a\otimes_{\tau_2}\nu}}\frac{L_{\tau_2}^{\a\otimes_{\tau_2}\nu}}{L_{\tau_2}^{\a\otimes_{\tau_2}\nu}}e^{\etaa\int_{\tau_1}^{\tau_2}\brak{c^{A}(X_r,\a_r)-s(X_r)}dr}J_{\tau_2}^A\brak{\xi,\P^{\nu}}}},\\
		&=\E_{\tau_1}^{0}\edg{\E_{\tau_2}^{0}\edg{\frac{L_T^{\a\otimes_{\tau_2}\nu}}{L_{\tau_2}^{\a\otimes_{\tau_2}\nu}}}\frac{L_{\tau_2}^{\a\otimes_{\tau_2}\nu}}{L_{\tau_1}^{\a\otimes_{\tau_2}\nu}}e^{\etaa\int_{\tau_1}^{\tau_2}\brak{c^{A}(X_r,\a_r)-s(X_r)}dr}J_{\tau_2}^A\brak{\xi,\P^{\nu}}},\\
		&=\E_{\tau_1}^{0}\edg{\frac{L_{\tau_2}^{\a\otimes_{\tau_2}\nu}}{L_{\tau_1}^{\a\otimes_{\tau_2}\nu}}e^{\etaa\int_{\tau_1}^{\tau_2}\brak{c^{A}(X_r,\a_r)-s(X_r)}dr}J_{\tau_2}^A\brak{\xi,\P^{\nu}}}.
		\end{split}
		\end{align}
		Now recall that for $0\leq t\leq \tau_2$ we have by definition $\brak{\a\otimes_{\tau_2}\nu}_t=\a_t$, and therefore $\frac{L_{\tau_2}^{\a\otimes_{\tau_2}\nu}}{L_{\tau_1}^{\a\otimes_{\tau_2}\nu}}=\frac{L_{\tau_2}^{\a}}{L_{\tau_1}^{\a}}$ leading to 
		\begin{align}
		V_{\tau_1}^A\brak{\xi} &\geq \E_{\tau_1}^{0}\edg{\frac{L_{\tau_2}^{\a}}{L_{\tau_1}^{\a}}e^{\etaa\int_{\tau_1}^{\tau_2}\brak{c^{A}(X_r,\a_r)-s(X_r)}dr}J_{\tau_2}^A\brak{\xi,\P^{\nu}}},\nonumber \\
		&= \E_{\tau_1}^{\a}\edg{e^{\etaa\int_{\tau_1}^{\tau_2}\brak{c^{A}(X_r,\a_r)-s(X_r)}dr}J_{\tau_2}^A\brak{\xi,\P^{\nu}}}.\label{eq: for dpp}
		\end{align}
		The inequality \eqref{eq: for dpp} holds for $\a\in\Uc$ and $\nu\in\Uc_{\tau_2}$, we can then by virtue of \rmi choose a sequence $\brak{\nu^n}_{n\in\N}$ of controls in $\Uc_{\tau_2}$ such that 
		\begin{align}
		V_{\tau_2}^A\brak{\xi} = \lim_{n\rightarrow+\infty}\uparrow J_{\tau_2}^{A}\brak{\xi,\P^{{\nu}^n}},
		\end{align}
		then we have by the monotone convergence theorem that for $\a\in\Uc$ 
		\begin{align}
		\begin{split}
		V_{\tau_1}^A\brak{\xi} &\geq \lim_{n\rightarrow+\infty}\uparrow\E_{\tau_1}^{\a}\edg{e^{\etaa\int_{\tau_1}^{\tau_2}\brak{c^{A}(X_r,\a_r)-s(X_r)}dr}J_{\tau_2}^A\brak{\xi,\P^{{\nu}^n}}}\\
		&=\E_{\tau_1}^{\a}\edg{e^{\etaa\int_{\tau_1}^{\tau_2}\brak{c^{A}(X_r,\a_r)-s(X_r)}dr}\lim_{n\rightarrow+\infty}\uparrow J_{\tau_2}^A\brak{\xi,\P^{{\nu}^n}}}\\
		&=\E_{\tau_1}^{\a}\edg{e^{\etaa\int_{\tau_1}^{\tau_2}\brak{c^{A}(X_r,\a_r)-s(X_r)}dr}V_{\tau_2}^A\brak{\xi}},
		\end{split}
		\end{align}
		concluding the proof of Lemma \ref{lemma : DPP}.
		\qed
	\end{proof}
	\subsubsection*{Proof of Proposition \ref{Proposition : existence BSDE}}
	Now that we proved the Dynamic Programming Principle, we move to the existence of the BSDE.  We have by definition $\Zc\subset\Xi$. To prove the second inclusion, we fix some $\xi\in\Xi$, and define agent's continuation utility as in \eqref{eq : continuation utility}.\\ 
	
	By virtue of Lemma \ref{lemma : DPP}, the family $\brak{V_{\tau}^A\brak{\xi}e^{\etaa\int_0^{\tau}\brak{c^{A}\brak{X_r,\a_r}-s(X_r)}dr}}_{\tau\in\Tc_{0,T}}$ is a $\brak{\P^{\a},\F}$--supermartingale system. Therefore, by the results of \cite{LD81}, it can be aggregated by a unique $\F$-optional process up to indistinguishability, which coincides with $\brak{V_t^A\brak{\xi}e^{\etaa\int_0^t\brak{c^{A}\brak{X_r,\a_r}-s(X_r)}dr}}_{t\in\edg{0,T}}$ and remains a $\brak{\P^{\a},\F}$--supermartingale, which then admits a c\`ad-l\`ag modification since the filtration considered satisfies the usual conditions.\\
	
	Then, from the admissibility constraint of the contract $\xi$; that is $\Pc^{\star}\brak{\xi}\neq\emptyset$, there exists some probability measure $\P^{\a^{\star}\brak{\xi}}$ (referred to as $\P^{\a^\star}$ to ease notations) such that $V_t^A\brak{\xi} = J_{t}^{A}\brak{\xi,\ \P^{\a^\star}}$, for $t\in\edg{0,T}$, and so the process $\brak{J_{t}^{A}\brak{\xi,\P^{\a^\star}}e^{\etaa\int_0^t\brak{c^{A}\brak{X_r,\a_r}-s(X_r)}dr}}_{t\in\edg{0,T}}$ is a $\brak{\P^{\a},\F}$--supermartingale, for $\P^{\a}\in\Pc$, while the processes $\brak{J_{t}^{A}\brak{\xi, \P^{\a^\star}}e^{\etaa\int_0^t\brak{c^{A}\brak{X_r,\a^{\star}_r}-s(X_r)}dr}}_{t\in\edg{0,T}}$ is a $\brak{\P^{{\a}^{\star}},\F}$--UI martingale.\\
	In fact the integrability is guaranteed by the ceiling function in $c^A$ and $s$, together with the admissibility condition on the contract $\xi$; for $t\in[0,T]$
	\begin{align}
	\begin{split}
		\E^{\P^{\a^{\star}}}\edg{\abs{J_{t}^{A}\brak{\xi, \P^{\a^\star}}e^{\etaa\int_0^t\brak{c^{A}\brak{X_r,\a_r}-s(X_r)}dr}}}&=\E^{\P^{\a^{\star}}}\edg{\E_{t}^{\P^{\a^{\star}}}\edg{e^{-\etaa\brak{\xi+\int_{0}^{T}\brak{s(X_r)-c^{A}\brak{X_r,(\a\otimes_{t}\a^{\star})_r}}dr}}}},\\
		&\leq K\E^{\P^{\a^{\star}}}\edg{\E_{t}^{\P^{\a^{\star}}}\edg{e^{-\etaa\xi}}}=K\E^{\P^{\a^{\star}}}\edg{e^{-\etaa\xi}}<+\infty,
	\end{split}
	\end{align}
	where we used again the ceiling function (and thus the boundedness of the exponential term), together with the admissibility condition on $\xi$.
	On the other hand, by the super-martingale inequality and the tower property of conditional expectations, we have for every $t_1\leq t_2 \in\edg{0,T}$:
	
	\begin{align}
	\begin{split}
	J_{t_1}^{A}\brak{\xi,\P^{\a^{\star}}}e^{\etaa\int_0^{t_1}\brak{c^{A}\brak{X_r,\a_r^{\star}}-s(X_r)}dr}&\geq\E^{\P^{\a^{\star}}}_{t_1}\edg{J_{t_2}^{A}\brak{\xi,\P^{\a^{\star}}}e^{\etaa\int_0^{t_2}\brak{c^{A}\brak{X_r,\a_r^{\star}}-s(X_r)}dr}},\\&=\E^{\P^{\a^{\star}}}_{t_1}\edg{\E^{\P^{\a^{\star}}}_{t_2}\edg{U_A\brak{\xi - \int_0^T\brak{c^{A}\brak{X_r,\a_r^{\star}}-s(X_r)}dr}}},\\&=\E^{\P^{\a^{\star}}}_{t_1}\edg{U_A\brak{\xi - \int_0^T\brak{c^{A}\brak{X_r,\a_r^{\star}}-s(X_r)}dr}},\\
	&=J_{t_1}^{A}\brak{\xi,\P^{\a^{\star}}}e^{\etaa\int_0^{t_1}\brak{c^{A}\brak{X_r,\a_r^{\star}}-s(X_r)}dr}.
	\end{split}
	\end{align}
	Therefore all the previous terms are equal a.s., in particular, for $t\in\edg{0,T}$ 
	{\small
		\begin{align}
		J_{t}^{A}\brak{\xi,\P^{\a^{\star}}}e^{\etaa\int_0^{t}\brak{c^{A}\brak{X_r,\a_r^{\star}}-s(X_r)}dr}=\E^{\P^{\a^{\star}}}_{t}\edg{U_A\brak{\xi +\int_0^T\brak{s(X_r)-c^{A}\brak{X_r,\a_r^{\star}}}dr}},
		\end{align}
	}
	which proves that $\brak{J_{t}^{A}\brak{\xi,\P^{\a^{\star}}}e^{\etaa\int_0^{t}\brak{c^{A}\brak{X_r,\a_r^{\star}}-s(X_r)}dr}}_{t\in\edg{0,T}}$is a $\P^{\a^{\star}}$--closed martingale, with a terminal value at $T$ given by
	\begin{align}
	J_T^{A}\brak{\xi,\P^{\a^{\star}}}e^{\etaa\int_0^T\brak{c^{A}\brak{X_r,\a_r^{\star}}-s(X_r)}dr}
	&=-e^{-\etaa\brak{ \xi-\int_0^T\brak{c^{A}\brak{X_r,\a_r^{\star}}-s(X_r)}dr}}.
	\end{align}
	Then, by the martingale representation theorem, there exists a predictable process $\tilde{Z}\in\mathbb{H}^2_{loc}$ valued in $\R^2$ such that 
	\begin{align}
	\begin{split}
	J_{t}^{A}\brak{\xi,\P^{\a^{\star}}}e^{\etaa\int_0^{t}\brak{c^{A}(X_r,\a_r^{\star})-s(X_r)}dr}&=J_{0}^{A}\brak{\xi,\P^{\a^{\star}}}+\int_0^t\tilde{Z}_sdW_s^{\a^{\star}},\\
	&=J_{0}^{A}\brak{\xi,\P^{\a^{\star}}}\Ec\brak{-\etaa\int_0^tZ_r\cdot \sigma\brak{X_r}dW_r^{\a^{\star}}},    
	\end{split}
	\end{align}
	where 
	\begin{align}
	Z_t := -\frac{\sigma^{-1}(X_t)\tilde{Z_t}}{\etaa J_t^{A}\brak{\xi,\P^{\a^{\star}}}e^{\etaa\int_0^{t}\brak{c^{A}(X_r,\a_r^{\star})-s(X_r)}dr}},
	\end{align}
	and
	\begin{align}
		\Ec\brak{-\etaa\int_0^tZ_r\cdot \sigma\brak{X_r}dW_r^{\a^{\star}}}_{t\in[0,T]}\mbox{ is a }\brak{\P^{{\a}^{\star}},\F}\mbox{--UI martingale}.
	\end{align}
	We define 
	\begin{align}
	Y_t=U_A^{-1}\brak{V_t^A\brak{\xi}},\mbox{ for }t\in[0,T],
	\end{align}
	and our goal is to prove that the pair $(Y,Z)$ is a solution to \eqref{eq : equation BSDE}, and that $Z\in \Vc$.\\  
	\no Recall that for an arbitrary $\a\in\Uc$, the process  $\tilde{Y}_t^{\a}:=J_{t}^{A}\brak{\xi,\P^{\a^{\star}}}e^{\etaa\int_0^{t}\brak{c^{A}(X_r,\a_r)-s(X_r)}dr}$ is a $\brak{\P^{\a},\F}$--supermartingale. Replacing $J_{t}^{A}\brak{\xi,\P^{\a^{\star}}}$ by its representation, we obtain
	\begin{align}
	\frac{\tilde{Y}_t^{\a}}{\etaa}&=\frac{1}{\etaa}J_{0}^{A}\brak{\xi,\P^{\a^{\star}}}\Ec\brak{-\etaa\int_0^tZ_r\cdot \sigma\brak{X_r}dW_r^{\a^{\star}}}{e^{\etaa\int_0^{t}\brak{c^{A}\brak{X_r,\a_r} - c^{A}(X_r,\a_r^{\star})}dr}}.
	\end{align}
	\no We apply then It\^o formula and Girsanov Theorem, therefore
	\begin{align}
	\begin{split}
	\frac{d\tilde{Y}_t^{\a}}{\etaa\tilde{Y}_t^{\a}}&=\brak{c^{A}\brak{X_t,\a_t} - c^{A}\brak{X_t,\a_t^{\star}}}dt-Z_t\cdot \sigma\brak{X_t}dW_t^{\a^{\star}},\\
	&=-Z_t\cdot \sigma\brak{X_t}dW_t^{\a}-\crl{\brak{Z_t\cdot\mu\brak{X_t,\a_t}-c^{A}\brak{X_t,\a_t}}-\brak{Z_t\cdot\mu\brak{X_t,\a_t^{\star}}-c^{A}\brak{X_t,\a_t^{\star}}}}dt,
	\end{split}
	\end{align}
	and by the supermartingale property and the sign of $\tilde{Y}_t^{\a}$ we conclude that 
	\begin{align}
	{\a}^{\star}\in\mbox{argmax}\brak{Z_t\cdot\mu\brak{X_t,\a_t}-c^{A}\brak{X_t,\a_t}}.
	\end{align}
	Finally, applying It\^o Formula 
	\begin{align}
	\begin{split}
	Y_t=&U_A^{-1}\brak{J_{0}^{A}\brak{\xi,\P^{\a^{\star}}}}+\int_0^tZ_r\cdot \sigma\brak{X_r}dW_r^{\a^{\star}}+\int_0^t(c^{A}\brak{X_r,\a_r^{\star}}-s(X_r))dr,\\
	=&U_A^{-1}\brak{J_{0}^{A}\brak{\xi,\P^{\a^{\star}}}}+\int_0^tZ_r\cdot dX_r-\int_0^t\brak{Z_r\cdot\mu\brak{X_r,\a_r^{\star}}+s(X_r)-c^{A}\brak{X_r,\a_r^{\star}}-\etaa|\sigma(X_r)Z_r|^2}dr,
	\end{split}
	\end{align}
	and so the pair $\brak{Y,Z}$, satisfy \eqref{eq : equation BSDE}. Furthermore, following the line of proof of \cite{briandh}[corollary 4] and using the integrability assumption on admissible contracts, we obtain
	\begin{align}
	\E^{\P^{\alpha^{\star}}}\edg{e^{(\etaa\lor\etap) (1+\delta) \sup_{t\in[0,T]}\abs{Y_t}}}<+\infty,
	\end{align}
	with $\P^{\alpha^{\star}}=\P^{\hat{\alpha}}$ from proposition \ref{Agent's response}, which concludes the proof. 
	\qed

	\subsection{Solving consumer's problem : optimal contract and capacity payment}\label{appendix : decomposition contrat}
	From the Proposition \ref{Proposition : existence BSDE}, consumer's problem is reduced to 
	\begin{align}\tag{\ref{eq : reduced principal problem}}
	V_0^{P}
	&=\sup_{Y_0\geq \Rc}\sup_{Z \in \Vc}J_{0}^{P}\brak{Y_T^{Y_0,Z},\P^{\hat{\a}}},
	\end{align}
	so the optimal contract we are looking for is of the form $\xi^{\star}=Y_T^{Y^{\star}_0,Z^{\star}}$, with a pair $(Y_0^\star,Z^{\star})\in [\Rc,+\infty[\x\Vc$.
	Using the identity $Y^{Y_0,Z}=Y_0+Y^{0,Z}$, we have
	\begin{align}
	\begin{split}
	V_0^{P}
	&=\sup_{Y_0\geq \Rc}\sup_{Z \in \Vc}\edg{-e^{-\etap\brak{-Y_T^{Y_0,Z}-S_T+\int_{0}^{T}c^{P}\brak{X_{t}}dt}}},\\
	&=\sup_{Y_0\geq \Rc}e^{\etap Y_0}\sup_{Z \in \Vc}\edg{-e^{-\etap\brak{-Y_T^{0,Z}-S_T+\int_{0}^{T}c^{P}\brak{X_{t}}dt}}},\\
	&=e^{\etap \Rc}\sup_{Z \in \Vc}\edg{-e^{-\etap\brak{-Y_T^{0,Z}-S_T+\int_{0}^{T}c^{P}\brak{X_{t}}dt}}},
	\end{split}
	\end{align}
	and therefore $Y_0^\star = \Rc$. We can rewrite \eqref{eq : reduced principal problem} as
	\begin{align}
	V_{0}^{P} & = \sup_{Z\in\Vc}\E^{\P^{\hat{\a}}}\edg{U_{P}\brak{-Y_{T}^{\Rc,Z}-\int_{0}^{T}s(X_t)dt+\int_{0}^{T}{c^{P}\brak{X_{t}}dt}}},
	\end{align}
	with the state variables following agent's optimal response (which we recall is the same as principal's recommendation) i.e., $\a_t^{\star} = \hat{\a}\brak{X^C_t,Z^C_t}$, for $t\in\edg{0,T}$ : 
	\begin{align}
		\begin{cases}
		X_{t}&=\brak{\begin{array}{c}
			x_0^C\\
			x_0^D
			\end{array}} + \int_{0}^{t}\mu\brak{X_r,\hat{\a}\brak{X^C_r,Z^C_r}}dr+\int_{0}^{t}\sigma\brak{X_r}dW^{\hat{\a}}_{r},\\
		Y_{t}&=\Rc + \int_{0}^{t}\brak{c^{A}\brak{X_r,\hat{\a}\brak{X^C_r,Z^C_r}}+\frac{\eta_{A}}{2}|\sigma\brak{X_r} Z_r|^2-s(X_r)}dr+\int_{0}^{t}Z_r\cdot\sigma\brak{X_r}dW^{\hat{\a}}_r.
		\end{cases}
	\end{align}
	Define then the continuation utility as
	\begin{align} \label{VP-markov}
	V^{P}\brak{t,x,y}:=\sup_{Z\in\Vc_t}\E_{t,x,y}^{\P^{\hat{\a}}}\edg{-e^{-\etap\brak{-Y_{T}-\int_{t}^{T}s(X_r)dr+\int_{t}^{T}{c^{P}\brak{X_r}dr}}}},
	\end{align}
	Observe that $Z = 0$ is an admissible control and therefore
	\begin{align}
	\begin{split}
	V^{P}\brak{t,x,y}&\geq \E_{t,x,y}^{\P^{\hat{\a}(0)}}\edg{-e^{\etap Y_{T}+\etap\int_{t}^{T}s(X_r)dr-\etap\int_{t}^{T}{c^{P}\brak{X_r}dr}}},\\
	&=-e^{\etap y}\E_{t,x,y}^{\P^{\hat{\a}(0)}}\edg{e^{\etap \int_{t}^T\brak{c^{A}\brak{X_r,\hat{\a}\brak{X^C_r,0}} - c^{P}(X_r)}dr}},\\
	&\geq -e^{\etap y + C(T-t)}\mbox{ for some constant }C,
	\end{split}
	\end{align}
	where the last inequality follows from the bound
	\begin{align}
	-kx_{\infty} - c^{A}(x_{\infty},\alpha_{max}) \leq c^P(x) - c^{A}(x,\alpha) \leq x_{\infty}(\theta - \kappa_1\alpha_{min}),\mbox{ for }(x,\alpha)\in \R^2_+\x[\alpha_{min},\alpha_{max}].
	\end{align}
	On the other hand, we have
	{\small
	\begin{align}
	\begin{split}
	\E_{t,x,y}^{\P^{\hat{\a}}}\edg{-e^{-\etap\brak{-Y_{T}-\int_{t}^{T}s(X_r)dr+\int_{t}^{T}{c^{P}\brak{X_r}dr}}}} &= -e^{\etap y}\E_{t,x,y}^{\P^{\hat{\a}}}\edg{e^{\etap\int_{t}^{T}\brak{c^{A}(X_r,\hat{\a}(X^C_r,Z^C_r))-c^{P}\brak{X_r}+\frac{\eta_{A}}{2}|\sigma\brak{X_r} Z_r|^2}dr+\etap\int_t^TZ_r\cdot\sigma\brak{X_r}dW^{\hat{\a}}_r}},\\
	&\leq-e^{\etap y+x_{\infty}\etap(\kappa_1\alpha_{min}-\theta) (T-t)}\E_{t,x,y}^{\P^{\hat{\a}}}\edg{e^{ \etap\int_t^TZ_r\cdot\sigma\brak{X_r}dW^{\hat{\a}}_r}},\\
	&\leq -e^{\etap y+x_{\infty}\etap(\kappa_1\alpha_{min}-\theta) (T-t)}, 
	\end{split}
	\end{align}}
	where the last inequality follows from the Jensen inequality and the concavity of $x\mapsto -e^{x}$. Therefore, taking the supremum over $Z\in\Vc_t$, we obtain
	\begin{align}
		V^{P}\brak{t,x,y}\leq -e^{\etap y+x_{\infty}\etap(\kappa_1\alpha_{min}-\theta) (T-t)}.
	\end{align}
	We can see then that $\abs{V^{P}\brak{t,x,y}}\leq e^{\etap y +C(T-t)}$ for some constant $C$.
	By standard stochastic control theory (\cite{touzi12}), 
	$V^{P}$ is characterized as the unique viscosity solution of the HJB equation :
	\begin{align}\label{eq : HJB Principal 1}
	\begin{cases}
		-\partial_{t}V^{P}-G\brak{x,V^P,DV^P,D^{2}V^P}=0,\mbox{ for }\brak{t,x,y} \in [0,T)\x\R^3,\\
		V^{P}\brak{T,x,y}=-\exp\brak{{\etap y}},\mbox{ for }\brak{x,y} \in \R^3,
	\end{cases}
	\end{align}
	with growth $\abs{V^{P}\brak{t,x,y}}\leq e^{\etap y +C(T-t)}$ for some constant $C$, and $G : \R^2\x\R\x\R^3\x\Mc_3\brak{\R}\rightarrow\R$ the hamiltonian defined as
	{\small
		\begin{align}
		G\brak{x,q,p,\gamma}  :=& \sup_{z\in\R^2}g\brak{x,q,p,\gamma,z},
		\end{align}
	 with
	{\small
		\begin{align}
		\begin{split}
		g\brak{x,q,p,\gamma,z}  :=& \left\lbrace\mu\brak{x,\hat{\a}(x^C,z^C)}\cdot p_{x}+\brak{c^{A}\brak{x,\hat{\a}(x,z)}+\frac{\etaa}{2}z^{\intercal}\sigma\sigma^{\intercal}(x)z-s(x) }p_{y}+\etap q(s(x)- c^{P}\brak{x})\right.\\
		& \left.+\frac{1}{2}\sigma\sigma^{\intercal}\brak{x}:\gamma_{xx}+\frac{1}{2}z^{\intercal}\sigma\sigma^{\intercal}\brak{x}z\gamma_{yy}+z^{\intercal}\sigma\sigma^{\intercal}\brak{x}\gamma_{xy}\right\rbrace.
		\end{split}
		\end{align}}
	\no Using the change of variable $V^P(t,x,y)=-e^{\etap (y - u(t,x))}$, we can simplify the PDE \eqref{eq : HJB Principal 1}, and express consumer's value function with a 2-dimensional state variable equation instead of 3, that is $u:\edg{0,T}\x\R^2\rightarrow\R$, which is the unique bounded viscosity solution to
	\begin{align}\label{eq : pde_u}
	\begin{cases}
	-\partial_{t}u-\bar{G}\brak{x,Du,D^{2}u}=0,\mbox{ for }\brak{t,x} \in [0,T)\x\R^2,\\
	u\brak{T,x}=0,\mbox{ for }x \in \R^2,
	\end{cases}
	\end{align}	
	where the boundedness is obtained from the growth condition of $V^P$, and with $\bar{G} : \R^2\x\R^2\x\Mc_2\brak{\R}\rightarrow\R$ defined as
	\begin{align}
		\bar{G}(x,p,\gamma):=\sup_{z\in\R^2}\bar{g}\brak{x,p,\gamma,z},
	\end{align}
	and
	{\small
		\begin{align}
		\bar{g}\brak{x,p,\gamma,z}  :=& \mu\brak{x,\hat{\a}(x^C,z^C)}\cdot p+ c^{P}\brak{x}-c^{A}\brak{x,\hat{\a}(x^C,z^C)}+\frac{1}{2}\sigma\sigma^{\intercal}\brak{x}:\brak{\gamma-\etap pp^{\intercal}}-\frac{\etaa+\etap}{2}z^{\intercal}\sigma\sigma^{\intercal}\brak{x}z+\etap z^{\intercal}\sigma\sigma^{\intercal}\brak{x}p,
		\end{align}}
	which is strictly concave (and separable) in the control variables $z^C$ and $z^D$
	\begin{align}
	\begin{split}
	\bar{g}\brak{x,p,\gamma,z}=& \tilde{\mu}(x)p + c^{P}\brak{x}-\tilde{c}^{A}(x)+\frac{1}{2}\sigma\sigma^{\intercal}\brak{x}:\brak{\gamma-\etap pp^{\intercal}},\\
	&-\frac{\kappa_2\brak{\underline{x}^{C}}^2}{2}(\hat{\a}(x^C,z^C) )^{2}+\hat{\a}(x^C,z^C) \brak{x^{C}p_{x^C}-\kappa_1 \underline{x}^{C}} -\frac{\etaa+\etap}{2}\brak{z^C}^2\brak{\sigma^Cx^C}^2+z^C\etap \brak{\sigma^Cx^C}^2p_{x^C},\\
	&-\frac{\etaa+\etap}{2}\brak{z^D}^2\brak{\sigma^Dx^D}^2+z^D\etap \brak{\sigma^Dx^D}^2p_{x^D},
	\end{split}
	\end{align}
and so 
\begin{align}
	\bar{G}(x,p,\gamma) = \bar{g}\brak{x,p,\gamma,\hat{z}(x,p)},
\end{align}
with the maximizer $\hat{z}$, that can be expressed explicitly as
\begin{align}\label{eq : form of optimal control}
	\hat{z}(x,p):= \brak{\begin{array}{cc}
		\hat{z}^{C}(x,p)\\
		\hat{z}^{D}(x,p)
		\end{array}},
\end{align}
where for sufficiently large $\abs{\alpha_{max}}$ and $\abs{\alpha_{min}}$,
\begin{align}\label{eq : control xc}
	\hat{z}^{C}(x,p) := \frac{\etap(\sigma^Cx^C)^2 + \frac{1}{\kappa_2}\brak{\frac{x^C}{\underline{x}^C}}^2}{(\etaa+\etap)(\sigma^Cx^C)^2+\frac{1}{\kappa_2}\brak{\frac{x^C}{\underline{x}^C}}^2}p_{x^C} + \frac{\kappa_1}{\kappa_2}\frac{\brak{\frac{x^C}{\underline{x}^C}}-\brak{\frac{x^C}{\underline{x}^C}}^2}{(\etaa+\etap)(\sigma^Cx^C)^2+\frac{1}{\kappa_2}\brak{\frac{x^C}{\underline{x}^C}}^2},
\end{align}
and 
\begin{align}
	\hat{z}^{D}(x,p) := \frac{\etap}{\brak{\etaa+\etap}}p_{x^D}.
\end{align}
\begin{Remark}
	{\upshape
		Whenever the truncation coefficients go to infinity $x_{\infty},\alpha_{max}\rightarrow +\infty$ and $\alpha_{min}\rightarrow -\infty$, we obtain
		\begin{align}\label{eq : hat z approx}
		\hat{z}\brak{x,p}\approx \brak{\begin{array}{cc}
			\frac{\etap\brak{\sigma^{C}x^{C}}^{2}+\frac{1}{\kappa_2}}{\brak{\etaa+\etap}\brak{\sigma^{C}x^{C}}^{2}+\frac{1}{\kappa_2}} & 0\\
			0 & \frac{\etap}{\brak{\etaa+\etap}}
			\end{array}}p,
		\end{align}
		which can be used as an approximation of \eqref{eq : form of optimal control}. Similarly, the PDE \eqref{eq : pde_u} can be approximated by the solution to \eqref{eq : pde_u approximation}
		{\footnotesize
			\begin{align}
			\begin{cases}\label{eq : pde_u approximation}
			\partial_{t}u+\brak{\tilde{\mu}\brak{x}-{ \brak{
						\begin{array}{c}
						\frac{\kappa_1}{\kappa_2}\\
						0
						\end{array}}}}\cdot Du+\frac{1}{2}\sigma\sigma^{\intercal}\brak{x}:D^{2}u+f\brak{x}{-\frac{\kappa^2_1}{2 \kappa_2}}+\frac{1}{2}\rho\brak{x}\cdot\brak{\begin{array}{c}\brak{\partial_{x^C}u}^{2}\\
				\brak{\partial_{x^D}u}^{2}
				\end{array}}=0\\
			u\brak{T,x}=0
			\end{cases},\mbox{ for }\brak{t,x}\in\edg{0,T}\x\R^2,
			\end{align}}
		where $f : \R^2\rightarrow\R$ is defined as $f:=\tilde{c}^{A}-c^{P}\!,$  and 
		\begin{align}
		\rho\brak{x}:=\brak{\begin{array}{c}\frac{\frac{1}{\kappa_2^2}+\frac{1}{\kappa_2}\etap\brak{\sigma^C x^{C}}^{2}-\etaa\etap\brak{\brak{\sigma^{C}x^{C}}^{2}}^{2}}{\brak{\etaa+\etap}\brak{\sigma^{C}x^{C}}^{2}+\frac{1}{\kappa_2}}\\
			\frac{-\etaa\etap\brak{\sigma^{D}x^{D}}^{2}}{\brak{\etaa+\etap}}
			\end{array}}.
		\end{align}}
\end{Remark}

	\no We next proceed by verification to solve the problem
\begin{Proposition}[Verification]\label{proposition : verification Principal reduced}
		\rmi Assume that \eqref{eq : pde_u} has a bounded  $C^{1,2}([0,T],\R^2)$ solution $u$. Then there exists a $C^{1,2}([0,T],\R^3)$ solution to the PDE \eqref{eq : HJB Principal 1} denoted $v$ with growth $\abs{v(t,x,y)}\leq  e^{\etap y +C(T-t)}$ for some constant $C$, which satisfies 
		\begin{align}
			V_0^P\leq v(0,x_0,\Rc).
		\end{align}
		\rmii Define
		\begin{align}\label{eq : hat_z}
		Z_t^{\star}:=\hat{z}\brak{X_t,Du\brak{t,X_t}}\mbox{ for }t\in[0,T],
		\end{align}
		and assume that $\brak{Z_t^{\star}}_{t\in[0,T]}\in\Vc$. Then
		\begin{align}
			V_0^P= v(0,x_0,\Rc),
		\end{align}
		and 
		\begin{align}
			\xi^{\star} := \Rc +\int_0^TZ_t^{\star}\cdot dX_t -\int_0^TH\brak{X_t,Z_t^{\star}}dt
		\end{align}
		is an optimal contract.		
	\end{Proposition}
	\begin{proof}
		\rmi Let $u$ be a bounded $C^{1,2}([0,T],\R^2)$ solution to the PDE \eqref{eq : pde_u}. Then defining $v(t,x,y):=-e^{\etap(y - u(t,x))}$, we can see that $v\in C^{1,2}([0,T],\R^3)$ as a composition of smooth functions, and that $v(T,x,y)=-e^{\etap y}$ and $v$ satisfies the growth condition $\abs{v(t,x,y)}\leq  e^{\etap y +C(T-t)}$ for some constant $C$ from the boundedness of $u$. By straightforward differentiation, we obtain that $v$ is a solution to $\eqref{eq : HJB Principal 1}$. Define then 
		\begin{align}
		\beta_{t_1,t_2}:=e^{\etap\int_{t_1}^{t_2}(s(X_r)-c^P\brak{X_r})dr}\mbox{ for }0\leq t_1\leq t_2\leq T,
		\end{align}
		and the sequence of stopping times
		\begin{align}
			T_n:=T\land\inf\crl{t\geq 0, \int_0^t(\beta_{0,r})^2\brak{\abs{\sigma(X_r)D_xv(r,X_r,Y_r)}^2 + \brak{D_yv(r,X_r,Y_r)}^2\abs{ \sigma(X_r)z_r}^2}dr\geq n}\mbox{ for }n\geq 1.
		\end{align}
		For an arbitrary control $z\in \Vc$, we apply It\^o formula and take the expectation under $\P^{\hat{\alpha}}$ to obtain
		\begin{align}
		\begin{split}
			\E^{\P^{\hat{\alpha}}}\edg{\beta_{t,T_n}v(T_n,X_{T_n},Y_{T_n})} =& v(t,x,y) +\E^{\P^{\hat{\alpha}}}\edg{ \int_t^{T_n}\beta_{t,r}\brak{\partial_tv + g\brak{X_r,v(r,X_r,Y_r),Dv(r,X_r,Y_r),D^2V(r,X_r,Y_r),z_r}}dr}\\
			&+\E^{\P^{\hat{\alpha}}}\edg{\int_t^{T_n}\beta_{t,r}D_xv(r,X_r,Y_r)\cdot \sigma(X_r)dW^{\hat{\a}}_r}+\E^{\P^{\hat{\alpha}}}\edg{\int_t^{T_n}\beta_{t,r}D_yv(r,X_r,Y_r)z_r\cdot \sigma(X_r)dW^{\hat{\a}}_r}.
		\end{split}
		\end{align}
		From the definition of the localizing sequence $\brak{T_n}_{n\ge 1}$, we obtain that
		\begin{align}
			0=\E^{\P^{\hat{\alpha}}}\edg{\int_t^{T_n}\beta_{t,r}D_xv(r,X_r,Y_r)\cdot \sigma(X_r)dW^{\hat{\a}}_r}=\E^{\P^{\hat{\alpha}}}\edg{\int_t^{T_n}\beta_{t,r}D_yv(r,X_r,Y_r)z_r\cdot \sigma(Xw_r)dW^{\hat{\a}}_r}.
		\end{align}
		Furthermore, since $v$ is a solution to the PDE \eqref{eq : HJB Principal 1}, and by the linearity of the expectation 
		\begin{align}\label{eq : verification capa ineq 1}
			\E^{\P^{\hat{\alpha}}}\edg{ \int_t^{T_n}\beta_{t,r}\brak{\partial_tv + g\brak{X_r,v(r,X_r,Y_r),Dv(r,X_r,Y_r),D^2V(r,X_r,Y_r),z_r}}dr}\leq 0,
		\end{align} 
		and so 
		\begin{align}\label{eq : verification capa ineq 2}
			\E^{\P^{\hat{\alpha}}}\edg{\beta_{t,T_n}v(T_n,X_{T_n},Y_{T_n})} \leq v(t,x,y).
		\end{align}
		When $n\rightarrow +\infty$, the following a.s convergence holds
		\begin{align}
			\beta_{t,T_n}v(T_n,X_{T_n},Y_{T_n})\rightarrow U_{P}\brak{-Y_{T}-\int_{t}^{T}s(X_r)dr+\int_{t}^{T}{c^{P}\brak{X_r}dr}},
		\end{align}
		 and for $n\geq 1$ and the growth condition of $v$ we have
		 \begin{align}
		 	\begin{split}
		 	\E^{\P^{\hat{\alpha}}}\edg{\abs{\beta_{t,T_n}v(T_n,X_{T_n},Y_{T_n})}}&\leq \E^{\P^{\hat{\alpha}}}\edg{\beta_{t,T_n}e^{\etap Y_{T_n} +C(T-T_n)}},\\
		 	&=\E^{\P^{\hat{\alpha}}}\edg{e^{\etap\int_{t}^{T_n}(s(X_r)-c^P\brak{X_r})dr}e^{\etap Y_{T_n} +C(T-T_n)}},\\
		 	&\leq \E^{\P^{\hat{\alpha}}}\edg{e^{ C(T-T_n)+\etap(T_n - t)x_{\infty}(P(0,x_{\infty})+k_1)}e^{\etap Y_{T_n}}},\\
		 	&\leq \E^{\P^{\hat{\alpha}}}\edg{e^{ C(T-T_n)+\etap(T_n - t)x_{\infty}(P(0,x_{\infty})+k_1)}e^{\etap \sup_{t\in[0,T]}\abs{Y_{t}}}}<+\infty\mbox{ for }n\geq 1, 
		 	\end{split}
		 \end{align}
		 where we used the bound $(s(x) - c^P(x)) \leq x_{\infty}(P(0,x_{\infty})+k_1)$ for $x\in \R^2_+$ and the estimate on $Y$ from Proposition \ref{Proposition : existence BSDE}. Therefore, by the dominated convergence theorem, 
		\begin{align}
			\E^{\P^{\hat{\alpha}}}\edg{U_{P}\brak{-Y_{T}-\int_{t}^{T}s(X_r)dr+\int_{t}^{T}{c^{P}\brak{X_r}dr}}}\leq v(t,x,y),
		\end{align}
		and
		\begin{align}
		\sup_{z\in\Vc}\E^{\P^{\hat{\alpha}}}\edg{U_{P}\brak{-Y_{T}-\int_{t}^{T}s(X_r)dr+\int_{t}^{T}{c^{P}\brak{X_r}dr}}}\leq v(t,x,y).
		\end{align}
		In particular for $t=0$, we have an upper bound for the value function $V_0^P$
		\begin{align}\label{eq : verification capa ineq 3}
			V_0^P\leq v(0,x_0,\Rc).
		\end{align}
		\rmii Assuming that $\brak{Z_t^{\star}}_{t\in[0,T]}\in \Vc$, we can go over the same steps as in \rmi with $Z^{\star}$ instead of an arbitrary control from $\Vc$. Therefore the inequalities \eqref{eq : verification capa ineq 1} and \eqref{eq : verification capa ineq 2} and \eqref{eq : verification capa ineq 3} become equalities with 
		\begin{align}
		V_0^P=v(0,x_0,\Rc),
		\end{align}
		so that the upper bound is reached for the control $Z^{\star}$ which is then an optimal feedback control to the problem \eqref{eq : reduced principal problem}, and therefore the optimal contract corresponds to terminal value of the controlled state variable $\brak{Y_t^{\Rc,Z^{\star}}}_{t\in[0,T]}$ which we denote $\xi^{\star}:=Y_T^{\Rc,Z^{\star}}$, and which satisfies $\xi^{\star}\in \Xi$ since $Y_0=\Rc$ and $Z^{\star}\in\Vc$, concluding the proof.
		\qed
	\end{proof}}

	\begin{Remark}
		{\upshape
	The HJB-PDE \eqref{eq : pde_u} (the same remark can be said about the approximating PDE\eqref{eq : pde_u approximation}) is a semi-linear parabolic PDE of second order; with a non-linearity in the gradient term. This PDE has a-priori no solution in the classical sense, \ie a smooth function (in $C^{1,2}([0,T],\R^2)$) with a clear definition of $\partial_tu$, $Du$, and $D^2u$ solution to \eqref{eq : pde_u} because of the non-linearity. However, the existence can be proved in a weaker sense, by taking the candidate $u(t,x):=y-U_p^{-1}(V^P(t,x,y))$ defined in \eqref{VP-markov}, which can be proved to be a viscosity solution to \eqref{eq : pde_u} (not necessarily smooth). In this case one cannot define $\partial_tu$, and $Du$, and $D^2u$, and a more technical approach is required. We make the assumption $u\in C^{1,2}([0,T],\R^2)$ to simplify the exposition. Nevertheless, \eqref{eq : pde_u} might satisfy some conditions (unknown to the authors) which insure the regularity of the viscosity solution defined above.}
\end{Remark}

	\subsection{Producer's participation constraint: the problem without capacity payment\label{appendix : participation constraint}}
	The problem \eqref{eq : uncontrolled producer} is a markovian stochastic control problem, and can be solved by classical techniques. We define producer's continuation utility function $\hat{V}^A: \edg{0,T}\x\R_+^2\rightarrow \R$ as
	\begin{align}
	\hat{V}^A\brak{t,x}:=\sup_{\P^{\a}\in\Pc}\E_{t,x}^{\P^{\a}}\edg{-e^{-\etaa\int_{t}^{T}(s(X_r)-c^A\brak{X_r,\a_r})dr}}.
	\end{align}
	Recall from the definition \eqref{eq : definition small s} and \eqref{eq : definition cost} of $s$ and $c^A$ that for $(x,\alpha)\in \R^2_+\x[\alpha_{min},\alpha_{max}]$ we have the following bounds
	\begin{align}\label{eq : bounds}
		0\leq s(x)\leq P(0,x_{\infty})x_{\infty},\mbox{ and }\alpha_{min}\kappa_1x_{\infty}\leq c^A\brak{x,\a}\leq c^A\brak{x_{\infty},\a_{max}},\mbox{ and }0\leq c^A\brak{x,0}\leq (a+b)x_{\infty},
	\end{align}
	so 
	\begin{align}
		\hat{V}^A\brak{t,x}\leq -e^{-\etaa(T-t)\brak{P(0,x_{\infty})x_{\infty} - \alpha_{min}\kappa_1x_{\infty}}},
	\end{align}
	and since $\alpha = 0$ is an admissible control, then
	\begin{align}
		\hat{V}^A\brak{t,x}\geq \E_{t,x}^{\P^{0}}\edg{-e^{-\etaa\int_{t}^{T}(s(X_r)-c^A\brak{X_r,0})dr}}\geq -e^{\etaa (T-t)(a+b)x_{\infty}},
	\end{align}
	so 
	\begin{align}
		\abs{\hat{V}^A\brak{t,x}}\leq e^{C(T-t)}\mbox{ for }(t,x)\in\edg{0,T}\x\R_+^2,\mbox{ and some constant }C.
	\end{align}
	We identify then a candidate for the solution of the problem with the unique viscosity solution to the HJB equation :  
	\begin{align}\label{eq : HJB PDE uncontrolled}
	\begin{cases}
	-\partial_t\hat{V}^A - \hat{H}\brak{x,\hat{V}^A,D\hat{V}^A,D^2\hat{V}^A} = 0,\mbox{ in }[0,T)\x\R^2,\\
	\hat{V}^A(T,.)=-1,
	\end{cases}
	\end{align}
	with growth controlled by $\abs{\hat{V}^A\brak{t,x}}\leq e^{C(T-t)}$ for some constant $C$, (in particular $\hat{V}^A$ is bounded),
	where $\hat{H}: \R^2\x\R^2\x\Mc_2\brak{\R}\rightarrow \R$ is the Hamiltonian of the producer acting on his own defined as
	\begin{align}
	\hat{H}\brak{x,y,z,\gamma}&:=\sup_{\alpha\in[\alpha_{min},\alpha_{max}]}\hat{h}\brak{x,y,z,\gamma,\alpha},	
	\end{align}
	with
	\begin{align}
		\hat{h}\brak{x,y,z,\gamma,\alpha}&:=\mu\brak{x,\a}\cdot z+\frac{1}{2}\sigma \sigma^{\intercal}:\gamma +\etaa (c^A\brak{x,\a}-s(x))y.
	\end{align}
	Remark that the maximum is attained for 
	\begin{align}\label{eq : optimal control uncontrolled agent}
	\a_{pc}^{\star}\brak{x,y,z}:=\alpha_{min}\lor\brak{-\Frac{z^Cx^C+\etaa y\kappa_1\underline{x}^C}{\etaa y (\underline{x}^C)^2\kappa_2}}\land \alpha_{max},\mbox{ for $\brak{x,y,z}\in\R_{+}^{2}\x\R_{+}\x\R^2$ and $x^C,y\neq 0$},
	\end{align}
	(where ``pc'' stands for participation constraint) and the PDE \eqref{eq : HJB PDE uncontrolled} can be written as 
	\begin{align}
	\begin{cases}
	-\partial_t\hat{V}^A-\hat{h}\brak{x,\hat{V}^A,D\hat{V}^A,D^2\hat{V}^A,\a_{pc}^{\star}\brak{x,\hat{V}^A,D\hat{V}^A}} = 0\mbox{ in }[0,T)\x\R^2_+,\\
	\hat{V}^A(T,.)=-1.
	\end{cases}
	\end{align}
	We next proceed by verification. 

	\begin{Proposition}\label{proposition : verification pb sans contrat}
		Assume that the PDE \eqref{eq : HJB PDE uncontrolled} has a $C^{1,2}([0,T],\R^2)$ solution $v$  with growth controlled by $\abs{v\brak{t,x}}\leq e^{C(T-t)}$ for some constant $C$, then 
		\begin{align}
			\hat{V}^{A}_0 = v(0,x_0),
		\end{align}
		and  $\brak{\a_{pc}^{\star}\brak{X_t,v(t,X_t),Dv(t,X_t)}}_{t\in[0,T]}$ with $\a_{pc}^{\star}$ defined in \eqref{eq : optimal control uncontrolled agent} is an optimal feedback control to the problem \eqref{eq : uncontrolled producer}.
	\end{Proposition}
	\begin{proof}
		Let $v$ be a $C^{1,2}([0,T],\R^2)$ solution to the PDE \eqref{eq : HJB PDE uncontrolled}, such that $\abs{v\brak{t,x}}\leq e^{C(T-t)}$ for some constant $C$ and define the sequence of stopping times
		\begin{align}
			T_n:=T\land\inf\crl{t\geq 0, \int_0^t\brak{\beta_{0,r}^{\alpha}}^2\abs{\sigma(X_r)Dv(r,X_r)}^2dr\geq n}\mbox{ for }n\geq 1.
		\end{align}
		For an arbitrary control $\alpha\in\Uc$, we denote
		\begin{align}
			\beta_{t_1,t_2}^{\alpha}:=e^{-\etaa\int_{t_1}^{t_2}(s(X_r)-c^A\brak{X_r,\a_r})dr}\mbox{ for }0\leq t_1\leq t_2\leq T,
		\end{align}
		and we have by applying It\^o formula and taking the expectation under $\P^{\alpha}$
		\begin{align}
		\begin{split}
			\E^{\P^{\a}}\edg{\beta_{t,T_n}^{\alpha}v(T_n,X_n)} =& v(t,x) + \E^{\P^{\a}}\edg{\int_t^{T_n}\beta_{t,r}^{\alpha}\brak{\partial_tv(r,X_r)+ \hat{h}\brak{X_r,v(r,X_r),Dv(r,X_r),D^2v(r,X_r),\alpha}}dr},\\
			&+\E^{\P^{\a}}\edg{\int_t^{T_n}\beta_{t,r}^{\alpha}Dv(r,X_r)\cdot \sigma(X_r)dW^{\alpha}_r}. 
		\end{split}
		\end{align} 
	By definition of the stopping sequence $T_n$, we have that $\E^{\P^{\a}}\edg{\int_t^{T_n}\beta_{t,r}^{\alpha}Dv(r,X_r)\cdot \sigma(X_r)dW^{\alpha}_r}=0$, and since $v$ is a solution to the PDE \eqref{eq : HJB PDE uncontrolled}, then  
	\begin{align}\label{eq : inequality verification 1 bu}
		\E^{\P^{\a}}\edg{\int_t^{T_n}\beta_{t,r}^{\alpha}\brak{\partial_tv(r,X_r)+ \hat{h}\brak{X_r,v(r,X_r),Dv(r,X_r),D^2v(r,X_r),\alpha}}dr}\leq 0,\mbox{ for }\alpha\in\Uc.
	\end{align}
	Therefore 
	\begin{align}\label{eq : inequality verification 2 bu}
	\E^{\P^{\a}}\edg{\beta_{t,T_n}^{\alpha}v(T_n,X_n)} \leq v(t,x)\mbox{ for }\alpha\in \Uc,
	\end{align}
	Remark then that the following almost sure convergence holds as $n\rightarrow +\infty$
	\begin{align}
		\beta_{t,T_n}^{\alpha}v(T_n,X_n)\rightarrow -e^{-\etaa\int_{t}^{T}(s(X_r)-c^A\brak{X_r,\a_r})dr}.
	\end{align}
	Furthermore, we have from the growth of $v$ and the bounds in \eqref{eq : bounds}
	\begin{align}
	\begin{split}
		\E^{\P^{\a}}\edg{\beta_{t,T_n}^{\alpha}v(T_n,X_n)}\leq &\E^{\P^{\a}}\edg{\beta_{t,T_n}^{\alpha}e^{C(T-T_n)}},\\
		=&\E^{\P^{\a}}\edg{{e^{-\etaa\int_{t}^{T_n}(s(X_r)-c^A\brak{X_r,\a_r})dr}e^{C(T-T_n)}}},\\
		\leq& \E^{\P^{\a}}\edg{{e^{\etaa c^A\brak{x_{\infty},\a_{max}}(T_n-t)}e^{C(T-T_n)}}}
		<+\infty\mbox{ for }n\geq 1.
	\end{split}
	\end{align}
	So by dominated convergence, we obtain that
	\begin{align}
		\E^{\P^{\a}}\edg{-e^{-\etaa\int_{t}^{T}(s(X_r)-c^A\brak{X_r,\a_r})dr}} \leq v(t,x)\mbox{ for }\alpha\in \Uc, 
	\end{align}
	and in particular 
	\begin{align}
		\hat{V}_0^A=\sup_{\P^{\a}\in\Pc}\E^{\P^{\a}}\edg{U_A\brak{\int_0^T\brak{s(X_t) - c^A\brak{X_t,\a_t}}dt}}\leq v(0,x_0)
	\end{align}
	So $v(0,x_0)$ is an upper bound for the maximization problem \eqref{eq : uncontrolled producer}.
	Furthermore, for $\brak{\a_{pc}^{\star}\brak{X_t,v(t,X_t),Dv(t,X_t)}}_{t\in[0,T]}$ and going through the same steps as previously, the inequalities \eqref{eq : inequality verification 1 bu} and \eqref{eq : inequality verification 2 bu} become equalities, and so the upper bound is reached for $\a_{pc}$ which is then an optimal control since it is admissible (as a progressively measurable process w.r.t $\F$ valued in $[\alpha_{min},\alpha_{max}]$) which concludes the proof.
	\qed
	\end{proof}
	\no Remark that the PDE \eqref{eq : HJB PDE uncontrolled} can be characterized in terms of certainty equivalent by a straightforward change of variable $\hat{V}^A(t,x)=-e^{-\etaa \hat{u}^A\brak{t,x}}$, which leads to $\hat{u}^A$ being the unique bounded viscosity solution to the PDE
	\begin{align}\label{eq : pde agent solo}
	\begin{cases}
	-\partial_t\hat{u}^A-\hat{h}\brak{x,-\frac{1}{\etaa},D\hat{u}^A,D^2\hat{u}^A - \etaa D\hat{u}^A(D\hat{u}^A)^{\intercal},\a_{pc}^{\star}\brak{x,-\frac{1}{\etaa},D\hat{u}^A}} = 0\mbox{ in }[0,T)\x\R^2_+,\\
	\hat{u}^A(T,.)=0,
	\end{cases}
	\end{align}
	and the optimal feedback control is then written (under further smoothness assumptions) as 
	\begin{align}
		\a_{pc}^{\star}\brak{x,-\frac{1}{\etaa},D\hat{u}^A(t,x)}=\alpha_{min}\lor\frac{\partial_{x^C}\hat{u}^A(t,x) x^C - \kappa_1\underline{x}^C}{\underline{x}^C\kappa_2}\land \alpha_{max}.
	\end{align}
	\subsection{Description of the numerical procedure}\label{Appendix : Numerical method}
		To implement the optimal capacity contract and optimal policy, the first step is to numerically solve the PDE \eqref{eq : pde_u} describing consumer's value function,
		which we recall 
		\begin{align}\tag{\ref{eq : pde_u}}
		\begin{cases}
		-\partial_{t}u-\bar{G}\brak{x,Du,D^{2}u}=0,\mbox{ for }\brak{t,x} \in [0,T)\x\R^2,\\
		u\brak{T,x}=0,\mbox{ for }x \in \R_+^2.
		\end{cases}
		\end{align}	
		We use a finite differences method with an explicit Euler
		scheme to solve this PDE. We discretize the time horizon $[0,T]$ by defining for some $n_T\in\N$ the vector $(t_0,t_2,..t_{n_T})$ with $t_k=k\frac{T}{n_T}$ for $k\in\crl{0,..n}$. We recall then that the state variables are non negative and so we define the boundaries of the space using positive constants $0<x_{c,min}<x_{c,max}$ and $0<x_{d,min}<x_{d,max}$. The space grid is of size $(n_c+1)\x (n_d+1)$ with $n_c, n_d\in\N$, and for $0\leq i\leq n_c$ and $0\leq j\leq n_d$, a function $M:\R_+^2\mapsto \R$ is approximated by a matrix $\M$  such that $\M_{ij}$ represents $M(x_{c,min}+i\frac{x_{c,max}}{n_c},x_{d,min}+j\frac{x_{d,max}}{n_d})$.\\
		We define then a collection $(U^k)_{\crl{0\leq k\leq n_T}}$ of $(n_T+1)$ matrices of size $(n_c+1)\x (n_d+1)$ to stock the (approximation) of the values of $u$. Our goal is to have 
		\begin{align}
			U^k_{ij}=u\brak{k\Delta t,x_{c,min}+i\Delta x_c,x_{d,min}+j\Delta x_d}\mbox{ for }k\in\crl{0,..n_T},i\in\crl{0,..n_c},j\in\crl{0,..n_d}.
		\end{align}
		where we defined $\Delta t:=\frac{T}{n_T}$, $\Delta x_c:=\frac{x_{c,max}}{n_c}$ and $\Delta x_d:=\frac{x_{d,max}}{n_d}$, and recall the following approximations based on taylor expansion
		\begin{align}
		\begin{split}
		\partial_tu(t,x) &= \frac{u(t+\Delta t,x)-u(t,x)}{\Delta t} + o(\Delta t),\\
		&\approx \frac{u(t+\Delta t,x)-u(t,x)}{\Delta t},
		\end{split}
		\end{align} 
		and
		\begin{align}
		\begin{split}
		Du(t,x) &= \brak{\begin{array}{c}
			\frac{u(t,x_c+\Delta x_c,x_d)-u(t,x_c-\Delta x_c,x_d)}{2 \Delta x_c}+o(\Delta x_c) \\
			\frac{u(t,x_c,x_d+\Delta x_d)-u(t,x_c,x_d-\Delta x_d)}{2 \Delta x_d}+o(\Delta x_d)
			\end{array}},\\
				&\approx\brak{\begin{array}{c}
					\frac{u(t,x_c+\Delta x_c,x_d)-u(t,x_c-\Delta x_c,x_d)}{2 \Delta x_c}\\
					\frac{u(t,x_c,x_d+\Delta x_d)-u(t,x_c,x_d-\Delta x_d)}{2 \Delta x_d}
					\end{array}},
		\end{split}
		\end{align}
		and
		\begin{align}
			D^2u(t,x) = \brak{\begin{array}{cc}
				\partial^2_{x^C}u(t,x)&\partial^2_{x^Cx^D}u(t,x)\\
				\partial^2_{x^Cx^D}u(t,x)&\partial^2_{x^D}u(t,x)\end{array}},
		\end{align}
		with 
		\begin{align}
		\begin{cases}
		\partial^2_{x^C}u(t,x)\approx \frac{u(t,x_c+\Delta x_c,x_d)+u(t,x_c-\Delta x_c,x_d)-2u(t,x_c,x_d)}{(\Delta x_c)^2},\\
		\partial^2_{x^D}u(t,x)\approx\frac{u(t,x_c,x_d+\Delta x_d)+u(t,x_c,x_d-\Delta x_d)-2u(t,x_c,x_d)}{(\Delta x_d)^2},\\
		\partial^2_{x^Cx^D}u(t,x)\approx \frac{u(t,x_c+\Delta x_c,x_d+\Delta x_d)+u(t,x_c-\Delta x_c,x_d-\Delta x_d)-u(t,x_c+\Delta x_c,x_d-\Delta x_d)-u(t,x_c-\Delta x_c,x_d+\Delta x_d)}{4(\Delta x_d)(\Delta x_c)}.
		\end{cases}
		\end{align}
		Inspired by the previous expressions, we define for a matrix $\M$ of size $(n_c+1)\x (n_d+1)$ the following matrices also of size $(n_c+1)\x (n_d+1)$ (corresponding to the gradient and the Hessian components)
		\begin{align}
			\mbox{diff}_{x_c}(\M),\mbox{ and }\mbox{diff}_{x_d}(\M)\mbox{ and }\mbox{diff2}_{x_c}(\M)\mbox{ and }\mbox{diff2}_{x_d}(\M)\mbox{ and }\mbox{diff}_{x_cx_d}(\M),
		\end{align}
		such that
		\begin{align}\label{eq : numerical derivatives}
			\begin{cases}
			\mbox{diff}_{x_c}(\M)_{ij} := \frac{\M_{i+1,j}-\M_{i-1,j}}{2 \Delta x_c},\\
			\mbox{diff}_{x_d}(\M)_{ij} := \frac{\M_{i,j+1}-\M_{i,j-1}}{2 \Delta x_d},\\
			\mbox{diff}_{x_cx_d}(\M)_{ij}:=\frac{\M_{i+1,j+1}+\M_{i-1,j-1}-\M_{i+1,j-1}-\M_{i-1,j+1}}{4(\Delta x_d)(\Delta x_c)},\\
			\mbox{diff2}_{x_c}(\M)_{ij}:=\frac{\M_{i+1,j}+\M_{i-1,j}-2\M_{i,j}}{(\Delta x_c)^2},\\
			\mbox{diff2}_{x_d}(\M)_{ij}:=\frac{\M_{i,j+1}+\M_{i,j-1}-2\M_{i,j}}{(\Delta x_d)^2},
			\end{cases}\mbox{ for }i\in\crl{0,..n_c},j\in\crl{0,..n_d},
		\end{align}
		with slight changes near the boundaries to avoid out of range indices, \ie by defining for $i\in\crl{0,..n_c}$ and $j\in\crl{0,..n_d}$
		\begin{align}
			\M_{(-1)j}:=\M_{0j}\mbox{ and }\M_{(n_c+1)j}:=\M_{n_cj}\mbox{ and }\M_{i(-1)}:=\M_{i0}\mbox{ and }\M_{i(n_d+1)}:=\M_{in_d},
		\end{align}
		used in the definition \eqref{eq : numerical derivatives}.
		We define then the equivalent of the hamiltonian of PDE \eqref{eq : pde_u} in the space grid as a matrix valued function $\Gb$ (of size $(n_c+1)\x (n_d+1)$) such that 
		\begin{align}
			\Gb(\M)_{ij}:=\bar{G}\brak{\brak{\begin{array}{c}
					x_{c,min}+i\Delta x_c\\
					x_{d,min}+j\Delta x_d
					\end{array}},\brak{\begin{array}{c}
					\mbox{diff}_{x_c}(\M)_{ij}\\
					\mbox{diff}_{x_d}(\M)_{ij}
					\end{array}},\brak{\begin{array}{cc}
					\mbox{diff2}_{x_c}(\M)_{ij}&\mbox{diff}_{x_cx_d}(\M)_{ij}\\
					\mbox{diff}_{x_cx_d}(\M)_{ij}&\mbox{diff2}_{x_d}(\M)_{ij}\end{array}}},\mbox{ for }0\leq i\leq n_c\mbox{ and }0\leq j\leq n_d.
		\end{align}
		The numerical approximation of \eqref{eq : pde_u} becomes straightforward, and consists in initializing the terminal value $U_{ij}^{n_T}=0$ for $0\leq i\leq n_c$ and $0\leq j\leq n_d$. Then computing by backward induction
		\begin{align}
			U^{k-1} := U^k+\Delta t\x \Gb(U^k),\mbox{ for }k\in\crl{n_T,n_T-1,..,1}.
		\end{align}
		Recall then the definition of the optimal feedback control
		\begin{align}\tag{\ref{eq : hat_z}}
		Z_t^{\star}=\brak{\begin{array}{c}
				Z_t^{C,\star}\\
				Z_t^{D,\star}
				\end{array}}:=\hat{z}\brak{X_t,Du\brak{t,X_t}}\mbox{ for }t\in[0,T].
		\end{align}
		We approximate then $\brak{Z_t^{C,\star}}_{t\in[0,T]}$ and $\brak{Z_t^{D,\star}}_{t\in[0,T]}$ with the collection of matrices $(Z_c^k)_{\crl{0\leq k\leq n_T}}$ and $(Z_d^k)_{\crl{0\leq k\leq n_T}}$ of sizes $(n_c+1)\x (n_d+1)$ defined as
		\begin{align}
		\brak{\begin{array}{c}
			\brak{Z_c}_{ij}^k\\
			\brak{Z_d}_{ij}^k
			\end{array}}:=\hat{z}\brak{\brak{\begin{array}{c}
				x_{c,min}+i\Delta x_c\\
				x_{d,min}+j\Delta x_d
				\end{array}},\brak{\begin{array}{c}
				\mbox{diff}_{x_c}(U^k)_{ij}\\
				\mbox{diff}_{x_d}(U^k)_{ij}
				\end{array}}}\mbox{ for }k\in\crl{0,..n_T},i\in\crl{0,..n_c},j\in\crl{0,..n_d}.
		\end{align}
		Similarly, the recommended effort $\brak{\hat{\a}(X^C_t,Z^{\star,C}_t)}_{t\in[0,T]}$ is approximated with matrices $(\hat{A}^k)_{\crl{0\leq k\leq n_T}}$ of sizes $(n_c+1)\x (n_d+1)$ defined as
		\begin{align}
			\hat{A}_{ij}^k:=\hat{\a}(x_{c,min}+i\Delta x_c,\brak{Z_c}_{ij}^k)\mbox{ for }k\in\crl{0,..n_T},i\in\crl{0,..n_c},j\in\crl{0,..n_d}.
		\end{align}
		\begin{Remark}
			{\upshape For a fixed $k\in\crl{0,..n_T}$, the matrices $U^k$, $(Z_c)^k$, $(Z_d)^k$ and $\hat{A}^k$ represent approximations on the whole space $\R^2_+$. So for a given $x=(x^C,x^D)\in\R^2_+$, one needs first to find the indices $(i,j)$ such that $(x_{c,min}+i\Delta x_c,x_{d,min}+j\Delta x_d)\approx x$; for example by taking $(i,j)=\brak{\left \lfloor{\frac{x^C -x_{c,min}}{\Delta x_c}}\right \rfloor ,\left \lfloor{\frac{x^D -x_{d,min}}{\Delta x_d}}\right \rfloor}$ to estimate the required quantities, for example $u(k\Delta t,x)\approx U_{ij}^k$.}
		\end{Remark}
		\no The second step in the simulation is the forward diffusion of state variables using the optimal controls	for a number of scenarios $N$, to use the Monte-carlo for the estimations. We recall the dynamics of the controlled state variables
		\begin{align}
		\begin{cases}
		X_{t}&=\brak{\begin{array}{c}
			x_0^C\\
			x_0^D
			\end{array}} + \int_{0}^{t}\mu\brak{X_r,\hat{\a}\brak{X^C_r,Z^{C,\star}_r}}dr+\int_{0}^{t}\sigma\brak{X_r}dW^{\hat{\a}}_{r},\\
		Y_{t}&=\Rc + \int_{0}^{t}\brak{c^{A}\brak{X_r,\hat{\a}\brak{X^C_r,Z^{C,\star}_r}}+\frac{\eta_{A}}{2}|\sigma\brak{X_r} Z^{\star}_r|^2-s(X_r)}dr+\int_{0}^{t}Z^{\star}_r\cdot\sigma\brak{X_r}dW^{\hat{\a}}_r.
		\end{cases}
		\end{align}
		This is also done through an explicit Euler scheme, which we provide for a single scenario. We define the matrix $\hat{X}$ of size $2\x(n_T+1)$ and the vector $\hat{Y}$ of size $1\x(n_T+1)$ which provides $\hat{X}^k\approx X_{k\Delta t}$ and $\hat{Y}^k\approx Y_{k\Delta t}$. We start by initializing $\hat{X}^0:=\brak{\begin{array}{c}
			x_0^C\\
			x_0^D
			\end{array}}$ and $\hat{Y}^0:=\Rc$. Then we compute by induction for $k\in\crl{1,..n_T}$:	
		\begin{align}
		\begin{cases}
		(i,j):=\brak{\left \lfloor{\frac{(x_{c,min}\lor\hat{X}_1^{k-1}\land x_{c,max}) -x_{c,min}}{\Delta x_c}}\right \rfloor ,\left \lfloor{\frac{(x_{d,min}\lor\hat{X}_2^{k-1}\land x_{d,max})-x_{d,min}}{\Delta x_d}}\right \rfloor},\\
		\mbox{Generate a (normalized) 2 dimensional Brownian increment $\brak{\Delta W}^k$ with a law $\Nc\brak{0_2,I_2}$},\\
		\hat{X}^{k}:=\hat{X}^{k-1}+\Delta t \mu\brak{\hat{X}^{k-1},\hat{A}_{ij}^{k-1}}+\sqrt{\Delta t}\sigma\brak{\hat{X}^{k-1}}\brak{\Delta W}^k,\\
		\hat{Y}^k:=\hat{Y}^{k-1}+\Delta t\brak{c^{A}\brak{\hat{X}^{k-1},\hat{A}_{ij}^{k-1}}+\frac{\eta_{A}}{2}\abs{\sigma\brak{\hat{X}^{k-1}} \brak{\begin{array}{c}
				\brak{Z_c}_{ij}^{k-1}\\
				\brak{Z_d}_{ij}^{k-1}
				\end{array}}}^2-s(\hat{X}^{k-1})}+\sqrt{\Delta t}\brak{\begin{array}{c}
			\brak{Z_c}_{ij}^{k-1}\\
			\brak{Z_d}_{ij}^{k-1}
			\end{array}}\cdot \sigma\brak{\hat{X}^{k-1}}\brak{\Delta W}^k.
		\end{cases}
		\end{align}
	Remark that it is fundamental that for each time step $k$, the Brownian increment used to compute $\hat{X}^{k}$ is the same as the one to compute $\hat{Y}^k$.\\ 	
	The resulting matrices correspond to a realization for a scenario of the capacity and demand ($\hat{X}^k_1$ and $\hat{X}^k_2$) for $k\in\crl{0,..n_T}$, and $\hat{Y}^{n_T}$ the the contract for this scenario. Prior values of $\hat{Y}$, \ie $\hat{Y}^k$ for $k<n_T$ could be used to understand the composition of the contract, but are only informative and less important than $\hat{Y}^{n_T}$ which represents the actual amount paid, since the contracting is in a lump-sum payment form.

\newpage
		
\bibliography{bib_HF}	
\end{document}